\PassOptionsToPackage{table,dvipsnames,svgnames}{xcolor}
\documentclass[acmsmall]{acmart}

\AtBeginDocument{%
  }

\setcopyright{acmlicensed}
\acmJournal{PACMMOD}
\acmYear{2025} \acmVolume{3} \acmNumber{3 (SIGMOD)}
\acmArticle{185} \acmMonth{6} \acmPrice{}
\acmDOI{10.1145/3725322}

\usepackage{amsmath}
\usepackage{amsthm}
\usepackage{amssymb}
\usepackage{mathtools}
\usepackage{bm}
\usepackage{bbm}
\usepackage{booktabs}
\usepackage{xspace}
\usepackage{color,soul}
\usepackage{multirow}
\usepackage{multicol}
\usepackage{caption}
\usepackage{enumitem}

\usepackage{microtype}
\usepackage{graphicx}
\usepackage{subfigure}
\usepackage{booktabs}
\usepackage{hyperref}
\usepackage{nicefrac}
\usepackage[normalem]{ulem}
\usepackage{tikz}
\usetikzlibrary{trees}
\usepackage{float}
\usepackage{subfigure}
\usepackage[ruled,noend,linesnumbered,vlined]{algorithm2e}
\usepackage{colortbl}
\usepackage{hhline}

\usepackage{listings}
\newcommand{\specialcell}[2][c]{%
	\begin{tabular}[#1]{@{}c@{}}#2\end{tabular}}
\newcommand{\mytextcircled}[1]{\textcircled{\raisebox{-0.8pt}{#1}}}

\newcommand{\blue}[1]{\textcolor{blue}{#1}}

\definecolor{myred}{rgb}{1.0,0.7,0.8}
\definecolor{mygreen}{RGB}{0,166,0}
\definecolor{lightgreen}{rgb}{0.56, 0.93, 0.56}
\definecolor{myorange}{RGB}{252,107,4}
\definecolor{darkgreen}{RGB}{0,153,102}
\definecolor{lightblue}{rgb}{0.53, 0.81, 0.92}
\definecolor{lightgray}{gray}{0.9}

\newcommand{\mysubsubsection}[1]{{\vspace{0.25em}\noindent\ul{\textbf{#1.\xspace}}}}

\DeclareMathOperator*{\argmin}{arg\,min}

\newcommand{\conj}[1]{\overline{#1}}

\newcommand{\system}{{\textsc{Malleus}}\xspace}

\newcommand{\dpd}{\conj{DP}}
\newcommand{\ppd}{\conj{PP}}
\newcommand{\ppdi}{\conj{PP}_{i}}
\newcommand{\tpd}{\conj{TP}}

\newtheorem{theorem}{Theorem}

\newtheorem{proposition}{Proposition}

\newtheorem*{assumption*}{\assumptionnumber}
\providecommand{\assumptionnumber}{}


\newtheorem*{statement*}{\statementnumber}
\providecommand{\statementnumber}{}
\makeatletter
\newenvironment{statement}[1]
 {%
  \renewcommand{\statementnumber}{Statement #1}%
  \begin{statement*}%
  \protected@edef\@currentlabel{#1}%
 }
 {%
  \end{statement*}
 }
\makeatother

\begin{document}

\title{\system: Straggler-Resilient Hybrid Parallel Training of Large-scale Models via Malleable Data and Model Parallelization}
\renewcommand{\shorttitle}{\system}

\author{Haoyang Li}
\authornote{Equal contribution}
\authornote{School of Computer Science \& Key Lab of High Confidence Software Technologies (MOE), Peking University}
\email{lihaoyang@stu.pku.edu.cn}
\affiliation{
\institution{Peking University}
\country{China}
}

\author{Fangcheng Fu}
\authornotemark[1]
\authornotemark[2]
\email{ccchengff@pku.edu.cn}
\affiliation{
\institution{Peking University}
\country{China}
}

\author{Hao Ge}
\authornotemark[2]
\email{gehao@stu.pku.edu.cn}
\affiliation{
\institution{Peking University}
\country{China}
}

\author{Sheng Lin}
\authornotemark[2]
\email{linsh@stu.pku.edu.cn}
\affiliation{
\institution{Peking University}
\country{China}
}

\author{Xuanyu Wang}
\authornotemark[2]
\email{wxyz0001@pku.edu.cn}
\affiliation{
\institution{Peking University}
\country{China}
}

\author{Jiawen Niu}
\authornotemark[2]
\email{niujiawen705@stu.pku.edu.cn}
\affiliation{
\institution{Peking University}
\country{China}
}

\author{Yujie Wang}
\authornotemark[2]
\email{alfredwang@pku.edu.cn}
\affiliation{
\institution{Peking University}
\country{China}
}

\author{Hailin Zhang}
\authornotemark[2]
\email{z.hl@pku.edu.cn}
\affiliation{
\institution{Peking University}
\country{China}
}

\author{Xiaonan Nie}
\authornotemark[2]
\email{xiaonan.nie@pku.edu.cn}
\affiliation{
\institution{Peking University}
\country{China}
}

\author{Bin Cui}
\authornotemark[2]
\authornote{Institute of Computational Social Science, Peking University (Qingdao)}
\email{bin.cui@pku.edu.cn}
\affiliation{
\institution{Peking University}
\country{China}
}

\renewcommand{\shortauthors}{Haoyang Li et al.}

\begin{abstract}
As the scale of models and training data continues to grow, there is an expanding reliance on more GPUs to train large-scale models, which inevitably increases the likelihood of encountering dynamic stragglers that some devices lag behind in performance occasionally. However, hybrid parallel training, one of the de facto paradigms to train large models, is typically sensitive to the stragglers. 

This paper presents \system, a straggler-resilient hybrid parallel training framework for large-scale models. 
\system quantifies the stragglers at the nuanced, per-GPU granularity during training, and develops a novel planning algorithm to deduce the optimal parallelization of GPU devices, pipeline stages, model layers, and training data, maximizing training efficiency when stragglers exist. 
In addition, once a shift in the straggler situation is detected, \system adaptively adjusts the parallelization via a re-planning process, and seamlessly and efficiently migrates the model states on the fly, without sacrificing the stability of the training tasks.
Empirical results on large language models with up to 110B parameters show that \system consistently outperforms existing parallel training frameworks under various straggler situations, delivering on average 2.63-5.28$\times$ of efficiency improvement. 
Source code is available at \url{https://github.com/PKU-DAIR/Hetu}.
\end{abstract}




% %%
% %% Keywords. The author(s) should pick words that accurately describe
% %% the work being presented. Separate the keywords with commas.

\received{October 2024}
\received[revised]{January 2025}
\received[accepted]{February 2025}

% \settopmatter{printacmref=true}
\maketitle

\section{Introduction}
\label{sec:intro}

\mysubsubsection{Background}
Recent years have witnessed the remarkable advancements of machine learning (ML) models in various data processing and data management applications~\cite{unlearn_db,NLP_SQL_EA,archetype,gptuner,dbgpt,neurdb}. 
In turns, the data management community has made essential efforts in supporting ML pipelines, including pre-ML data preprocessing~\cite{saga_data_cleaning,goldminer,dl_preprocess_pipeline,tf_data,data_cleaning_meets_ai}, in-ML model training~\cite{BladeDISC,angelptm,dl_on_ds,dl_across_clouds_ea}, and post-ML model deployment~\cite{dm4ml_survey,flashllm,Biathlon,SmartLite}. 
In particular, how to design systems for the efficient parallel training of ML models is of paramount significance to our community and has attracted great research interests.
To begin with, given the increasing volume of training data, numerous efforts have been made to optimize the parallelization and layout of training data~\cite{pytorch_ddp,pytorch_fsdp,DimmWitted,SketchML}. 
Subsequently, in addition to the training data, there is a trend to expand the scale of models themselves to strengthen their powers~\cite{scalinglaws,attn,gpt3,llama3}. 
This has made the management of model data an ever increasingly important topic~\cite{data_and_ai_model_markets,dm4llm,ai_meets_db}, and how to parallelize large-scale models over massive GPU devices has drawn extensive attention in the data management community~\cite{galvatron,MiCS,FlexMoE,saturn}. 

Typically, hybrid parallel~\cite{megatron_2} is known as the combination of data parallel and model parallel, and has become a pivotal foundation for training large-scale models over multiple GPU devices. 
When training with hybrid parallel, synchronization is expected among these GPUs. For instance, model parallel techniques like tensor parallel and pipeline parallel necessitate exchanging the intermediate results~\cite{gpipe,pipedream_flush,megatron_1}, whilst data parallel techniques require synchronizing the model replicas~\cite{pytorch_ddp,horovod,ps_limu}. 
Undoubtedly, it prefers a homogeneous environment to avoid idle waiting. 
Nevertheless, the hardware environment would be far from ideally homogeneous in the training of large-scale models --- involving more GPUs in the training inevitably raises the probability of encountering \textit{dynamic stragglers}. 
To be specific, several studies~\cite{megascale,imbue_report} have observed that a few GPUs occasionally become slower than expected, resulting in stragglers. 
Although such stragglers are relatively rare in terms of their population within a cluster, it is timely and important to address the dynamic straggler problem.

For one thing, despite the small population, the occurrences of stragglers are complex and unforeseeable, with a diverse range of manifestation patterns, such as device auto-throttling, heat issues, network jitter, and unknown sharing~\cite{imbue_report}. 
In fact, as the scale of GPU clusters increases, the dynamic straggler problem has become increasingly important. 
For instance, NVIDIA has also released an official library that facilitates straggler detection~\cite{nvidia_resiliency_ext}. 
However, given the diverse straggler situations, there lacks a concrete mathematical framework to analyze the impact on training efficiency, impeding the development of straggler-resilient solutions for large-scale model training.

For another, even a small number of stragglers could severely damage the training efficiency. 
As we will empirically show in \S\ref{sec:expr}, even in the mildest straggler situation (only 1 straggling GPU), the training time of existing frameworks slows down by around 2$\times$. 
The fundamental reason is that distributed training necessitates synchronization among GPUs, whilst existing works mainly adopt a uniform partitioning of model and data to evenly distribute the training workloads, ignoring the discrepancy in hardware capabilities when straggler(s) exist. 
As a result, there is a need to redesign the parallelization plan to address the stragglers.

\mysubsubsection{The Current Landscape}
To mitigate the dynamic straggler problem, there have been many efforts and they can be categorized into two lines. 
The first line of efforts~\cite{flexrr,ssp_ps,hetero_ps,dynamic_ssp,partial_reduce} relaxes the synchronization protocol of data parallel, allowing the non-straggling devices to update models earlier than the straggling ones. 
The second line~\cite{torch_elastic,horovod_elastic,dynamic_minibatch_sgd,resource_elasticity} elastically changes the number of devices involved in training to remove the stragglers, and dynamically adjusts the global batch size (i.e., the number of training data per step). 
However, these approaches are developed for data parallel, and address the dynamic straggler problem at the granularity of model replicas. Whilst in hybrid parallel, each model replica is served by multiple devices, making existing works ineffective (detailed in \S\ref{sec:pre_limit_analysis}). 
Worse still, impacts on the model convergence are inevitable with these approaches. Since large-scale models require an extremely long time to train, it is impractical to run numerous trials to tune hyper-parameters for these lossy approaches. 

Beyond data parallel, straggler-resilient hybrid parallel for large models is under-explored yet. In practice, it demands engineers to discover the stragglers during training and manually replace them with backup devices~\cite{megascale}. However, let alone the substantial time cost in problem shooting and restarting, this works only when there are spare nodes in the cluster, which is unrealistic in general due to the GPU shortage problem~\cite{skypilot,gpu_shortage_cross_region_cloud}. 
When there are no spare nodes, it requires extra expert experiences to manually adjust the parallel configuration. 
Besides, such an approach needs to remove the entire node (machine), yet existing studies have found that the straggler problem usually appears at the GPU granularity rather than the node granularity~\cite{megascale,imbue_report}. 
Thus, simply removing the entire node would lead to a waste of computing resources when there are non-straggling GPUs on it, calling for a more nuanced solution.

In short, to fill the gap between straggler-resilience and hybrid parallel training, there are three key challenges as follows.
\begin{itemize}[leftmargin=*]
\item%[\mytextcircled{1}]
\textbf{Data-driven Device Management}: 
Existing parallel training systems usually organize the devices into equal-sized meshes and distribute the training workloads evenly among each mesh. 
However, when there are stragglers, devices within the same mesh must wait in idle to synchronize with the slowest straggler(s).
Therefore, given the nuanced granularity of the straggler problem, it is essential to quantify the hardware capability of each device in real time to model the straggler situation. 
Based on which, a data-driven approach is expected to organize the devices so that the fast ones would not be affected severely. 
Unfortunately, this is unexplored so far. 
\item%[\mytextcircled{2}]
\textbf{Model Data and Training Data Parallelization}: 
It is an extremely complicated problem to determine the optimal model and data parallelization configuration when devices have diverse capabilities. 
For one thing, it is common to break down the model into small parts (e.g., pipeline stages) and assign them to different devices. However, it is non-trivial to balance the running time of each part by matching the workloads of different parts and the divergent hardware capabilities. 
For another, unlike existing works that evenly distribute the training data among all data parallel groups (i.e., model replicas), we must meticulously adjust the training data assignment to balance the running time among all model replicas.

\item%[\mytextcircled{3}]
\textbf{Model Data Migration}: 
Due to the dynamicity of the straggler problem, the optimal configuration changes over time. 
An intuitive approach is to restart the training task with a new parallelization configuration to maintain high performance. 
However, since the model size is substantial, restarting is extremely expensive as we need to save the model checkpoint before each restart and re-load it later. 
Thus, there expresses a need for migrating the model data on the fly without interrupting the training task. 
\end{itemize}

\mysubsubsection{Our Solution}
To address these challenges, this work introduces \system\footnote{In our work, the parallelization is designed to be malleable in order to accommodate the dynamic straggler problem. Thus, we name our framework ``Malleus'', which is the Latin etymological origin of the word ``Malleable''.}, a straggler-resilient hybrid parallel training framework for large-scale models. 
\system pinpoints the dynamic stragglers at the nuanced, per-GPU granularity by quantifying the \textit{straggling rate} of each device, which denotes how much a straggler is slower than a normal one. 
Based on this, \system develops a novel \textit{parallelization planning algorithm} to determine the optimal parallelization plan that maximizes training efficiency. It takes the real-time straggler situation into account and strikes a good balance among the training workloads on different GPUs through a series of non-uniform partitioning w.r.t. GPU devices, pipeline stages, model layers, and training data. 
Lastly, \system achieves \textit{malleable parallelization} by promptly adjusting the parallelization on the fly to accommodate the dynamicity of straggler situations.

Specifically, given the GPUs of diverse straggling rates and the training task (the model and training data), we formulate a bi-level hierarchical optimization problem for the deduction of a parallelization plan that maximizes training efficiency. 

To begin with, we aim to address the first challenge in the upper-level problem, whose objective is to discover how to partition the GPUs into tensor parallel groups and how to orchestrate multiple training pipelines with these groups, with each group serving as one stage of a training pipeline. 
We tackle this problem with two major efforts. 
First, we establish theorems to analyze how to partition the available GPUs into groups with performance guarantees, based on which a partition-then-split approach is devised. Second, given these groups, we formulate a mixed-integer non-linear programming (MINLP) problem in response to the varying efficiencies among groups. By solving the problem, we achieve the optimal orchestration of training pipelines. 

Given any possible solution to the upper-level problem, the lower-level problem aims to tackle the second challenge by determining the best parallelization of model layers within each pipeline as well as the best assignment of training data across the pipelines. 
Particularly, we reformulate it as a joint-optimization problem of model layer and training data assignments, which aims to minimize the training time by balancing the workloads across the stages within each pipeline and across the pipelines simultaneously. Then, we decouple the joint problem into multiple sub-problems in the form of integer linear programming (ILP), and solve them to obtain the best plan. 

To address the third challenge and facilitate the adaptation to the dynamic stragglers, we introduce an efficient re-planning process that adjusts the training task to accommodate the straggler situations in real time. 
For one thing, we design an asynchronous re-planning mechanism, which derives the optimal parallelization based on the immediate straggler situation without halting the training task. 
For another, according to the new plan, \system automatically migrates the model partitioning and device mapping on the fly, enhancing the stability and straggler-resilience in the training of large-scale models. 

The contributions of this work are summarized as follows:
\begin{itemize}[leftmargin=*]
\item We develop \system, a brand new straggler-resilient hybrid parallel training framework for large-scale models. 
\item \system quantifies the straggler problems into straggling rates at the nuanced, per-GPU granularity, and introduces a novel parallelization planning algorithm to automatically deduce the optimal model and data parallelization that maximizes training efficiency given the straggling rates. 
\item In response to the dynamic changes in straggler situations, \system supports adjusting the parallelization plan on the fly, eliminating the need of re-starting the training task. 
\item We conduct experiments using LLMs with up to 110B parameters. The results in various straggler situations show that \system outperforms existing parallel training frameworks in terms of training speed by up to 6.73$\times$ and 2.63-5.28$\times$ on average. Furthermore, \system accommodates dynamic changes in straggler situations and consistently achieves at least 90\% of the theoretic optimal performance. 
\end{itemize}

\section{Preliminaries}
\label{sec:pre}

\subsection{Parallelization in Model Training}
\label{sec:pre_parallel}

With the explosive growth in model sizes and training data, parallelization of data and models has become an essential bedrock to train large-scale models. 

\mysubsubsection{Data parallel}
\textit{Data parallel (DP)}~\cite{pytorch_ddp,horovod,ps_limu,SketchML,bagua,angelptm,tinyscript} scatters the training data to multiple devices, yet requires each device to hold one model replica. In each training step, the devices execute forward and backward propagation individually to compute the model gradients, which are then synchronized through communication (e.g., \texttt{all-reduce}) to ensure the model replicas on all devices are updated in the same way. 

\mysubsubsection{Model parallel}
Model parallel splits the model states across multiple devices to support large-scale models. Notably, \textit{tensor parallel (TP)}~\cite{megatron_1} and \textit{pipeline parallel (PP)}~\cite{gpipe,pipedream,pipedream_flush,pipeline_survey} are two well-known categories of model parallel.

\begin{figure}[!t]
\centering
\includegraphics{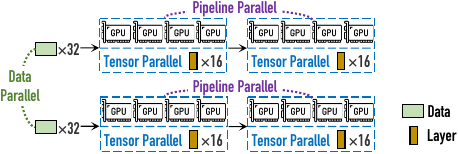}
\caption{\small{An example of 3-Dimensional (3D) parallel encompassing data parallel (DP), tensor parallel (TP), and pipeline parallel (PP)}. The model consists of 32 layers and the global batch size is 64.}
\label{fig:3d_parallel}
\end{figure}

\begin{figure}[!t]
\centering
\includegraphics{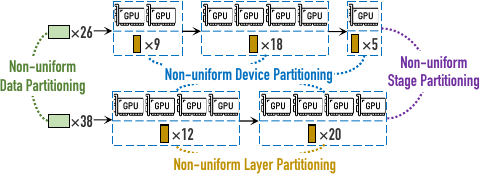}
\caption{\small{Key characteristics of the parallelization plans in our work, illustrated with the example workload in Figure~\ref{fig:3d_parallel}.}}
\label{fig:para_plan_design}
\end{figure}

TP splits the model parameters of computationally intensive operations (e.g., matrix multiplication) across devices. However, TP necessitates network communications to exchange the intermediate results in both forward and backward propagation, implying a need for high communication bandwidth among the devices. Thus, TP is typically applied on GPUs within the same node (machine), because intra-node connections have higher communication bandwidth than inter-node connections in most cases. 

PP treats a model as a sequence of layers, and partitions the layers into multiple stages. These stages are distributed across devices to form a pipeline, and peer-to-peer communication is leveraged to transmit the intermediate results between consecutive stages. Since only the activations across stages need to be transmitted, PP usually entails a much lower communication volume than TP, and is able to accommodate both intra- and inter-node network connections.

\mysubsubsection{Hybrid parallel}
Hybrid parallel~\cite{megatron_2,flexflow_soap,alpa,osdp,galvatron,galvatron_bmw,piper,MiCS,flexsp} incorporates the strengths and weaknesses of different parallelization. Particularly, one of the most widely used hybrid parallel approach is the 3-Dimensional (3D) parallel approaches in Megatron-LM~\cite{megatron_2,megatron_3}, which is illustrated in Figure~\ref{fig:3d_parallel}. Such a 3D parallel approach encompasses DP, TP, and PP, and has become fundamental for training large-scale models.

\subsection{Straggler-resilient Parallel Training}
\label{sec:straggler_resilient}

Since parallel training needs to synchronize model gradients and/or intermediate results to ensure correct convergence, when stragglers exist, most GPUs must sit idle until the slowest stragglers have caught up. Obviously, it results in a huge waste of resources and poor efficiency. 
There have been numerous approaches to mitigate the straggler problem. We classify them into two lines. 

The first line of approaches~\cite{flexrr,ssp_ps,hetero_ps,dynamic_ssp,partial_reduce} is developed based on the idea of \textit{stale synchronous parallel (SSP)}, which relaxes the synchronization protocol and allows updating different model replicas asynchronously. Specifically, fast DP groups can synchronize model gradients and update models without waiting for the slow DP groups. There is usually a threshold that constrains the differences in terms of the training steps between the fastest and slowest groups. By doing so, the SSP-based approaches reduce the idle periods of the non-straggling devices to improve efficiency. 

The second line~\cite{torch_elastic,horovod_elastic,dynamic_minibatch_sgd,resource_elasticity} addresses the straggler problem by dynamically adjusting the number of devices during training. Particularly, once any stragglers are detected, these approaches remove the related devices from the training task, and (optionally) try to request new devices from the cluster. Consequently, all devices involved in the training task will not be slowed down by stragglers. In essence, these approaches manage the devices at the granularity of model replicas, which is beneficial to ensure aligned performance among DP groups. When the number of DP groups changes, the global batch size is also adjusted proportionally.

\subsection{Limitation Analysis}
\label{sec:pre_limit_analysis}

Despite the numerous efforts of straggler mitigation, we find that they are developed for DP, but fall short in the scenario of hybrid parallel training of large-scale models. 

First, existing efforts solely aim at DP, whilst model parallel, which is indispensable due to the tremendous scale of models, is not their primary focus. In hybrid parallel, each model replica is served by a training pipeline consisting of multiple GPUs, making existing efforts ineffective. For one thing, SSP-based approaches cannot prevent the non-straggling GPUs from being affected by the stragglers within the same pipeline. For another, for approaches based on straggler removal, removing one straggler would make an entire pipeline inexecutable, so they must keep the other GPUs in the pipeline idle until a new device is ready to join. 

Second, existing efforts inevitably impact model convergence of the training task. For instance, SSP-based approaches are proven to slow down the convergence rate and may even lead to divergence due to the staleness in model gradients~\cite{ssp_ps}. 
For approaches that based on straggler removal, they also need to adapt the batch size to the DP degree, which influences model convergence. 
Since the training of large-scale models is time-consuming (e.g., it takes weeks or even months to train an LLM), it is impractical to run numerous trials to tune the hyper-parameters. 

To tackle these limitations, this work proposes to solve the straggler problem with a nuanced, per-GPU granularity, rather than the coarse, per-replica or per-node granularity. Besides, regardless of the straggler situations, our work does not adjust the global batch size and the synchronization protocol across pipelines to be lossless.

\section{System Overview}
\label{sec:system}

In this section, we present an overview of our straggler-resilient hybrid parallel training framework, namely \system.

\subsection{Design of Parallelization Plans}
\label{sec:system_plan_design}

Figure~\ref{fig:para_plan_design} illustrates the key characteristics of the hybrid parallel training in \system. To harness the stragglers, \system integrates non-uniform partitioning of GPU devices, pipeline stages, model layers, and training data, as discussed below.
\begin{itemize}[leftmargin=*]
\item[\mytextcircled{1}] 
Since GPUs in a TP group should synchronize frequently, in response to the efficiency variation among GPUs, \system supports \ul{\textit{non-uniform device partitioning}} by strategically allowing the numbers of GPUs in TP groups to be different. 
\item[\mytextcircled{2}] 
In hybrid parallel, each TP group serves as one unit that executes a pipeline stage. To accommodate the efficiency variation among TP groups, \system enables \ul{\textit{non-uniform stage partitioning}}, which allows the pipelines to have varying numbers of stages (i.e., TP groups). 
\item[\mytextcircled{3}] 
Within each pipeline, the TP groups would differ in efficiency. To address this issue, \system uses \ul{\textit{non-uniform layer partitioning}}, which supports appointing different numbers of model layers to the stages within each pipeline. 
\item[\mytextcircled{4}] 
Since the efficiencies of different pipelines would also be unaligned, \system facilitates \ul{\textit{non-uniform data partitioning}} that allocates different volumes of training data to the pipelines to balance their training time. 
\end{itemize}
In short, a solution of such non-uniform partitioning is determined to accommodate the stragglers. We call the solution a \ul{\textbf{\textit{parallelization plan}}}. And the goal of \system is to adaptively adjust the parallelization plan against the dynamic straggler situations to maximize the training efficiency.

\begin{figure}[!t]
\centering
\includegraphics{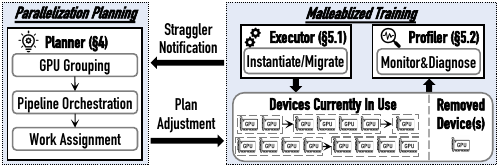}
\caption{\small{Architecture overview of \system.}}
\label{fig:system_overview}
\end{figure}

\subsection{System Components}
\label{sec:system_components}

The architecture overview of \system is shown in Figure~\ref{fig:system_overview}. 
There are three major components, as introduced below.

\mysubsubsection{Planner}
The planner is responsible for analyzing and deducing the most suitable parallelization plans. 
It takes as input the task description (e.g., model architecture, mini-batch size, etc.) provided by the user and the profiled information collected by the profiler. 
Then, it deduces the parallelization plan aiming to adapt the training plan to the stragglers in real time. The deduction is done based on a brand new planning algorithm, which will be detailed in \S\ref{sec:planner}.

\mysubsubsection{Executor}
The executor is in charge of the training. Whenever the planner makes a new decision about the parallelization plan, the executor triggers a migration process that adjusts the model partitioning and device mapping. Besides, to realize efficient training of our parallelization plans, the executor manages the model sharding and gradient synchronization across the pipelines in a non-uniform manner. 
More details will be elaborated in \S\ref{sec:malleable_training_model_management}.

\mysubsubsection{Profiler}
The profiler monitors the real-time hardware efficiency of each device and provides the GPU straggling rates on the fly. 
To be specific, it measures the running time of each GPU, identifies the stragglers, and calculates their straggling rates by comparing to the non-stragglers.
In addition to recording the running time on devices that are currently in use, it also examines devices that were previously removed due to high straggling rates since these devices may be useful later (detailed in \S\ref{sec:malleable_training_device_management}). 
During training, if any straggling rate gives a change greater than 5\% between two consecutive iterations, the profiler will consider it as an obvious shift in the straggling situation and notify the planner immediately.

\mysubsubsection{Overall Routine}
The overall routine of \system is as follows.
\mytextcircled{1} The training process begins with an initial parallelization plan, which can be deduced by the planner (with all straggling rates being 1) or manually configured by the user.
\mytextcircled{2} The executor instantiates the model states over the GPUs and carries out the training. 
\mytextcircled{3} Meanwhile, the profiler consistently tracks the performance and examines the GPU straggling rates. 
\mytextcircled{4} 
Once any straggling rate changes over 5\% between two consecutive iterations, the planner is informed and then triggers a re-planning process to adjust the parallelization plan to adapt to the dynamic stragglers.

\begin{table}[!t]
\small
\centering
\caption{\small{Notations used in this work. DP, TP, and PP refer to data parallel, tensor parallel, and pipeline parallel, respectively.}}
\label{tb:notatiions}
\begin{tabular}{ll}
\hline
\toprule
$\dpd$ & DP degree (the number of pipelines) 
\\
$\ppdi$ & PP degree (the number of stages) of the $i$-th pipeline
\\
$x$ & The GPU straggling rate of a specific GPU, which denotes its slowdown rate \\
& compared to a normal GPU (higher indicates slower and $x=1$ indicates the GPU is not a straggler)
\\
$y$ & The group straggling rate of a specific TP group
\\
$L$ & The number of layers in the model
\\
$l_{i,j}$ & The number of model layers assigned to the $j$-th stage in the $i$-th pipeline
\\
$B$ & The global batch size (number of training data per step) 
\\
$b$ & The micro-batch size 
\\
$m_i$ & The number of micro-batches assigned to the $i$-th pipeline
\\
$\mathbb{N}_0$ & The set of non-negative integers, i.e., $\{0, 1, 2, \cdots\}$
\\
\bottomrule
\end{tabular}
\end{table}

\begin{figure*}[!t]
\centering
\includegraphics[width=\textwidth]{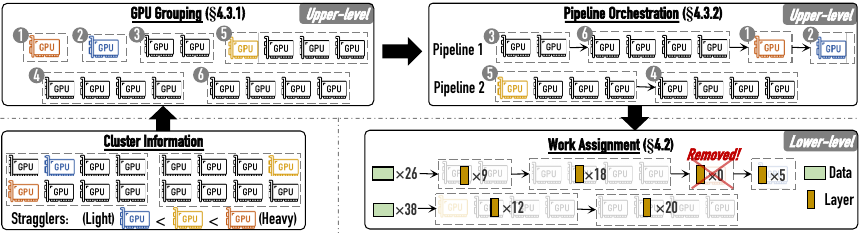}
\caption{\small{Overview of the planning routine, illustrated with an example where 3 out of 16 GPUs are stragglers (depicted in the lower left). The model consists of 32 layers and the global batch size is 64. The planning routine partitions the GPUs into 6 groups (upper left), organizes them into 2 pipelines (upper right), and assigns the model layers across different pipeline stages as well as assigns the training data across the pipelines (lower right).}}
\label{fig:planner_routine}
\end{figure*}

\section{Parallelization Planning}
\label{sec:planner}

In this section, we introduce our planning algorithm, which aims to optimize the overall training efficiency.
Frequently used notations are listed in Table~\ref{tb:notatiions}.

\subsection{Problem Formulation}
\label{sec:planner_formulation}

In a nutshell, a parallelization plan consists of four components. 
\mytextcircled{1} The \textit{GPU grouping} that indicates how the available GPUs are partitioned into different TP groups.
\mytextcircled{2} The \textit{pipeline orchestration} that describes how to organize multiple training pipelines with the TP groups, where each group serves as one stage of one pipeline.
\mytextcircled{3} The \textit{layer assignment} of each pipeline, which represents how the model layers are assigned to the corresponding groups.
\mytextcircled{4} The \textit{training data assignment} of each pipeline, which corresponds to how a global batch of training data is scattered among the pipelines.

To ease the planning, we formulate the deduction of parallelization plan as a bi-level optimization problem as follows.
\begin{itemize}[leftmargin=*]
\item \textbf{Upper-level problem:} 
Suppose there are $N$ GPUs available, and denote the straggling rate of the $i$-th GPU as $x_i$. The objective of the upper-level problem is to find out the best combination of GPU grouping and pipeline orchestration that minimizes the training time. 
\item \textbf{Lower-level problem:}
Suppose there is an arbitrary combination of GPU grouping and pipeline orchestration, and denote the straggling rate of the $j$-th stage in the $i$-th pipeline as $y_{i,j}$. The objective of the lower-level problem is to determine the best combination of layer assignment and training data assignment minimizing the training time. 
\end{itemize}
Undoubtedly, the Cartesian product of the solution regions of the two problems covers all feasible parallelization plans, so we only need to solve the bi-level problem. 

Figure~\ref{fig:planner_routine} presents an overview of our planning routine. There are three processes. 
Given the profiled information, \system initiates the GPU grouping process (\S\ref{sec:planner_upper_level_gpu_grouping}), which aims to construct TP groups given the available GPUs. 
Then, the pipeline orchestration process (\S\ref{sec:planner_upper_level_pipeline_orchestration}) determines how to organize multiple training pipelines with these groups. Both the GPU grouping and pipeline orchestration processes are for the upper-level problem. 
Finally, the work assignment process (\S\ref{sec:planner_lower_level}) solves the lower-level problem by a joint optimization of layer and training data assignments.

\subsection{Solving the Lower-level Problem}
\label{sec:planner_lower_level}
Suppose we have orchestrated $\dpd$ pipelines, and the $i$-th pipeline consists of $\ppdi$ stages. As introduced in \S\ref{sec:system_plan_design}, the pipeline stages may vary in the number of GPUs and performance due to the stragglers. 
The lower-level problem aims to jointly deduce the optimal assignments of model layers and training data, in order to minimize the training time. 

\mysubsubsection{Definition of Group Straggling Rates}
We first define the straggling rates of TP groups, which are utilized to estimate the efficiency of a group in our work. 
Suppose there is a TP group containing $n$ GPUs with straggling rates of $\{x_{k_1}, \cdots, x_{k_n}\}$. 
Such GPU straggling rates are given by the profiler, as introduced in \S\ref{sec:system_components}. 
Then, the group straggling rate is calculated based on two considerations:
\begin{itemize}[leftmargin=*]
\item The slowest straggler dominates the running time of the whole group due to the synchronous nature of TP. Hence, the group straggling rate is dependent on $\max\{ x_{k_{1}}, \cdots, x_{k_{n}} \}$. 
\item Since the workload in a TP group is evenly distributed, the number of GPUs in a group matters. Denote $\zeta_{n}$ as the time cost of a unit workload (e.g., one Transformer layer with a batch size of 1) with $n$ non-straggling GPUs. Then, we compute $\rho_{n} = \zeta_{n} / \max_{n^\prime}\{\zeta_{n^\prime}\}$ as the coefficient of efficiency degradation when the group consists of $n$ GPUs.  
Such coefficients can be profiled and computed beforehand. 
\end{itemize}
Putting them together, the group straggling rate is calculated as $y = \rho_n \times \max\{x_{k_1}, \cdots, x_{k_n}\}$.

\mysubsubsection{Cost Model}
For all $ i \in [1, \dpd], j \in [1, \ppdi]$, denote $m_i$ as the number of micro-batches assigned to the $i$-th pipeline, and $l_{i,j}$ as the number of layers assigned to the $j$-th stage in the $i$-th pipeline. We introduce the cost modelling for training time and memory consumption, respectively. 

To begin with, we focus on the training time. 
For a micro-batch size of $b$, we denote $\tau(b)$ as the running time (including forward and backward) of one layer when the group straggling rate is 1, which can be profiled in advance. Since modern large models (e.g., LLMs) typically consist of identical layers, the running time of the $j$-th stage in the $i$-th pipeline for one micro-batch can be modelled as $t_{i,j} = y_{i,j} \times l_{i,j} \times \tau(b)$. 

Following prior works~\cite{alpa,oobleck,galvatron_bmw}, we model the running time of the $i$-th pipeline as 
$T_i = (m_i - 1) \times \max_{j}\{t_{i,j}\} + \sum_{j}{t_{i,j}}$,
where the first term represents the running time of the 1F1B phase and the second term is for the warm-up and cool-down phases. 
Since large models are usually trained with a huge batch size (e.g., the global batch size in LLM training is usually a few millions of tokens) whilst the maximum accommodable micro-batch size is relatively small due to the limited GPU memory, the number of micro-batches ($m_i$) is usually much larger than that of pipeline stages ($\ppdi$). 
Therefore, we simplify the modelling for training time as 
$T_i \approx m_i \times \max_{j}\{t_{i,j}\} = m_i \times \max_{j}\{y_{i,j} \times l_{i,j}\} \times \tau(b)$. 
Despite the simplification, empirical results in \S\ref{sec:expr_e2e} show that the estimated running time given by our cost model is extremely close to the actual running time.

Next, for memory modelling, 
similar to prior works \cite{alpa,galvatron_bmw}, we focus on the memory of model states and forward activations. Both are proportional to the number of layers, whilst the memory of activations is also related to the stage index $j$ in 1F1B pipeline execution. For simplicity, we denote $l_{i,j} \times \mu_{i,j}(b) + \nu_{i,j}(b)$ as the memory usage for the $j$-th stage in the $i$-th pipeline, where $\mu_{i,j}(b), \nu_{i,j}(b)$ are stage-specific coefficients calculated by profiling. Besides, as groups can vary in the number of GPUs, we denote $C_{i,j}$ as the memory capacity. Due to space constraints, we leave the details of how to calculate $\mu_{i,j}(b), \nu_{i,j}(b), C_{i,j}$ to Supplementary Appendix~B.4 %\ref{appendix:deduc_memory}. 

\mysubsubsection{Deriving the Assignments}
Based on the cost model, we wish to deduce the optimal values for $l_{\cdot, \cdot}, m_{\cdot}, b$. Recalling that the maximum accommodable micro-batch size is usually small, we enumerate the value for $b \in \{1, 2, \cdots\}$ until all assignments exceed the memory constraint. 
For each enumerated $b$, the lower-level problem, i.e., how to minimize the running time of the slowest pipeline by partitioning the model layers ($l_{\cdot, \cdot}$) and training data ($m_{\cdot}$), can be formulated as
\begin{equation}
\label{eq:lower_problem_one_case_origin}
\small
\begin{aligned}
& \argmin_{l_{i,j}, m_i} 
\max_{i \in \left[1, \dpd\right]} \left\{ \max_{j \in \left[1, \ppdi\right]} \left\{y_{i,j} \times l_{i,j}\right\} \times m_i \times \tau(b) \right\} \\
& \text{ s.t. }  
 \sum_{i=1}^{\dpd} m_i \times b = B, \;
\sum_{j=1}^{\ppdi} l_{i,j} = L \text{, } \forall i \in \left[1, \dpd\right] \\
& \;\;\;\;\;\;\; l_{i, j} \times \mu_{i,j}(b) + \nu_{i,j}(b) \leq C_{i,j} \text{, } \forall i \in \left[1, \dpd\right], \forall j \in {\left[1, \ppdi\right]} \\
& \;\;\;\;\;\;\; \vphantom{\sum_{i=1}^{\dpd}} l_{i,j}, m_i \in \mathbb{N}_0 \text{, } \forall i \in \left[1, \dpd\right], \forall j \in {\left[1, \ppdi\right]} 
\end{aligned}
\end{equation}
It can be proved that the partitioning of training data and model layers is orthogonal, so solving Eq.~\eqref{eq:lower_problem_one_case_origin} is equivalent to solving two orthogonal types of sub-problems (detailed proofs provided in Supplementary Appendix~B.5. %\ref{appendix:deduce_eq-1-2-3}).
The first type consists of $\dpd$ sub-problems, with the $i$-th sub-problem written as
\begin{equation}
\label{eq:lower_problem_one_case_layer}
\small
\begin{aligned}
& \argmin_{l_{i,j}} \max_{j \in \left[1, \ppdi\right]} \left\{ y_{i,j} \times l_{i,j} \right\} \\
& \text{ s.t. } 
\sum_{j=1}^{\ppdi} l_{i,j} = L \text{, } \forall i \in \left[1, \dpd\right] \\
& \;\;\;\;\;\;\; l_{i, j} \times \mu_{i,j}(b) + \nu_{i,j}(b) \leq C_{i,j}, \;\; %\text{ for } \forall j \in {[1, \ppdi]} \\
l_{i,j} \in \mathbb{N}_0 \text{, } \forall j \in {\left[1, \ppdi\right]} 
\end{aligned}
\end{equation}
Eq.~\eqref{eq:lower_problem_one_case_layer} is a well-formulated integer linear programming (ILP) problem, which can be solved efficiently. Denote $o_i$ as the optimal value for the $i$-th sub-problem of the first type, then the second type of sub-problem is as follows.
\begin{equation}
\label{eq:lower_problem_one_case_data}
\small
\begin{aligned}
& \argmin_{m_i} \max_{i \in \left[1, \dpd\right]} \left\{ o_i \times m_i \right\} \times \tau(b) \\
& \text{ s.t. } 
\sum_{i=1}^{\dpd} m_i \times b = B, \;\;
m_i \in \mathbb{N}_0 \text{, } \forall i \in \left[1, \dpd\right]
\end{aligned}
\end{equation}
Eq.~\eqref{eq:lower_problem_one_case_data} is also a well-formulated ILP problem. By solving these $\dpd + 1$ sub-problems, 
we obtain the optimal layer assignment for each stage and the optimal training data assignment for each pipeline. 

It is worth noting that solving these ILP problems can automatically assign zero layers to groups with high straggling rates. The GPUs in these groups are removed from the training for better efficiency, as exemplified in Figure~\ref{fig:planner_routine}. 

Last but not least, the numbers of decision variables in Eq.~\eqref{eq:lower_problem_one_case_layer} and Eq.~\eqref{eq:lower_problem_one_case_data} are $\ppdi$ and $\dpd$, respectively, which are very small in practice. Thus, the lower-level problem is solvable within a short period of time. 
In Supplementary Appendix~A.2, %\ref{appendix:scalability}, 
we perform a detailed time breakdown of our parallelization planning algorithm, which shows that the time cost of the lower-level problem is negligible.

\subsection{Solving the Upper-level Problem}
\label{sec:planner_upper_level}

Given $N$ available GPUs, there are tremendous ways to organize the pipelines, and finding the best one is far from trivial. 
For one thing, as introduced in \S\ref{sec:system_plan_design}, we allow the number of stages to be varied among pipelines and the number of GPUs to be varied among the stages, so how to partition the GPUs into different groups is challenging. 
For another, to construct the pipelines given a feasible grouping result, we are required to determine which pipeline stage should each group serve as, and it is impossible to enumerate all candidate constructions. 
In this section, we focus on the two processes of GPU grouping and pipeline orchestration.

\subsubsection{\ul{\textbf{GPU Grouping}}}
\label{sec:planner_upper_level_gpu_grouping}
We first describe the GPU grouping process in our work. 

\mysubsubsection{Even Partitioning}
Suppose all GPU straggling rates are close, then the GPUs are expected to be evenly partitioned into TP groups to balance their performance. 
As mentioned in \S\ref{sec:pre_parallel}, it is a common practice to enforce TP within the same node due to its high communication cost. Thus, we can break down the GPU grouping into different nodes individually. 
In particular, we have the following theorem (proofs provided in Supplementary Appendix~B.1). %\ref{appendix:thm-1}).
\begin{theorem}
\label{thm:tp_in_one_node}
Suppose there are $n$ GPUs in a node with straggling rates $\{x_1, \cdots, x_n\}$, and we need to partition them into $n/k$ groups (each with $k$ GPUs). 
Denote $\{i_1, \cdots, i_n\}$ as the ordering satisfying $x_{i_1} \geq \cdots \geq x_{i_n}$. 
Then, the best grouping result that minimizes the running time is 
$\{ 
    \{ x_{i_1}, \cdots, x_{i_{k}} \},
    \{ x_{i_{k+1}}, \cdots, x_{i_{2k}} \}, 
    \cdots,
    \{ x_{i_{n - k + 1}}, \cdots, x_{i_n} \}
\}$.
\end{theorem}
Theorem~\ref{thm:tp_in_one_node} suggests that putting GPUs with similar performance into the same group is preferable --- with the slower GPUs grouped together, we can mitigate their mutual delays and ensure that they do not affect the performance of the other (faster) GPU groups.
Besides, Theorem~\ref{thm:tp_in_one_node} does not make any assumptions on how the pipelines are constructed based on these groups nor how the layers and training data are assigned. 
Recalling that we do not consider cross-node grouping, we can perform the partitioning within each node and summarize the groups on all nodes.

\begin{figure}[!t]
\centering
\includegraphics{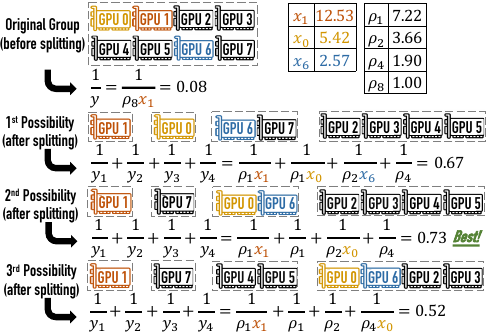}
\caption{\small{An example of GPU grouping with splitting.}}
\label{fig:split_and_group}
\end{figure}

\mysubsubsection{Group Splitting}
In practice, there may exist heavy stragglers that slow down the corresponding groups severely. In such cases, solving the lower-level problem (Eq.~\eqref{eq:lower_problem_one_case_origin} in \S\ref{sec:planner_lower_level}) would give the result that very few or even none of the layers should be assigned to the corresponding groups, leading to the waste of resources since the other GPUs in these groups become extremely under-utilized. Undoubtedly, it would be preferable if we split the groups to isolate the heavy stragglers (i.e., letting them form individual groups with a TP degree of 1).

Suppose there are 8 GPUs in the original group, after a heavy straggler has been isolated, then we need to re-group the rest 7 GPUs. 
However, Theorem~\ref{thm:tp_in_one_node} becomes inapplicable as 7 GPUs cannot be partitioned into groups with equal size. Besides, to obtain three groups with 1, 2, and 4 GPUs, respectively, there exist up to 6 possible grouping results (see details in Supplementary Appendix~B.7). % \ref{appendix:deduce_grouping-possibilities}). 
To address this, we devise a mechanism to estimate the theoretical efficiency of an arbitrary grouping result, motivated by 
the following theorem (proofs provided in Supplementary Appendix~B.2). % \ref{appendix:thm-2}).

\begin{theorem}
\label{thm:theoretic_optimality}
Suppose there are two different grouping results that consist of $M^{\prime}$ and $M^{\prime\prime}$ groups with straggling rates of $\{ y^{\prime}_1, \cdots, y^{\prime}_{M^{\prime}} \}$ and $\{ y^{\prime\prime}_1, \cdots, y^{\prime\prime}_{M^{\prime\prime}} \}$, respectively. 
If we ignore the memory constraints in Eq.~\eqref{eq:lower_problem_one_case_origin}, and further assume the layer and training data assignments are not restricted to integers (i.e., $l_{i,j}, m_i$ in Eq.~\eqref{eq:lower_problem_one_case_origin} can be any positive real numbers), then the minimum training time of the two grouping results satisfy 
${T^{\prime}}/{T^{\prime\prime}} = {(\sum_{i=1}^{M^{\prime\prime}} 1 / y^{\prime\prime}_i)}/{(\sum_{i=1}^{M^{\prime}} 1 / y^{\prime}_i)}$.
\end{theorem}

Theorem~\ref{thm:theoretic_optimality} mathematically proves the ratio of optimal training times achievable by different grouping results under several assumptions to relax the constraints in Eq.~\eqref{eq:lower_problem_one_case_origin}. 
Despite the assumptions, it provides a constant-time method to examine how each possible grouping result after splitting performs compared with that before splitting.
If none of the possible grouping results improve the estimated efficiency, then we will keep the straggling GPU. Otherwise, the best one after splitting will be chosen, as depicted in Figure~\ref{fig:split_and_group}. 

To sum up, all GPUs are first evenly partitioned into groups by Theorem~\ref{thm:tp_in_one_node}. Then, we iterate the straggling GPUs in descending order w.r.t. their straggling rates. For each straggling GPU, we examine whether we should isolate it and update the grouping by Theorem~\ref{thm:theoretic_optimality}. Such a routine will be executed for each candidate TP degree in $\{1,2,4,8\}$, producing 4 grouping results in total, which are then forwarded to the pipeline orchestration process.

\subsubsection{\ul{\textbf{Pipeline Orchestration}}}
\label{sec:planner_upper_level_pipeline_orchestration}
For each grouping result, to orchestrate multiple pipelines, there are two decisions to make: (i) how to divide the groups into multiple pipelines, and (ii) how to order the groups within each pipeline. \system makes the decisions through two steps, namely pipeline division and group ordering.

\mysubsubsection{Pipeline division}
We first focus on the first step. 
A na\"ive approach is to enumerate all possible division results and solve the lower-level problem to achieve the best one. However, the number of possible division results (a.k.a. the Bell number) grows exponentially w.r.t. the number of groups, making the na\"ive approach infeasible. 

To cope with this issue, we simplify the problem from two perspectives.
First, we leverage the fact that most GPUs are not stragglers, which further indicates most groups are associated with the same value of $y$. Thus, we treat these groups as identical to reduce the possible division results. 
Second, we loosen the constraints in Eq.~\eqref{eq:lower_problem_one_case_origin} in order to accelerate the evaluation of each division result. 

To elaborate, suppose there are $M$ groups in total, where the majority of them share the same straggling rate of $\hat{y}$ (fast groups), and meanwhile $M_{s}$ of them are different from the majority with straggling rates $\{y_1, \cdots, y_{M_s}\}$ (slow groups). 
If we ignore the memory constraints in Eq.~\eqref{eq:lower_problem_one_case_origin}, and further assume the layer assignments are not restricted to integers (i.e., $l_{i,j}$ in Eq.~\eqref{eq:lower_problem_one_case_origin} can be any positive real numbers), then the problem of finding the best pipeline division to minimize the running time can be formulated as follows (detailed deduction provided in Supplementary Appendix~B.6). % \ref{appendix:deduce_eq-4}).

\begin{equation}
\label{eq:pipeline_division_estimation}
\small
\begin{aligned}
& \argmin_{m_i, h_i, q_{i,k}} 
\max_{i \in \left[1, \dpd\right]} \left\{ \frac{m_i \times \tau(b)}{h_i \times \hat{y} + \sum_{k=1}^{M_{s}} {q_{i,k}}/{y_{k}}} \right\} \\
& \text{ s.t. }
\sum_{i=1}^{\dpd} m_i = \frac{B}{b}, \sum_{i=1}^{\dpd} h_i = M - M_{s}, \sum_{i=1}^{\dpd} q_{i,k} = 1 \text{, } \forall k \in \left[1, M_{s}\right] \\
& \;\;\;\;\;\;\; q_{i,k} \in \{0, 1\}, m_i, h_i \in \mathbb{N}_0\text{, } \forall i \in \left[1, \dpd\right], \forall k \in \left[1, M_{s}\right]
\end{aligned}
\end{equation}
where $h_i$ denotes the number of fast groups in the $i$-th pipeline, $q_{i,k} = 1$ indicates the $k$-th slow group is partitioned to the $i$-th pipeline, and $q_{i,k} = 0$ vice versa.
The problem in Eq.~\eqref{eq:pipeline_division_estimation} is a mixed-integer non-linear programming (MINLP) problem, a kind of complex combinatorial optimization problems that is usually more time-consuming to solve compared with ILP problems. 
Nevertheless, thanks to the reduction in possible division results and the loosened constraints, there is only a handful of decision variables in Eq.~\eqref{eq:pipeline_division_estimation}. Thus, we can solve the problem within a manageable time.

\mysubsubsection{Group ordering} 
Since Eq.~\eqref{eq:pipeline_division_estimation} is irrelevant to the ordering of groups within each pipeline, before feeding the pipelines to the lower-level problem, we need to determine the group ordering. 

When there are no heavy stragglers to isolate, our GPU grouping will produce groups with the same number of GPUs. In such cases, the group ordering can be determined straightforwardly. Formally, we have the following theorem (proofs provided in Supplementary Appendix~B.3). % \ref{appendix:thm-3}). 
\begin{theorem}
\label{thm:group_ordering_in_pipeline}
Suppose the groups assigned to the same pipeline have the same number of GPUs, then the best ordering of pipeline stages satisfies that the groups are in descending order w.r.t. the group straggling rates. 
\end{theorem}
It suggests that faster groups should serve as the ending stages in a pipeline. The rationale behind Theorem~\ref{thm:group_ordering_in_pipeline} is that the beginning stages need to reserve more memory for the forward activations, so putting faster groups to the ending stages allows them to process more layers for better efficiency. 
Therefore, we can determine the ordering easily but effectively.

When some heavy stragglers are isolated, our GPU grouping may produce groups with unequal numbers of GPUs, so we cannot apply Theorem~\ref{thm:group_ordering_in_pipeline} directly.   
Fortunately, the number of GPUs in a group is exactly the TP degree of the corresponding pipeline stage, which is typically restricted in $\{1,2,4,8\}$. Consequently, we can enumerate the ordering of TP degrees. To be specific, for each pipeline, groups with the same number of GPUs are first bundled together, and groups in the same bundle are sorted via Theorem~\ref{thm:group_ordering_in_pipeline}. Then, we enumerate the ordering of bundles and solve the lower-level problem to evaluate the efficiency of each enumeration. Eventually, the best enumeration will be selected. 
Considering that there are 24 ordering of bundles at most, and solving the ILP problems in \S\ref{sec:planner_lower_level} is extremely fast, such an enumeration-based approach works well.

\subsubsection{\ul{\textbf{Putting Them Together}}}

We finish this section by summarizing the overall routine. 
In the GPU grouping process, we enumerate the maximum TP degrees in $\{1,2,4,8\}$ to obtain 4 grouping results. 
Then, the pipeline orchestration process takes each grouping result as input, and forms $\dpd$ pipelines\footnote{Since the memory consumption of model parameters increases w.r.t. the DP degree in hybrid parallel, we maintain the DP degree before and after the parallelization plan adjustment. It is also feasible to consider different DP degrees, e.g., by simply enumerating DP degrees within a small range.}. 
By doing so, there are 4 candidate solutions to the upper-level problem, which will be fed into the lower-level problem to determine the best one. 

As we will show in \S\ref{sec:expr_e2e}, our planning algorithm is very efficient, taking merely 10-30 seconds in all experiments. 
Additionally, in Supplementary Appendix~A.2, % \ref{appendix:scalability}, 
we present a time breakdown to analyze the time cost of our planning algorithm and evaluate its efficiency when more GPUs are involved.

\section{Malleableized Training}
\label{sec:malleable_training}

This section describes how \system manages the model states and hardware devices to achieve malleableized training, enabling the immediate adjustment of parallelization plans to handle dynamic changes in straggler situations.

\subsection{Model Management}
\label{sec:malleable_training_model_management}

\mysubsubsection{Model Sharding}
It is a popular choice to partition the model states via the ZeRO-1 optimizer~\cite{zero} in hybrid parallel. 
Suppose $\tpd$ is the TP degree of an arbitrary model layer, then the associated model states are sharded into $\dpd \times \tpd$ slices, and these slices are scattered to unique GPUs, as shown in Figure~\ref{fig:model_sharding}(a). 
Whilst in \system, we adjust the model sharding to accommodate varying TP degrees, as shown in Figure~\ref{fig:model_sharding}(b). 
To be precise, for an arbitrary model layer, suppose its TP degree in the $i$-th pipeline is $\tpd_i$ and let $\tpd_{max} = \max_{i} \{\tpd_i\}$, then the corresponding model states are sharded into $\dpd \times \tpd_{max}$ slices, where each GPU in the $i$-th pipeline is responsible for $\tpd_{max} / \tpd_i$ slices. 
 
After the backward propagation, each GPU holding two or more slices should invoke multiple \texttt{reduce-scatter} communications to synchronize the gradients, along with multiple \texttt{all-gather} communications to retrieve the updated models. 
\system automatically identifies these GPUs and handles the ordering of communication calls to avoid deadlocks.

\begin{figure}[!t]
\centering
\includegraphics{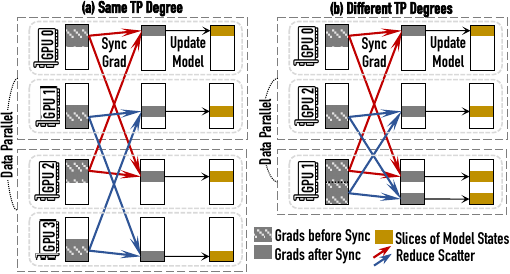}
\caption{\small{Illustration of model sharding. The red and blue arrows indicate different \texttt{reduce-scatter} communications for gradient synchronization. The \texttt{all-gather} communications after model update work inversely and are omitted.}}
\label{fig:model_sharding}
\end{figure}

\mysubsubsection{Model Migration}
Once the parallelization plan is adjusted, we must migrate the model slices to fulfill the new plan. 
For each layer, we locate the source and destination of each model slice, summarizing the many-to-many communication among the GPUs. Then, we fuse the migration of different slices with the \texttt{batched-send-recv} primitive for better efficiency. In addition, we pack the migration of multiple layers (4 layers by default) together to make full use of the network bandwidth. 

However, if a failure occurs, some GPUs may not be responding and the model states owned by them become unavailable. Although it is possible to incorporate the idea of storing redundant model states on each GPU~\cite{bamboo}, it increases memory consumption and leads to performance degradation. Thus, in such cases, we recover the training task by loading the latest model checkpoint onto the remaining GPUs and setting the straggling rates of unresponsive GPUs as infinite.

\subsection{Device Management}
\label{sec:malleable_training_device_management}

\mysubsubsection{Straggler Detection}
During the training, some devices may occasionally lag behind and become stragglers. 
To detect such dynamic stragglers, the profiler in \system records the hardware efficiency based on CUDA events.\footnote{Timing with CUDA events incurs negligible overhead, which can be validated by the results in \S\ref{sec:expr_e2e} that \system achieves comparable performance against Megatron-LM when there are no stragglers.}
In particular, we assess the computation and communication time cost of each GPU to distinguish the slower ones, and compute their GPU straggling rates by comparing them with the normal ones. Besides, we add a threshold for communication calls during training in order to detect failures.

\mysubsubsection{Elastic Scaling}
As introduced in \S\ref{sec:planner_lower_level}, \system strategically removes heavy stragglers by assigning zero layers. However, these GPUs could be back to normal or become light stragglers later. Thus, instead of removing these GPUs permanently, we maintain them as standby devices, and periodically conduct micro-benchmarks to assess their GPU straggling rates. During each time of re-planning, our planning algorithm is able to adaptively determine whether there are removed GPUs to be involved as well as whether there are new heavy stragglers to be removed. By doing so, \system supports elastic scaling of involved GPUs during the training.

\subsection{Re-planning with Overlapping}
\label{sec:malleable_training_replanning_overlapping}

When any of the GPU straggling rates have changed considerably (greater than 5\% in our implementation), \system triggers a re-planning process to accommodate the dynamicity, which involves the execution of the planning algorithm to derive a new parallelization plan as well as the migration of model states. However, although the time cost of our planning algorithm is not substantial (around 10-30 seconds in our experiments), it still leads to non-negligible idle periods if we halt the training task during the planning. To cope with this problem, we devise an asynchronous re-planning mechanism --- instead of leaving the GPUs idle, we continue training with the current parallelization plan, and execute the planning algorithm concurrently. 
Specifically, the re-planning is carried out using background processes on the CPUs, and meanwhile, the training keeps running on the GPUs with the current parallelization plan. 
Once re-planning is complete, if the new plan differs from the current one, the model migration to the updated plan occurs after the current training iteration finishes.
In practice, we find that the planning finishes within one training step, achieving satisfactory overlapping. Although the model migration cannot be overlapped, it only takes a short period of time (around 1-5 seconds in our experiments), which is acceptable.

\section{Implementation}
\label{sec:impl}

\system is designed to adapt to the dynamic stragglers, featuring a series of non-uniform partitioning and real-time adjustments in the parallelization plans. 
Fulfilling this aim requires complicated data management as well as graph information to facilitate model migration after re-planning. 
However, obtaining the computation graph information in PyTorch~\cite{pytorch} is non-trivial.
Hence, we implement \system atop our self-developed DL system, namely Hetu\footnote{https://github.com/PKU-DAIR/Hetu}~\cite{hetu,hotspa}. 
We develop the planner and the elastic executor with 3.5K LoC in Python, realize the model migration functionality with 2.3K LoC in C++/CUDA and 0.5K LoC in Python, and implement a computation graph-based system for managing non-uniform data, layers, stages, and device partitioning in distributed training with 27.3K LoC in C++/CUDA. 
To implement the planning algorithm, we use the PuLP~\cite{pulp} and Pyomo~\cite{pyomo} libraries to solve the ILP and MINLP problems, respectively. 
Hetu is particularly optimized for LLM training, with communication primitives implemented with NCCL~\cite{nccl} and computation kernels accelerated via libraries such as FlashAttention~\cite{flash_attn,flash_attn_v2}, cuBLAS~\cite{cublas}, and cutlass~\cite{cutlass}, matching the performance of Megatron-LM when there are no stragglers (evaluated in \S\ref{sec:expr_e2e}).
Note that our implementation and evaluation focus on LLMs due to their huge model sizes and the demand of training with massive GPUs, whilst the proposed designs for parallelization planning and malleableized training are applicable to more forms of deep learning models. We leave them as potential future extensions.

\section{Experimental Evaluation}
\label{sec:expr}

\subsection{Experimental Setup}
\label{sec:expr_setup}

\mysubsubsection{Hardware Environments}
We conduct all experiments on 8 GPU servers equipped with 8$\times$A800 (80G) GPUs, containing 64 GPUs in total. The GPUs within the same server are connected via NVLink with a bandwidth of 400GB/s, and the servers are connected via InifiBand with a bandwidth of 200GB/s. 
Note that although we consider NVIDIA GPUs in our evaluation, we do not make any assumptions about the choice of hardware accelerators. Thus, our work is also applicable to other hardware accelerators.

\mysubsubsection{Workloads}
We consider three LLMs in the LLaMA-2 architecture~\cite{llama2} with 32B, 70B, and 110B parameters, respectively. We train the 32B model over 32 GPUs and the other two models over 64 GPUs. The context length is set as 4K following most open-sourced LLMs. 
The global batch size (i.e., $B$) is set as 64 by default, constituting each batch with 256K tokens.

\mysubsubsection{Baselines}
To the best of our knowledge, none of the existing hybrid parallel training frameworks address the dynamic straggler problem. 
Thus, we mainly compare \system with two \textit{state-of-the-art (SOTA)} LLM training frameworks:
(1) Megatron-LM, a powerful LLM training framework that integrates DP, TP (empowered by sequence parallel~\cite{megatron_3}), and PP; 
(2) DeepSpeed, which utilizes the ZeRO-3 optimizer~\cite{zero} to scatter model states across the devices (a.k.a. Fully Sharded Data Parallel~\cite{pytorch_fsdp}) and requires gathering model parameters for each layer in both forward and backward propagation. 
Furthermore, for both baselines, we also assess their performance by manually excluding the nodes with straggling GPUs and restarting the training task, which are denoted as ``Megatron-LM w/ Restart'' and ``DeepSpeed w/ Restart'', respectively.

Besides, we further compare \system with a \textit{fault-tolerant} training framework, namely Oobleck~\cite{oobleck}, which supports dynamic migration to recover from failures. 
To be specific, we treat stragglers as faulty GPUs and let Oobleck perform migrations to exclude the stragglers and continue training.

\mysubsubsection{Straggler Simulation}
As analyzed in recent studies~\cite{megascale,imbue_report}, there are various kinds of root causes that could lead to stragglers. It is difficult to develop a benchmark for reproducing these root causes since they are hard to control. 
Thus, we control the dynamic straggler patterns by simulation to achieve a fair comparison.

Specifically, we launch extra computing processes on some GPUs to make them straggling, and we consider three levels of stragglers by launching 1-3 processes, indicated as level-1, -2, and -3 stragglers, respectively. We generate a trace containing six straggler situations to simulate diverse scenarios: 
(S1) one level-1 straggler; (S2) one level-3 straggler; (S3) one level-1 straggler and one level-3 straggler, residing in different nodes; (S4) one level-1 straggler, one level-2 straggler, and one level-3 straggler, residing in different nodes; (S5) eight level-1 stragglers on the same node and one level-2 straggler on another node; (S6) eight level-1 stragglers on the same node. 

The generated trace consists of GPU-granular stragglers (i.e., S1 to S4), node-granular stragglers (i.e., S6), and a complex situation with both GPU- and node-level stragglers (i.e., S5) to examine the robustness of the evaluated frameworks. Moreover, it contains the transitions where the stragglers appear or disappear, which matches the dynamicity in real-world scenarios. For instance, we place the most severe situation (i.e., S5) second to last, to assess whether our work can detect and adjust when the straggler disappears.

\mysubsubsection{Protocols}
We focus on the training time of each competitor under various straggler situations. For the baselines, we tune their configuration for each training task to achieve the best performance. 
The parallelization plans of \system are generated by our planning algorithm. It is noteworthy that when all straggling ratios are 1 (i.e., no stragglers), our planning algorithm produces the same 3D parallel configurations as Megatron-LM.

\begin{figure*}[!t]
\centering
\includegraphics[width=\textwidth]{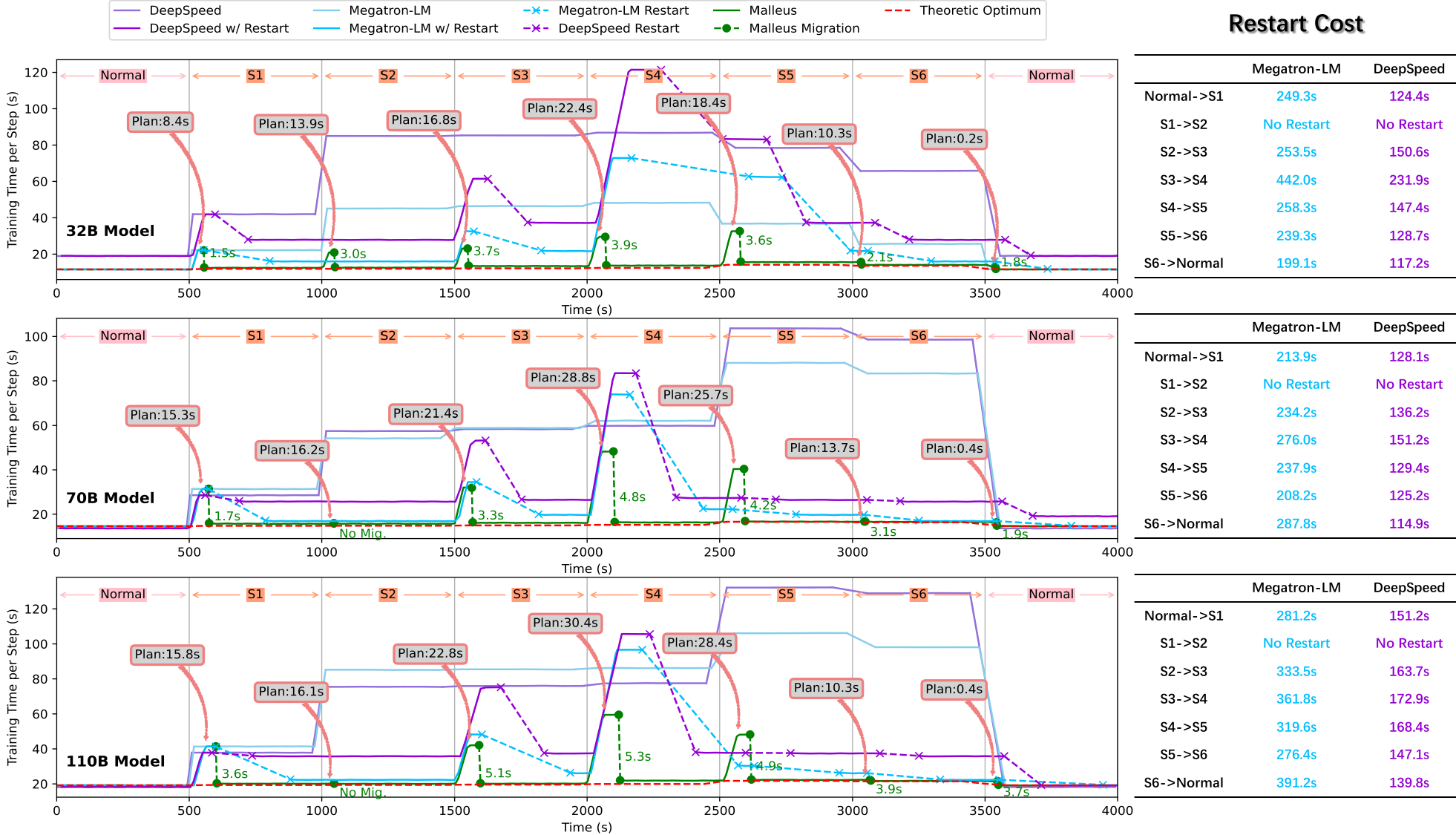}
\caption{\small{End-to-end evaluation on a trace consisting of six straggler situations (``Normal'' indicates there are no stragglers). The $x$-axis represents the trace and the $y$-axis represents the running time per step (in seconds). During each time of re-planning of \system, the time cost of planning (highlighted in background gray color) is overlapped by the training. Additionally, the time cost of migration of \system is provided in green (on the left), and the time cost of restarting of Megatron-LM and DeepSpeed is provided in blue and purple (on the right).}}
\label{fig:e2e_dynamic}
\end{figure*}

\begin{table*}[!t]
\small
\centering
\caption{\small{{Averaged running time per step (in seconds) under the straggler situations in Figure~\ref{fig:e2e_dynamic}. We also provide the Model FLOPS Utilization (MFU) when there are no stragglers. The values in parentheses represent the improvement achieved by \system compared to the baselines. 
``Avg. Improv.'' indicates the average improvement measured in terms of geometric mean.
``Theoretic Opt.'' indicates the theoretic optimum, which is computed as $T_{\textup{normal}} \times {N}/{((N-n) + \sum_{i=1}^{n}{1 / x_i})}$, where $T_{\textup{normal}}$ is the running time when there are no stragglers (i.e., considering the hardware capability to be inversely proportional to the straggling rates).}
}}
\label{tb:e2e}
\begin{tabular}{c||c|c||c|c|c|c|c|c|c}
\hline
\toprule
& & \specialcell{Normal\\(Time, MFU)} & S1 & S2 & S3 & S4 & S5 & S6 & \specialcell{Avg.\\Improv.}
\\
%%%%%%%%%%%%%%%%%% 32B
\hline
\hline
\multirow{10}*{\rotatebox[origin=c]{90}{\small{32B}}}
& \specialcell{DeepSpeed\\w/o Restart} & 19.0, 29.6\%  
& \specialcell{42.0\\(3.36$\times$)}
& \specialcell{84.9\\(6.73$\times$)} 
& \specialcell{85.3\\(6.36$\times$)} 
& \specialcell{86.8\\(6.38$\times$)} 
& \specialcell{78.5\\(5.03$\times$)} 
& \specialcell{65.8\\(4.67$\times$)} 
& 5.28$\times$ 
\\
\cline{2-10}
& \specialcell{Megatron-LM\\w/o Restart} & 11.6, 48.5\%
& \specialcell{22.2\\(1.77$\times$)} 
& \specialcell{45.1\\(3.58$\times$)} 
& \specialcell{46.4\\(3.46$\times$)} 
& \specialcell{48.2\\(3.54$\times$)} 
& \specialcell{36.8\\(2.35$\times$)} 
& \specialcell{25.6\\(1.81$\times$)} 
& 2.63$\times$ 
\\
\cline{2-10}
&\specialcell{DeepSpeed\\w/ Restart} & 19.0, 29.6\% 
& \specialcell{27.9\\(2.23$\times$)} 
& \specialcell{27.9\\(2.21$\times$)} 
& \specialcell{37.2\\(2.78$\times$)} 
& \specialcell{83.1\\(6.11$\times$)} 
& \specialcell{37.2\\(2.38$\times$)} 
& \specialcell{27.9\\(1.98$\times$)} 
& 2.71$\times$ 
\\
\cline{2-10}
& \specialcell{Megatron-LM\\w/ Restart} & 11.6, 48.5\%
& \specialcell{16.0\\(1.28$\times$)} 
& \specialcell{16.0\\(1.27$\times$)} 
& \specialcell{21.7\\(1.62$\times$)} 
& \specialcell{62.4\\(4.59$\times$)} 
& \specialcell{21.7\\(1.39$\times$)} 
& \specialcell{16.0\\(1.13$\times$)} 
& 1.63$\times$ 
\\
\cline{2-10}
& \system & 11.6, 48.5\%
& \textbf{12.5} 
& \textbf{12.6} 
& \textbf{13.4} 
& \textbf{13.6} 
& \textbf{15.6} 
& \textbf{14.1} 
& -
\\
\hhline{~|*9-}
& \cellcolor{gray!20}{Theoretic Opt.} & \cellcolor{gray!20}{-}
& \cellcolor{gray!20}{11.9}
& \cellcolor{gray!20}{11.9}
& \cellcolor{gray!20}{12.2}
& \cellcolor{gray!20}{12.4}
& \cellcolor{gray!20}{14.2}
& \cellcolor{gray!20}{13.6}
& \cellcolor{gray!20}{-}
\\
%%%%%%%%%%%%%%%%%% 70B
\hline
\hline
\multirow{10}*{\rotatebox[origin=c]{90}{\small{70B}}}
& \specialcell{DeepSpeed\\w/o Restart} & 13.6, 48.2\% 
& \specialcell{28.5\\(1.81$\times$)} 
& \specialcell{57.4\\(3.65$\times$)} 
& \specialcell{58.3\\(3.62$\times$)} 
& \specialcell{59.9\\(3.67$\times$)} 
& \specialcell{103.6\\(6.24$\times$)}  
& \specialcell{98.6\\(5.97$\times$)} 
& 3.85$\times$ 
\\
\cline{2-10}
& \specialcell{Megatron-LM\\w/o Restart} & 14.6, 44.9\% 
& \specialcell{31.3\\(1.99$\times$)} 
& \specialcell{54.2\\(3.45$\times$)} 
& \specialcell{58.9\\(3.66$\times$)} 
& \specialcell{62.1\\(3.81$\times$)} 
& \specialcell{88.1\\(5.30$\times$)} 
& \specialcell{83.3\\(5.05$\times$)} 
& 3.70$\times$ 
\\
\cline{2-10}
& \specialcell{DeepSpeed\\w/ Restart} & 13.6, 48.2\% 
& \specialcell{25.7\\(1.64$\times$)} 
& \specialcell{25.7\\(1.64$\times$)} 
& \specialcell{26.5\\(1.65$\times$)} 
& \specialcell{27.3\\(1.67$\times$)} 
& \specialcell{26.5\\(1.60$\times$)} 
& \specialcell{25.7\\(1.58$\times$)} 
& 1.63$\times$ 
\\
\cline{2-10}
& \specialcell{Megatron-LM\\w/ Restart} & 14.6, 44.9\%
& \specialcell{16.9\\(1.08$\times$)} 
& \specialcell{16.9\\(1.08$\times$)} 
& \specialcell{19.9\\(1.22$\times$)} 
& \specialcell{22.3\\(1.37$\times$)} 
& \specialcell{19.7\\(1.19$\times$)} 
& \specialcell{16.9\\(1.02$\times$)} 
& 1.15$\times$ 
\\
\cline{2-10}
& \system & 14.6, 44.9\% 
& \textbf{15.7} 
& \textbf{15.7} 
& \textbf{16.1} 
& \textbf{16.3} 
& \textbf{16.6}
& \textbf{16.5}
& -
\\
\hhline{~|*9-}
& \cellcolor{gray!20}{Theoretic Opt.} & \cellcolor{gray!20}{-}
& \cellcolor{gray!20}{14.7}
& \cellcolor{gray!20}{14.7}
& \cellcolor{gray!20}{14.9}
& \cellcolor{gray!20}{15.2}
& \cellcolor{gray!20}{16.5}
& \cellcolor{gray!20}{16.3}
& \cellcolor{gray!20}{-}
\\
%%%%%%%%%%%%%%%%%% 110B
\hline
\hline
\multirow{10}*{\rotatebox[origin=c]{90}{\small{110B}}}
& \specialcell{DeepSpeed\\w/o Restart} & 18.2, 52.9\%
& \specialcell{37.9\\(1.81$\times$)} 
& \specialcell{75.4\\(3.75$\times$)} 
& \specialcell{76.0\\(3.76$\times$)} 
& \specialcell{77.4\\(3.55$\times$)} 
& \specialcell{132.3\\(5.93$\times$)} 
& \specialcell{129.0\\(5.94$\times$)} 
& 3.84$\times$ 
\\
\cline{2-10}
& \specialcell{Megatron-LM\\w/o Restart} & 19.2, 50.8\% 
& \specialcell{41.4\\(2.06$\times$)} 
& \specialcell{85.1\\(4.23$\times$)} 
& \specialcell{85.4\\(4.22$\times$)} 
& \specialcell{86.1\\(3.95$\times$)} 
& \specialcell{106.1\\(4.75$\times$)} 
& \specialcell{98.1\\(4.52$\times$)} 
& 3.82$\times$ 
\\
\cline{2-10}
& \specialcell{DeepSpeed\\w/ Restart} & 18.2, 52.9\% 
& \specialcell{35.8\\(1.78$\times$)} 
& \specialcell{35.8\\(1.78$\times$)} 
& \specialcell{37.4\\(1.98$\times$)} 
& \specialcell{37.8\\(1.73$\times$)} 
& \specialcell{37.4\\(1.68$\times$)} 
& \specialcell{35.8\\(1.65$\times$)} 
& 1.76$\times$ 
\\
\cline{2-10}
& \specialcell{Megatron-LM\\w/ Restart} & 19.2, 50.8\%
& \specialcell{22.3\\(1.11$\times$)} 
& \specialcell{22.3\\(1.11$\times$)} 
& \specialcell{26.1\\(1.29$\times$)} 
& \specialcell{30.3\\(1.39$\times$)} 
& \specialcell{26.1\\(1.17$\times$)} 
& \specialcell{22.3\\(1.03$\times$)} 
& 1.17$\times$ 
\\
\cline{2-10}
& \system & 19.2, 50.8\% 
& \textbf{20.1} 
& \textbf{20.1} 
& \textbf{20.2} 
& \textbf{21.8} 
& \textbf{22.3}
& \textbf{21.7} 
& -
\\
\hhline{~|*9-}
& \cellcolor{gray!20}{Theoretic Opt.} & \cellcolor{gray!20}{-}
& \cellcolor{gray!20}{19.4}
& \cellcolor{gray!20}{19.4}
& \cellcolor{gray!20}{19.6}
& \cellcolor{gray!20}{20.0}
& \cellcolor{gray!20}{21.7}
& \cellcolor{gray!20}{21.5}
& \cellcolor{gray!20}{-}
\\
\bottomrule
\end{tabular}
\end{table*}

\subsection{End-to-end Evaluation}
\label{sec:expr_e2e}

We first conduct experiments under the six straggler situations. Figure~\ref{fig:e2e_dynamic} presents how the running time of each competitor changes when there is a shift in the straggler situation, and Table~\ref{tb:e2e} lists the running time in detail. These empirical results demonstrate that \system consistently achieves the best performance under all straggler situations for all models. 

\mysubsubsection{Comparison to SOTA Baselines without Restarts}
To begin with, we compare \system to the two state-of-the-art (SOTA) baselines without restarts. 
Both Megatron-LM and DeepSpeed suffer from the existence of stragglers, leading to significant performance degradation. For instance, when training the 110B model under the most severe straggler situations (S5), their running time bumps by 5.52$\times$ (from 19.2 to 106.1 seconds per step) and 7.27$\times$ (from 18.2 to 132.3 seconds per step), respectively. Even under the mildest straggler situations (S1), the performance reduction is still around 2$\times$, which is quite unsatisfactory since the other 63 non-straggling GPUs become extremely under-utilized. 
In contrast, the performance reduction of \system is merely 1.05-1.16$\times$ (from 19.2 to 20.1-22.3 seconds per step). It validates that \system is resilient to various straggler situations --- by adjusting the parallelization plan adaptively, it is capable of harnessing the stragglers and thereby maintaining a high performance. Eventually, \system outperforms Megatron-LM and DeepSpeed by up to 5.30$\times$ and 6.73$\times$, respectively. And the strength of \system is consistent over the three models, provisioning 2.63-5.28$\times$ of speed up on average compared to the two baselines. 

Besides, we find that DeepSpeed is more sensitive to stragglers --- when there are no stragglers, it runs a bit faster than \system and Megatron-LM on the 70B and 110B model, whilst gradually surpassed when the straggler situation becomes more severe. This is not surprising since the ZeRO-3 optimizer needs to gather model parameters for each layer, which is globally synchronous by nature. Instead, for hybrid parallel approaches, only GPUs within the same TP group need to synchronize per layer, so the idle periods are shorter compared with DeepSpeed. 
As a result, hybrid parallel is a better fit for straggler-resilient training of large-scale models.

\mysubsubsection{Comparison to SOTA Baselines with Restarts}
We then compare \system with the two SOTA baselines with restarts, which involves removing straggling nodes and/or reintegrating recovered nodes, and restarting the training task. 
The experiment results demonstrate that, by allowing restarting, the two baselines give better training efficiency since no stragglers will be involved. 
However, \system still consistently outperforms the baselines, achieving speedups of 1.15–1.63$\times$ on average over Megatron-LM w/ Restart and 1.76–2.71$\times$ over DeepSpeed w/ Restart, respectively. 
This is reasonable as \system supports adjusting the parallelization plan at the GPU granularity, whilst the baselines can only add/remove an entire node. 
Besides, as we will shown in \S\ref{sec:expr_case_and_ablation}, \system can also make use of the stragglers by assigning fewer workloads to them, whilst the baselines need to exclude all stragglers. 
Let alone the training efficiency, the restart-based approaches suffer from two critical drawbacks:
\begin{itemize}[leftmargin=*]
\item \textbf{Significant restart overhead:} 
To restart a training task, we need to save the latest model checkpoint, re-run the training framework (which involves initialization steps like resource allocation and communication group construction), and load the model checkpoint to continue the training. 
Consequently, the restart overhead is substantial, ranging from 199.1 to 442.0 seconds for Megatron-LM and 114.9 to 231.9 seconds for DeepSpeed. 
\item \textbf{Manual intervention for optimal parallel configuration:} 
Since a few nodes are excluded or included upon each restart, we must manually tune the parallel configuration to maximize efficiency given the memory constraints and the integer division requirements of model and data, which also incurs significant time.\footnote{For instance, training the 32B model with Megatron-LM under S4 requires activation checkpointing technique to avoid out-of-memory errors, whilst training the 70B model under S2 necessitates assigning fewer layers to the first pipeline stage so that 56 GPUs can use a PP degree of 7 to partition the 80 total layers. Finding such configurations necessitates expert experiences and laborious testing. In addition, the configurations vary across different straggler situations (we have provided the detailed configurations in Supplementary Appendix~A.3), % \ref{appendix:configuration}), 
which exacerbates the problem.}
Note that such time cost is not included in Figure~\ref{fig:e2e_dynamic}, meaning that the restart overhead should be much larger in practice. 
\end{itemize}
On the contrary, \system can automatically determine the stragglers, deduce a new parallelization plan for better efficiency, and migrate the model states in real time. 
More importantly, the asynchronous re-planning mechanism perfectly hides the planning time by overlapping, so the only overhead is brought by model migration, which is negligible (around 5 seconds or shorter). 
As a result, \system is superior in handling dynamic stragglers through malleableizing the parallelization plans on the fly, improving the stability and straggler-resilience in the training of large-scale models.

\begin{figure}[!t]
\centering
\includegraphics[width=0.6\linewidth]{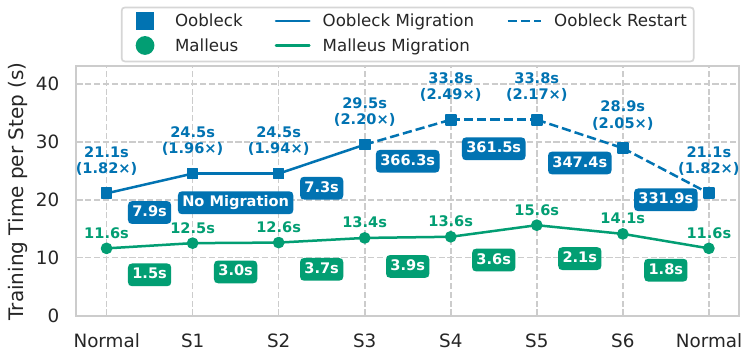}
\caption{\small{Comparison with Oobleck using the 32B model. 
The dashed lines indicate that Oobleck fails to dynamically migrate the model and needs to restart the training task. 
The training time per step is shown above the lines (values in parentheses indicate the improvement achieved by \system over Oobleck), whilst the migration/restart time cost is presented below the lines.}}
\label{fig:e2e_oobleck}
\end{figure}

\mysubsubsection{Comparison with the Fault-tolerant Baseline}
Subsequently, we compare \system with Oobleck, a fault-tolerant training framework. 
Specifically, when the straggler situation changes, Oobleck treats the real-time stragglers as faulty GPUs and performs migrations to exclude them. 
However, as shown in Figure~\ref{fig:e2e_oobleck}, Oobleck's performance is unsatisfactory. 
For one thing, its training time is 1.82-2.49$\times$ of that of \system across the straggler situations, exhibiting a slow training speed even when there are no stragglers. 
For another, it fails to migrate in many cases (S3->S4, S4->S5, S5->S6, and S6->Normal) and can only restart the entire training task to exclude the stragglers, incurring expensive overhead.

Such unsatisfactory performance is due to the fact that Oobleck trades off training efficiency for fault tolerance.
In particular, in order to be fault-tolerant, Oobleck requires specific modifications to model parallelization throughout the training process (even when there are no stragglers), which comes at the cost of degraded training efficiency.
In addition, when the straggler situation changes, Oobleck can only deploy a limited set of predefined pipeline templates, so it fails to handle various cases and must fall back to the expensive restarts. 
In contrast, \system is designed to optimize training efficiency for any straggler situations --- \system is comparable against Megatron-LM and DeepSpeed when there are no stragglers and is capable of quickly adapting to the straggler situation.

\begin{table}[!t]
\centering
\caption{\small{The ratios of time cost of training with stragglers to that of training without stragglers, where $R_{\text{actual}}$ is the actual ratio computed by the values in Table~\ref{tb:e2e}, $R_{\text{opt}}$ is the theoretically optimal ratio, and $R_{\text{est}}$ is the ratio estimated by our planning algorithm.}}
\label{tb:theoretic_ratios}
\small
\begin{tabular}{c||c||c||c|c||c|c}
\hline
\toprule
& & $R_{\text{actual}}$ & $R_{\text{opt}}$ & $1 - \frac{R_{\text{opt}}}{R_{\text{actual}}}$ & $R_{\text{est}}$ & $1 - \frac{R_{\text{est}}}{R_{\text{actual}}}$ 
\\
\hline
\multirow{5}*{\rotatebox[origin=c]{90}{\small{32B}}}
& S1 & 1.08 & 1.03 & 4.63\% & 1.06 & 1.85\%
\\
& S2 & 1.08 & 1.03 & 4.63\% & 1.06 & 1.85\%
\\
& S3 & 1.16 & 1.05 & 9.48\% & 1.13 & 2.58\%
\\
& S4 & 1.17 & 1.07 & 9.32\% & 1.18 & 0.00\%
\\
& S5 & 1.34 & 1.22 & 8.95\% & 1.37 & -2.24\%
\\
& S6 & 1.22 & 1.17 & 4.10\% & 1.20 & 1.64\%
\\
\hline
\multirow{5}*{\rotatebox[origin=c]{90}{\small{70B}}}
& S1 & 1.08 & 1.01 & 6.48\% & 1.03 & 4.63\%
\\
& S2 & 1.08 & 1.01 & 6.48\% & 1.03 & 4.63\%
\\
& S3 & 1.10 & 1.02 & 7.27\% & 1.04 & 5.45\%
\\
& S4 & 1.11 & 1.04 & 6.30\% & 1.04 & 6.30\%
\\
& S5 & 1.14 & 1.13 & 0.88\% & 1.15 & -0.88\%
\\
& S6 & 1.13 & 1.12 & 0.88\% & 1.13 & 0.00\%
\\
\hline
\multirow{5}*{\rotatebox[origin=c]{90}{\small{110B}}}
& S1 & 1.05 & 1.01 & 3.81\% & 1.03 & 1.90\%
\\
& S2 & 1.05 & 1.01 & 3.81\% & 1.03 & 1.90\%
\\
& S3 & 1.05 & 1.02 & 2.85\% & 1.05 & 0.00\%
\\
& S4 & 1.14 & 1.04 & 8.77\% & 1.08 & 3.70\%
\\
& S5 & 1.16 & 1.13 & 2.59\% & 1.17 & -0.01\%
\\
& S6 & 1.13 & 1.12 & 0.88\% & 1.13 & 0.00\%
\\
\bottomrule
\end{tabular}
\end{table}

\mysubsubsection{Comparison with Theoretic Optimum}
Suppose there are $N$ GPUs and $n$ of them are stragglers with rates $\{x_1, \cdots, x_n\}$. Theoretically speaking, if the hardware capability (e.g., TFLOPs) is inversely proportional to the straggling rates, then the optimal ratio of the time cost of running with stragglers to that of running without stragglers should be ${N}/{((N-n) + \sum_{i=1}^{n}{1 / x_i})}$. 
As shown in Table~\ref{tb:theoretic_ratios}, the performance loss of \system compared with the theoretic optimum is within 10\% under all situations and even within 5\% in more than half of the cases, verifying that the performance of \system is very close to the theoretic optimum. 
In addition, we also present the estimated performance obtained by our planning algorithm (i.e., via the solution to Eq.~\eqref{eq:lower_problem_one_case_origin}) in Table~\ref{tb:theoretic_ratios}, which shows that our cost model is accurate --- the estimated errors are not higher than 6.3\% in all experiments. Undoubtedly, this is vital to the deduction of parallelization plans.

\begin{table}[!t]
\centering
\caption{\small{Case studies of parallelization plans. 
Straggling GPUs are highlighted in red. Groups after splitting or containing stragglers are highlighted in blue background color.}}
\label{tb:case_study}
\small
\label{tb:case_studies}
\begin{tabular}{|c|c|c|c|c|}
\hline
\multicolumn{5}{|c|}{110B under S4 (\textcolor{red}{$x_{0}$} $ = 5.42$, \textcolor{red}{$x_{8}$} $ = 3.75$,
\textcolor{red}{$x_{16}$} $ = 2.57$)} \\
% \hline
\hhline{|*5-}
\multirow{4}{*}{\specialcell{$m_1 = 33$\\(8 stages)}} &
\cellcolor{blue!10}{$x_7$} &
\cellcolor{blue!10}{$x_{15}$} &
\cellcolor{blue!10}{$x_{23}$} &
\cellcolor{blue!10}{$x_{1} \sim x_{4}$}
\\
& 
\cellcolor{blue!10}{$l_{1,1} = 2$} & 
\cellcolor{blue!10}{$l_{1,2} = 2$} & 
\cellcolor{blue!10}{$l_{1,3} = 2$} & 
\cellcolor{blue!10}{$l_{1,4} = 10$}
\\
% \cline{2-5}
\hhline{~|*4-}
& 
\cellcolor{blue!10}{$x_{9} \sim x_{12}$} &
\cellcolor{blue!10}{$x_{17} \sim x_{20}$} & 
$x_{40} \sim x_{47}$ &
$x_{56} \sim x_{63}$
\\
&
\cellcolor{blue!10}{$l_{1,5} = 11$} & 
\cellcolor{blue!10}{$l_{1,6} = 11$} &
$l_{1,7} = 21$ &
$l_{1,8} = 21$ 
\\
% \cline{1-5}
\hhline{|*5-}
\multirow{4}{*}{\specialcell{$m_2 = 31$\\(6 stages)}} &
\cellcolor{blue!10}{$x_{5}, x_{6}$} & 
\cellcolor{blue!10}{$x_{13}, x_{14}$} & 
\cellcolor{blue!10}{$x_{21}, x_{22}$} 
\\
& 
\cellcolor{blue!10}{$l_{2,1} = 4$} & 
\cellcolor{blue!10}{$l_{2,2} = 5$} & 
\cellcolor{blue!10}{$l_{2,3} = 5$} 
\\
% \cline{2-4}
\hhline{~|*3-}
& 
$x_{32} \sim x_{39}$ &
$x_{48} \sim x_{55}$ &
$x_{24} \sim x_{31}$
\\
& 
$l_{2,4} = 22$ & 
$l_{2,5} = 22$ & 
$l_{2,6} = 22$
\\
% \cline{1-4}
\hhline{|*4-}
\end{tabular}

\medskip

\begin{tabular}{|c|c|c|c|c|}
\hline
\multicolumn{5}{|c|}{32B under S5 (\textcolor{red}{$x_{0} \sim x_{7}$} $ = 2.62$, \textcolor{red}{$x_{8}$} $ = 3.8$)} \\
% \hline
\hhline{|*5-}
\multirow{2}{*}{\specialcell{$m_1 = 7$\\(4 stages)}} &
\cellcolor{blue!10}{\textcolor{red}{$x_2$}, \textcolor{red}{$x_1$}} &
\cellcolor{blue!10}{\textcolor{red}{$x_4$}, \textcolor{red}{$x_3$}} &
\cellcolor{blue!10}{\textcolor{red}{$x_0$}, \textcolor{red}{$x_5$}} & 
\cellcolor{blue!10}{$x_{15}$} 
\\
& 
\cellcolor{blue!10}{$l_{1,1} = 15$} & 
\cellcolor{blue!10}{$l_{1,2} = 17$} & 
\cellcolor{blue!10}{$l_{1,3} = 17$} & 
\cellcolor{blue!10}{$l_{1,4} = 11$} 
\\
% \cline{1-5}
\hhline{|*5-}
\multirow{2}{*}{\specialcell{$m_2 = 17$\\(4 stages)}} &
\cellcolor{blue!10}{\textcolor{red}{$x_6$}, \textcolor{red}{$x_7$}} &
$x_{20} \sim x_{21}$ &
$x_{26} \sim x_{27}$ & 
$x_{10} \sim x_{9}$ 
\\
& 
\cellcolor{blue!10}{$l_{2,1} = 7$} & 
$l_{2,2} = 17$ & 
$l_{2,3} = 18$ & 
$l_{2,4} = 18$ 
\\
% \cline{1-5}
\hhline{|*5-}
\multirow{2}{*}{\specialcell{$m_3 = 20$\\(4 stages)}} &
$x_{16} \sim x_{17}$ &
$x_{22} \sim x_{23}$ &
$x_{28} \sim x_{29}$ & 
$x_{12} \sim x_{11}$ 
\\
& 
$l_{3,1} = 15$ & 
$l_{3,2} = 15$ & 
$l_{3,3} = 15$ & 
$l_{3,4} = 15$ 
\\
% \cline{1-5}
\hhline{|*5-}
\multirow{2}{*}{\specialcell{$m_4 = 20$\\(4 stages)}} &
$x_{18} \sim x_{19}$ &
$x_{24} \sim x_{25}$ &
$x_{30} \sim x_{31}$ & 
$x_{14} \sim x_{13}$ 
\\
& 
$l_{4,1} = 15$ & 
$l_{4,2} = 15$ & 
$l_{4,3} = 15$ & 
$l_{4,4} = 15$ 
\\
% \cline{1-5}
\hhline{|*5-}
\end{tabular}
\end{table}

\subsection{Case Studies and Ablation Studies}
\label{sec:expr_case_and_ablation}

\mysubsubsection{Case Studies}
Table~\ref{tb:case_study} presents two parallelization plans discovered by \system. When training the 110B model under the S4 situation, \system eliminates stragglers on all three nodes and forms new GPU groups with 1, 2, and 4 GPUs, respectively, summing up to 9 groups in total. \system distributes these 9 groups (in different sizes) to two pipelines to achieve load balancing. When training the 32B model under the S5 situation, \system eliminates the level-2 straggler on the second node and retains all level-1 stragglers on the first node. By assigning these stragglers fewer layers ($x_{6}, x_{7}$) and less data ($x_{0} \sim x_{5}$), the overall training time is minimized.

\mysubsubsection{Ablation Studies}
As introduced in \S\ref{sec:system_plan_design}, our parallelization features four types of non-uniform partitioning. 
We assess the effectiveness of them using the 110B model. To better demonstrate the agility of \system in dealing with complex straggler situations, we further introduce a severe straggler in level-8 (i.e., running 8 extra processes). Depicted in Figure~\ref{fig:ablation}, we experiment with three stragglers in level-1, level-3, and level-8, appearing on one node, two nodes, and three nodes, respectively. 

When the stragglers are on only one node, we find that the non-uniform partitioning of layers and data (solved by the lower-level problem) can greatly alleviate the impact of stragglers, with a gap of only about 10\% from the theoretic optimum. For instance, our planning algorithm strikes a good balance by strategically assigning only 2 layers to the slowest group (containing all three stragglers) whilst evenly distributing the remaining 78 layers to the other three non-straggling groups in the same pipeline. 

However, when the stragglers appear on multiple nodes, adjusting the partitioning of layers and data alone is no longer sufficient to produce satisfactory outcomes, with a significant gap of 20-40\% from the theoretic optimum. At this time, the introduction of non-uniform partitioning of devices and stages (solved by the upper-level problem) becomes particularly important. By isolating the stragglers and orchestrating pipelines in diverse forms, we can make better use of the non-straggling GPUs, reducing the gap from the theoretic optimum to at most 8.7\%.

\begin{figure*}[!t]
\centering
\includegraphics[width=\textwidth]{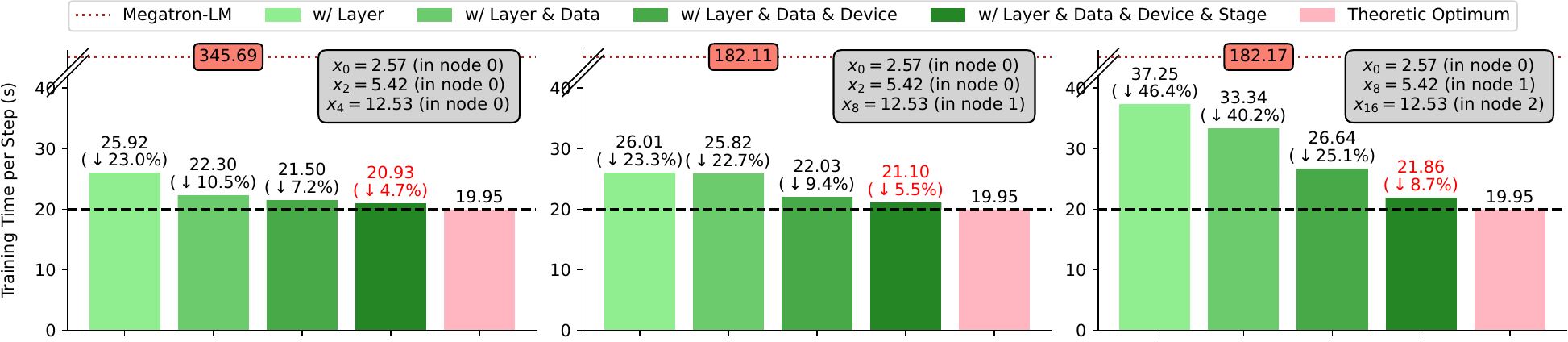}
\caption{\small{Effectiveness of the non-uniform partitioning on each dimension, evaluated on the 110B model. The ratios represents the gap from the theoretic optimum, computed by $1 - T_{\text{opt}}/T_{\text{actual}}$, where $T_{\text{opt}}$ and $T_{\text{actual}}$ denote the theoretic optimal time and actual running time.}}
\label{fig:ablation}
\end{figure*}

\mysubsubsection{More Experiment Results}
Due to the space constraint, we put more experiment results and analysis in Supplementary Appendix~A, % \ref{appendix:more_expr}, 
including the effectiveness of our cost model, the scalability of our planning algorithm, and the detailed parallel configurations. Please refer to our supplementary material for more details.

\section{Related Works}
\label{sec:related}

\mysubsubsection{Heterogeneous Training}
Due to the rising concerns of the GPU shortage problem~\cite{skypilot,gpu_shortage_cross_region_cloud}, several approaches have been developed for distributed training over heterogeneous types of GPUs~\cite{acc_par,whale,hap,amp_hetero_model_parallel,hethub}. 
Our work differs from them for two reasons. Firstly, they focus on static heterogeneity --- the efficiency differences among GPUs do not change during training. In contrast, we consider a more complex scenario of dynamic stragglers, where the efficiency variation is dynamic and unforeseeable. Secondly, they assume the GPUs within the same node are of the same type and thereby provision identical hardware efficiency, whilst our work addresses the straggler problem at the nuanced, per-GPU granularity.

\mysubsubsection{Elastic Training}
Elastic training is an essential technique to handle failures or device defectiveness. Checkpointing is the most common choice for elastic training. For instance, TorchElastic~\cite{torch_elastic} and HorovodElastic~\cite{horovod_elastic} support restarting and loading the model checkpoints at failure. Gemini~\cite{gemini_failure_recovery} studies how to accelerate failure recovery. 
There are also approaches that focus on resilience when the failures are informed or detected. For instance, Varuna~\cite{varuna} and Bamboo~\cite{bamboo} consider training over spot instances (a kind of preemptible cloud instances), 
whilst Oobleck~\cite{oobleck} and Recycle~\cite{recycle} detect the failures in dedicated clusters. These approaches typically re-configure the training task before or upon failures. However, they are orthogonal to our work since straggler mitigation is not their focus. 
Elasticity is also important in job scheduling to improve cluster usages~\cite{pollux_goodput_scheduling,coddl_elastic_sharing,lyra_elastic_scheduling,heet_elastic_scheduling_in_hetero}, yet these works mainly focus on data parallel and do not consider the straggler-resilience of a single job. 

Besides, these approaches primarily tackle node-level failures and recoveries, rather than considering the fine-grained removal or addition of individual GPUs. In contrast, \system supports elasticity at the GPU granularity.
It is noteworthy that node failure problem is essentially a subset of GPU failure problem (i.e., the simultaneous failure of 8 GPUs means a node failure), and the GPU failure problem could be regarded as a special case of straggler problem as well (i.e., resulting in a straggling rate of infinity). 
The parallelization planning process used in \system can be effectively applied to these scenarios by simply setting the straggling rate of completely failed GPU(s) to infinity. And thanks to the GPU-granular elasticity offered by \system, we can handle the failure of individual GPUs, whilst existing approaches need to remove/replace the entire node. Hence, we believe that \system addresses a broader set of problems compared to previous elastic training approaches.

\mysubsubsection{Straggler Mitigation}
Relaxing the synchronization protocol in data parallel for straggler mitigation has been explored for long. This line of research breaks the barrier of gradient/model synchronization to reduce idle periods of non-straggling devices. 
Notable efforts include asynchronous parallel~\cite{dist_belief,project_adam} and stale synchronous parallel~\cite{ssp_ps,hetero_ps,dynamic_ssp,partial_reduce}. 
However, these works are developed for data parallel, and inevitably impact model convergence. 
On the contrary, our work focuses on straggler-resilient hybrid parallel training of large-scale models without affecting model convergence.

\nocite{ai_sys_for_llm_training_survey}

There is also a concurrent work, Falcon~\cite{falcon_straggler_hybrid_parallel}, that considers straggler mitigation in hybrid parallel training, confirming the importance of our target scenario. 
However, Falcon focuses on revealing the impact of stragglers and how to pinpoint them. It only considers several heuristic methods for straggler mitigation. 
In contrast, our work addresses the dynamic straggler problem at the nuanced, per-GPU granularity, formulates an optimization problem, and tackles it through a series of non-uniform partitioning, which are not considered in Falcon.

\section{Conclusion}
\label{sec:conc}

This work focuses on the straggler-resilience of hybrid parallel training for large-scale models. Specifically, we introduced \system, a hybrid parallel training framework that captures the dynamic stragglers at the nuanced, per-GPU granularity. We developed a parallelization planning algorithm that co-optimizes the non-uniform partitioning of GPU devices, pipeline stages, model layers, and training data, speeding up the training under various straggler situations. In response to the dynamicity in stragglers, we proposed a re-planning process that adjusts the parallelization plan and migrates model states on the fly. Empirical results show that \system can be on average 2.63-5.28$\times$ faster than existing parallel training frameworks and handle the straggler dynamicity efficiently.

\begin{acks}
This work is supported by National Science and Technology Major Project (2022ZD0116315), National Natural Science Foundation of China (U23B2048, 62402011), Beijing Municipal Science and Technology Project (Z231100010323002), China National Postdoctoral Program for Innovative Talents (BX20230012), China Postdoctoral Science Foundation (2024M750103), Beijing Natural Science Foundation (4244080), research grant No. SH-2024JK29, PKU-Tencent joint research Lab, and High-performance Computing Platform of Peking University. 
Fangcheng Fu and Bin Cui are the corresponding authors.
\end{acks}

\bibliographystyle{ACM-Reference-Format}
\bibliography{reference}


\begin{thebibliography}{95}

%%% the Web address.

\ifx \showCODEN    \undefined \def \showCODEN     #1{\unskip}     \fi
\ifx \showISBNx    \undefined \def \showISBNx     #1{\unskip}     \fi
\ifx \showISBNxiii \undefined \def \showISBNxiii  #1{\unskip}     \fi
\ifx \showISSN     \undefined \def \showISSN      #1{\unskip}     \fi
\ifx \showLCCN     \undefined \def \showLCCN      #1{\unskip}     \fi
\ifx \shownote     \undefined \def \shownote      #1{#1}          \fi
\ifx \showarticletitle \undefined \def \showarticletitle #1{#1}   \fi
\ifx \showURL      \undefined \def \showURL       {\relax}        \fi
% invisible to TeX
\providecommand\bibfield[2]{#2}
\providecommand\bibinfo[2]{#2}
\providecommand\natexlab[1]{#1}
\providecommand\showeprint[2][]{arXiv:#2}

\bibitem[pul(2009)]%
        {pulp}
 \bibinfo{year}{2009}\natexlab{}.
\newblock \bibinfo{title}{Optimization with PuLP}.
\newblock
\newblock
\shownote{\url{https://coin-or.github.io/pulp/}}.


\bibitem[hor(2019)]%
        {horovod_elastic}
 \bibinfo{year}{2019}\natexlab{}.
\newblock \bibinfo{title}{Elastic Horovod}.
\newblock
\newblock
\shownote{\url{https://horovod.readthedocs.io/en/latest/elastic_include.html}}.


\bibitem[tor(2023)]%
        {torch_elastic}
 \bibinfo{year}{2023}\natexlab{}.
\newblock \bibinfo{title}{Torch Distributed Elastic}.
\newblock
\newblock
\shownote{\url{https://pytorch.org/docs/stable/distributed.elastic.html}}.


\bibitem[Athlur et~al\mbox{.}(2022)]%
        {varuna}
\bibfield{author}{\bibinfo{person}{Sanjith Athlur}, \bibinfo{person}{Nitika Saran}, \bibinfo{person}{Muthian Sivathanu}, \bibinfo{person}{Ramachandran Ramjee}, {and} \bibinfo{person}{Nipun Kwatra}.} \bibinfo{year}{2022}\natexlab{}.
\newblock \showarticletitle{Varuna: scalable, low-cost training of massive deep learning models}. In \bibinfo{booktitle}{\emph{Proceedings of the Seventeenth European Conference on Computer Systems (EuroSys 2022)}}. \bibinfo{pages}{472--487}.
\newblock


\bibitem[Brown et~al\mbox{.}(2020)]%
        {gpt3}
\bibfield{author}{\bibinfo{person}{Tom~B. Brown}, \bibinfo{person}{Benjamin Mann}, \bibinfo{person}{Nick Ryder}, \bibinfo{person}{Melanie Subbiah}, \bibinfo{person}{Jared Kaplan}, \bibinfo{person}{Prafulla Dhariwal}, \bibinfo{person}{Arvind Neelakantan}, \bibinfo{person}{Pranav Shyam}, \bibinfo{person}{Girish Sastry}, \bibinfo{person}{Amanda Askell}, \bibinfo{person}{Sandhini Agarwal}, \bibinfo{person}{Ariel Herbert{-}Voss}, \bibinfo{person}{Gretchen Krueger}, \bibinfo{person}{Tom Henighan}, \bibinfo{person}{Rewon Child}, \bibinfo{person}{Aditya Ramesh}, \bibinfo{person}{Daniel~M. Ziegler}, \bibinfo{person}{Jeffrey Wu}, \bibinfo{person}{Clemens Winter}, \bibinfo{person}{Christopher Hesse}, \bibinfo{person}{Mark Chen}, \bibinfo{person}{Eric Sigler}, \bibinfo{person}{Mateusz Litwin}, \bibinfo{person}{Scott Gray}, \bibinfo{person}{Benjamin Chess}, \bibinfo{person}{Jack Clark}, \bibinfo{person}{Christopher Berner}, \bibinfo{person}{Sam McCandlish}, \bibinfo{person}{Alec Radford}, \bibinfo{person}{Ilya Sutskever},
  {and} \bibinfo{person}{Dario Amodei}.} \bibinfo{year}{2020}\natexlab{}.
\newblock \showarticletitle{Language Models are Few-Shot Learners}. In \bibinfo{booktitle}{\emph{Annual Conference on Neural Information Processing Systems 2020 (NeurIPS 2020)}}.
\newblock


\bibitem[Bynum et~al\mbox{.}(2021)]%
        {pyomo}
\bibfield{author}{\bibinfo{person}{Michael~L Bynum}, \bibinfo{person}{Gabriel~A Hackebeil}, \bibinfo{person}{William~E Hart}, \bibinfo{person}{Carl~D Laird}, \bibinfo{person}{Bethany~L Nicholson}, \bibinfo{person}{John~D Siirola}, \bibinfo{person}{Jean-Paul Watson}, \bibinfo{person}{David~L Woodruff}, {et~al\mbox{.}}} \bibinfo{year}{2021}\natexlab{}.
\newblock \bibinfo{booktitle}{\emph{Pyomo-optimization modeling in python}}. Vol.~\bibinfo{volume}{67}.
\newblock \bibinfo{publisher}{Springer}.
\newblock


\bibitem[Chai et~al\mbox{.}(2023)]%
        {dm4ml_survey}
\bibfield{author}{\bibinfo{person}{Chengliang Chai}, \bibinfo{person}{Jiayi Wang}, \bibinfo{person}{Yuyu Luo}, \bibinfo{person}{Zeping Niu}, {and} \bibinfo{person}{Guoliang Li}.} \bibinfo{year}{2023}\natexlab{}.
\newblock \showarticletitle{Data Management for Machine Learning: {A} Survey}.
\newblock \bibinfo{journal}{\emph{{IEEE} Trans. Knowl. Data Eng.}} \bibinfo{volume}{35}, \bibinfo{number}{5} (\bibinfo{year}{2023}), \bibinfo{pages}{4646--4667}.
\newblock


\bibitem[Chang et~al\mbox{.}(2024)]%
        {Biathlon}
\bibfield{author}{\bibinfo{person}{Chaokun Chang}, \bibinfo{person}{Eric Lo}, {and} \bibinfo{person}{Chunxiao Ye}.} \bibinfo{year}{2024}\natexlab{}.
\newblock \showarticletitle{Biathlon: Harnessing Model Resilience for Accelerating {ML} Inference Pipelines}.
\newblock \bibinfo{journal}{\emph{Proc. {VLDB} Endow.}} \bibinfo{volume}{17}, \bibinfo{number}{10} (\bibinfo{year}{2024}), \bibinfo{pages}{2631--2640}.
\newblock


\bibitem[Chilimbi et~al\mbox{.}(2014)]%
        {project_adam}
\bibfield{author}{\bibinfo{person}{Trishul~M. Chilimbi}, \bibinfo{person}{Yutaka Suzue}, \bibinfo{person}{Johnson Apacible}, {and} \bibinfo{person}{Karthik Kalyanaraman}.} \bibinfo{year}{2014}\natexlab{}.
\newblock \showarticletitle{Project Adam: Building an Efficient and Scalable Deep Learning Training System}. In \bibinfo{booktitle}{\emph{11th {USENIX} Symposium on Operating Systems Design and Implementation (OSDI 2014)}}. \bibinfo{publisher}{{USENIX} Association}, \bibinfo{pages}{571--582}.
\newblock


\bibitem[Dao(2024)]%
        {flash_attn_v2}
\bibfield{author}{\bibinfo{person}{Tri Dao}.} \bibinfo{year}{2024}\natexlab{}.
\newblock \showarticletitle{FlashAttention-2: Faster Attention with Better Parallelism and Work Partitioning}. In \bibinfo{booktitle}{\emph{International Conference on Learning Representations 2024 (ICLR 2024)}}.
\newblock


\bibitem[Dao et~al\mbox{.}(2022)]%
        {flash_attn}
\bibfield{author}{\bibinfo{person}{Tri Dao}, \bibinfo{person}{Daniel~Y. Fu}, \bibinfo{person}{Stefano Ermon}, \bibinfo{person}{Atri Rudra}, {and} \bibinfo{person}{Christopher R{\'{e}}}.} \bibinfo{year}{2022}\natexlab{}.
\newblock \showarticletitle{FlashAttention: Fast and Memory-Efficient Exact Attention with IO-Awareness}. In \bibinfo{booktitle}{\emph{Annual Conference on Neural Information Processing Systems 2022 (NeurIPS 2022)}}.
\newblock


\bibitem[Dean et~al\mbox{.}(2012)]%
        {dist_belief}
\bibfield{author}{\bibinfo{person}{Jeffrey Dean}, \bibinfo{person}{Greg Corrado}, \bibinfo{person}{Rajat Monga}, \bibinfo{person}{Kai Chen}, \bibinfo{person}{Matthieu Devin}, \bibinfo{person}{Quoc~V. Le}, \bibinfo{person}{Mark~Z. Mao}, \bibinfo{person}{Marc'Aurelio Ranzato}, \bibinfo{person}{Andrew~W. Senior}, \bibinfo{person}{Paul~A. Tucker}, \bibinfo{person}{Ke Yang}, {and} \bibinfo{person}{Andrew~Y. Ng}.} \bibinfo{year}{2012}\natexlab{}.
\newblock \showarticletitle{Large Scale Distributed Deep Networks}. In \bibinfo{booktitle}{\emph{26th Annual Conference on Neural Information Processing Systems 2012 (NeurIPS 2022)}}. \bibinfo{pages}{1232--1240}.
\newblock


\bibitem[Dubey et~al\mbox{.}(2024)]%
        {llama3}
\bibfield{author}{\bibinfo{person}{Abhimanyu Dubey}, \bibinfo{person}{Abhinav Jauhri}, {et~al\mbox{.}}} \bibinfo{year}{2024}\natexlab{}.
\newblock \showarticletitle{The Llama 3 Herd of Models}.
\newblock \bibinfo{journal}{\emph{CoRR}}  \bibinfo{volume}{abs/2407.21783} (\bibinfo{year}{2024}).
\newblock


\bibitem[Erben et~al\mbox{.}(2024)]%
        {dl_across_clouds_ea}
\bibfield{author}{\bibinfo{person}{Alexander Erben}, \bibinfo{person}{Ruben Mayer}, {and} \bibinfo{person}{Hans{-}Arno Jacobsen}.} \bibinfo{year}{2024}\natexlab{}.
\newblock \showarticletitle{How Can We Train Deep Learning Models Across Clouds and Continents? An Experimental Study}.
\newblock \bibinfo{journal}{\emph{Proc. {VLDB} Endow.}} \bibinfo{volume}{17}, \bibinfo{number}{6} (\bibinfo{year}{2024}), \bibinfo{pages}{1214--1226}.
\newblock


\bibitem[Feuer et~al\mbox{.}(2024)]%
        {archetype}
\bibfield{author}{\bibinfo{person}{Benjamin Feuer}, \bibinfo{person}{Yurong Liu}, \bibinfo{person}{Chinmay Hegde}, {and} \bibinfo{person}{Juliana Freire}.} \bibinfo{year}{2024}\natexlab{}.
\newblock \showarticletitle{ArcheType: {A} Novel Framework for Open-Source Column Type Annotation using Large Language Models}.
\newblock \bibinfo{journal}{\emph{Proc. {VLDB} Endow.}} \bibinfo{volume}{17}, \bibinfo{number}{9} (\bibinfo{year}{2024}), \bibinfo{pages}{2279--2292}.
\newblock


\bibitem[Fu et~al\mbox{.}(2020)]%
        {tinyscript}
\bibfield{author}{\bibinfo{person}{Fangcheng Fu}, \bibinfo{person}{Yuzheng Hu}, \bibinfo{person}{Yihan He}, \bibinfo{person}{Jiawei Jiang}, \bibinfo{person}{Yingxia Shao}, \bibinfo{person}{Ce Zhang}, {and} \bibinfo{person}{Bin Cui}.} \bibinfo{year}{2020}\natexlab{}.
\newblock \showarticletitle{Don't Waste Your Bits! Squeeze Activations and Gradients for Deep Neural Networks via TinyScript}. In \bibinfo{booktitle}{\emph{International Conference on Machine Learning 2020 (ICML 2020)}}. \bibinfo{pages}{3304--3314}.
\newblock


\bibitem[Gan et~al\mbox{.}(2021)]%
        {bagua}
\bibfield{author}{\bibinfo{person}{Shaoduo Gan}, \bibinfo{person}{Xiangru Lian}, \bibinfo{person}{Rui Wang}, \bibinfo{person}{Jianbin Chang}, \bibinfo{person}{Chengjun Liu}, \bibinfo{person}{Hongmei Shi}, \bibinfo{person}{Shengzhuo Zhang}, \bibinfo{person}{Xianghong Li}, \bibinfo{person}{Tengxu Sun}, \bibinfo{person}{Jiawei Jiang}, \bibinfo{person}{Binhang Yuan}, \bibinfo{person}{Sen Yang}, \bibinfo{person}{Ji Liu}, {and} \bibinfo{person}{Ce Zhang}.} \bibinfo{year}{2021}\natexlab{}.
\newblock \showarticletitle{{BAGUA:} Scaling up Distributed Learning with System Relaxations}.
\newblock \bibinfo{journal}{\emph{Proc. {VLDB} Endow.}} \bibinfo{volume}{15}, \bibinfo{number}{4} (\bibinfo{year}{2021}), \bibinfo{pages}{804--813}.
\newblock


\bibitem[Gandhi et~al\mbox{.}(2024)]%
        {recycle}
\bibfield{author}{\bibinfo{person}{Swapnil Gandhi}, \bibinfo{person}{Mark Zhao}, \bibinfo{person}{Athinagoras Skiadopoulos}, {and} \bibinfo{person}{Christos Kozyrakis}.} \bibinfo{year}{2024}\natexlab{}.
\newblock \showarticletitle{ReCycle: Resilient Training of Large DNNs using Pipeline Adaptation}. In \bibinfo{booktitle}{\emph{Proceedings of the 30th {ACM} Symposium on Operating Systems Principles (SOSP 2024)}}. \bibinfo{pages}{211--228}.
\newblock


\bibitem[Ge et~al\mbox{.}(2024)]%
        {hotspa}
\bibfield{author}{\bibinfo{person}{Hao Ge}, \bibinfo{person}{Fangcheng Fu}, \bibinfo{person}{Haoyang Li}, \bibinfo{person}{Xuanyu Wang}, \bibinfo{person}{Sheng Lin}, \bibinfo{person}{Yujie Wang}, \bibinfo{person}{Xiaonan Nie}, \bibinfo{person}{Hailin Zhang}, \bibinfo{person}{Xupeng Miao}, {and} \bibinfo{person}{Bin Cui}.} \bibinfo{year}{2024}\natexlab{}.
\newblock \showarticletitle{Enabling Parallelism Hot Switching for Efficient Training of Large Language Models}. In \bibinfo{booktitle}{\emph{Proceedings of the 30th {ACM} Symposium on Operating Systems Principles (SOSP 2024)}}. \bibinfo{pages}{178--194}.
\newblock


\bibitem[Guan et~al\mbox{.}(2024)]%
        {pipeline_survey}
\bibfield{author}{\bibinfo{person}{Lei Guan}, \bibinfo{person}{Dong{-}Sheng Li}, \bibinfo{person}{Jiye Liang}, \bibinfo{person}{Wen{-}Jian Wang}, \bibinfo{person}{Ke{-}shi Ge}, {and} \bibinfo{person}{Xicheng Lu}.} \bibinfo{year}{2024}\natexlab{}.
\newblock \showarticletitle{Advances of Pipeline Model Parallelism for Deep Learning Training: An Overview}.
\newblock \bibinfo{journal}{\emph{J. Comput. Sci. Technol.}} (\bibinfo{year}{2024}).
\newblock


\bibitem[Harlap et~al\mbox{.}(2016)]%
        {flexrr}
\bibfield{author}{\bibinfo{person}{Aaron Harlap}, \bibinfo{person}{Henggang Cui}, \bibinfo{person}{Wei Dai}, \bibinfo{person}{Jinliang Wei}, \bibinfo{person}{Gregory~R. Ganger}, \bibinfo{person}{Phillip~B. Gibbons}, \bibinfo{person}{Garth~A. Gibson}, {and} \bibinfo{person}{Eric~P. Xing}.} \bibinfo{year}{2016}\natexlab{}.
\newblock \showarticletitle{Addressing the straggler problem for iterative convergent parallel {ML}}. In \bibinfo{booktitle}{\emph{Proceedings of the Seventh {ACM} Symposium on Cloud Computing (SOCC 2016)}}. \bibinfo{pages}{98--111}.
\newblock


\bibitem[Ho et~al\mbox{.}(2013)]%
        {ssp_ps}
\bibfield{author}{\bibinfo{person}{Qirong Ho}, \bibinfo{person}{James Cipar}, \bibinfo{person}{Henggang Cui}, \bibinfo{person}{Seunghak Lee}, \bibinfo{person}{Jin~Kyu Kim}, \bibinfo{person}{Phillip~B. Gibbons}, \bibinfo{person}{Garth~A. Gibson}, \bibinfo{person}{Gregory~R. Ganger}, {and} \bibinfo{person}{Eric~P. Xing}.} \bibinfo{year}{2013}\natexlab{}.
\newblock \showarticletitle{More Effective Distributed {ML} via a Stale Synchronous Parallel Parameter Server}. In \bibinfo{booktitle}{\emph{Annual Conference on Neural Informatio Processing Systems 2013 (NeurIPS 2013)}}. \bibinfo{pages}{1223--1231}.
\newblock


\bibitem[Huang et~al\mbox{.}(2019)]%
        {gpipe}
\bibfield{author}{\bibinfo{person}{Yanping Huang}, \bibinfo{person}{Youlong Cheng}, \bibinfo{person}{Ankur Bapna}, \bibinfo{person}{Orhan Firat}, \bibinfo{person}{Dehao Chen}, \bibinfo{person}{Mia~Xu Chen}, \bibinfo{person}{HyoukJoong Lee}, \bibinfo{person}{Jiquan Ngiam}, \bibinfo{person}{Quoc~V. Le}, \bibinfo{person}{Yonghui Wu}, {and} \bibinfo{person}{Zhifeng Chen}.} \bibinfo{year}{2019}\natexlab{}.
\newblock \showarticletitle{GPipe: Efficient Training of Giant Neural Networks using Pipeline Parallelism}. In \bibinfo{booktitle}{\emph{Annual Conference on Neural Information Processing Systems 2019 (NeurIPS 2019)}}. \bibinfo{pages}{103--112}.
\newblock


\bibitem[Hwang et~al\mbox{.}(2021)]%
        {coddl_elastic_sharing}
\bibfield{author}{\bibinfo{person}{Changho Hwang}, \bibinfo{person}{Taehyun Kim}, \bibinfo{person}{Sunghyun Kim}, \bibinfo{person}{Jinwoo Shin}, {and} \bibinfo{person}{KyoungSoo Park}.} \bibinfo{year}{2021}\natexlab{}.
\newblock \showarticletitle{Elastic Resource Sharing for Distributed Deep Learning}. In \bibinfo{booktitle}{\emph{18th {USENIX} Symposium on Networked Systems Design and Implementation (NSDI 2021)}}. \bibinfo{pages}{721--739}.
\newblock


\bibitem[Isenko et~al\mbox{.}(2022)]%
        {dl_preprocess_pipeline}
\bibfield{author}{\bibinfo{person}{Alexander Isenko}, \bibinfo{person}{Ruben Mayer}, \bibinfo{person}{Jeffrey Jedele}, {and} \bibinfo{person}{Hans{-}Arno Jacobsen}.} \bibinfo{year}{2022}\natexlab{}.
\newblock \showarticletitle{Where Is My Training Bottleneck? Hidden Trade-Offs in Deep Learning Preprocessing Pipelines}. In \bibinfo{booktitle}{\emph{Proceedings of the 2022 ACM International Conference on Management of Data (SIGMOD 2022)}}. \bibinfo{publisher}{{ACM}}, \bibinfo{pages}{1825--1839}.
\newblock


\bibitem[Jang et~al\mbox{.}(2023)]%
        {oobleck}
\bibfield{author}{\bibinfo{person}{Insu Jang}, \bibinfo{person}{Zhenning Yang}, \bibinfo{person}{Zhen Zhang}, \bibinfo{person}{Xin Jin}, {and} \bibinfo{person}{Mosharaf Chowdhury}.} \bibinfo{year}{2023}\natexlab{}.
\newblock \showarticletitle{Oobleck: Resilient Distributed Training of Large Models Using Pipeline Templates}. In \bibinfo{booktitle}{\emph{Proceedings of the 29th {ACM} Symposium on Operating Systems Principles (SOSP 2023)}}. \bibinfo{pages}{382--395}.
\newblock


\bibitem[Jia et~al\mbox{.}(2022)]%
        {whale}
\bibfield{author}{\bibinfo{person}{Xianyan Jia}, \bibinfo{person}{Le Jiang}, \bibinfo{person}{Ang Wang}, \bibinfo{person}{Wencong Xiao}, \bibinfo{person}{Ziji Shi}, \bibinfo{person}{Jie Zhang}, \bibinfo{person}{Xinyuan Li}, \bibinfo{person}{Langshi Chen}, \bibinfo{person}{Yong Li}, \bibinfo{person}{Zhen Zheng}, \bibinfo{person}{Xiaoyong Liu}, {and} \bibinfo{person}{Wei Lin}.} \bibinfo{year}{2022}\natexlab{}.
\newblock \showarticletitle{Whale: Efficient Giant Model Training over Heterogeneous GPUs}. In \bibinfo{booktitle}{\emph{2022 {USENIX} Annual Technical Conference ({ATC} 2022)}}. \bibinfo{pages}{673--688}.
\newblock


\bibitem[Jia et~al\mbox{.}(2019)]%
        {flexflow_soap}
\bibfield{author}{\bibinfo{person}{Zhihao Jia}, \bibinfo{person}{Matei Zaharia}, {and} \bibinfo{person}{Alex Aiken}.} \bibinfo{year}{2019}\natexlab{}.
\newblock \showarticletitle{Beyond Data and Model Parallelism for Deep Neural Networks}. In \bibinfo{booktitle}{\emph{Proceedings of Machine Learning and Systems 2019 (MLSys 2019)}}.
\newblock


\bibitem[Jiang et~al\mbox{.}(2017)]%
        {hetero_ps}
\bibfield{author}{\bibinfo{person}{Jiawei Jiang}, \bibinfo{person}{Bin Cui}, \bibinfo{person}{Ce Zhang}, {and} \bibinfo{person}{Lele Yu}.} \bibinfo{year}{2017}\natexlab{}.
\newblock \showarticletitle{Heterogeneity-aware Distributed Parameter Servers}. In \bibinfo{booktitle}{\emph{Proceedings of the 2017 ACM International Conference on Management of Data (SIGMOD 2017)}}. \bibinfo{pages}{463--478}.
\newblock


\bibitem[Jiang et~al\mbox{.}(2018)]%
        {SketchML}
\bibfield{author}{\bibinfo{person}{Jiawei Jiang}, \bibinfo{person}{Fangcheng Fu}, \bibinfo{person}{Tong Yang}, {and} \bibinfo{person}{Bin Cui}.} \bibinfo{year}{2018}\natexlab{}.
\newblock \showarticletitle{SketchML: Accelerating Distributed Machine Learning with Data Sketches}. In \bibinfo{booktitle}{\emph{Proceedings of the 2018 ACM International Conference on Management of Data (SIGMOD 2018)}}. \bibinfo{publisher}{{ACM}}, \bibinfo{pages}{1269--1284}.
\newblock


\bibitem[Jiang et~al\mbox{.}(2023)]%
        {osdp}
\bibfield{author}{\bibinfo{person}{Youhe Jiang}, \bibinfo{person}{Fangcheng Fu}, \bibinfo{person}{Xupeng Miao}, \bibinfo{person}{Xiaonan Nie}, {and} \bibinfo{person}{Bin Cui}.} \bibinfo{year}{2023}\natexlab{}.
\newblock \showarticletitle{{OSDP:} Optimal Sharded Data Parallel for Distributed Deep Learning}. In \bibinfo{booktitle}{\emph{Proceedings of the Thirty-Second International Joint Conference on Artificial Intelligence ({IJCAI} 2023)}}. \bibinfo{pages}{2142--2150}.
\newblock


\bibitem[Jiang et~al\mbox{.}(2024)]%
        {megascale}
\bibfield{author}{\bibinfo{person}{Ziheng Jiang}, \bibinfo{person}{Haibin Lin}, \bibinfo{person}{Yinmin Zhong}, \bibinfo{person}{Qi Huang}, \bibinfo{person}{Yangrui Chen}, \bibinfo{person}{Zhi Zhang}, \bibinfo{person}{Yanghua Peng}, \bibinfo{person}{Xiang Li}, \bibinfo{person}{Cong Xie}, \bibinfo{person}{Shibiao Nong}, \bibinfo{person}{Yulu Jia}, \bibinfo{person}{Sun He}, \bibinfo{person}{Hongmin Chen}, \bibinfo{person}{Zhihao Bai}, \bibinfo{person}{Qi Hou}, \bibinfo{person}{Shipeng Yan}, \bibinfo{person}{Ding Zhou}, \bibinfo{person}{Yiyao Sheng}, \bibinfo{person}{Zhuo Jiang}, \bibinfo{person}{Haohan Xu}, \bibinfo{person}{Haoran Wei}, \bibinfo{person}{Zhang Zhang}, \bibinfo{person}{Pengfei Nie}, \bibinfo{person}{Leqi Zou}, \bibinfo{person}{Sida Zhao}, \bibinfo{person}{Liang Xiang}, \bibinfo{person}{Zherui Liu}, \bibinfo{person}{Zhe Li}, \bibinfo{person}{Xiaoying Jia}, \bibinfo{person}{Jianxi Ye}, \bibinfo{person}{Xin Jin}, {and} \bibinfo{person}{Xin Liu}.} \bibinfo{year}{2024}\natexlab{}.
\newblock \showarticletitle{MegaScale: Scaling Large Language Model Training to More Than 10, 000 GPUs}. In \bibinfo{booktitle}{\emph{21st {USENIX} Symposium on Networked Systems Design and Implementation (NSDI 2024)}}. \bibinfo{pages}{745--760}.
\newblock


\bibitem[Kaplan et~al\mbox{.}(2020)]%
        {scalinglaws}
\bibfield{author}{\bibinfo{person}{Jared Kaplan}, \bibinfo{person}{Sam McCandlish}, \bibinfo{person}{Tom Henighan}, \bibinfo{person}{Tom~B. Brown}, \bibinfo{person}{Benjamin Chess}, \bibinfo{person}{Rewon Child}, \bibinfo{person}{Scott Gray}, \bibinfo{person}{Alec Radford}, \bibinfo{person}{Jeffrey Wu}, {and} \bibinfo{person}{Dario Amodei}.} \bibinfo{year}{2020}\natexlab{}.
\newblock \showarticletitle{Scaling Laws for Neural Language Models}.
\newblock \bibinfo{journal}{\emph{CoRR}}  \bibinfo{volume}{abs/2001.08361} (\bibinfo{year}{2020}).
\newblock


\bibitem[Korthikanti et~al\mbox{.}(2022)]%
        {megatron_3}
\bibfield{author}{\bibinfo{person}{Vijay Korthikanti}, \bibinfo{person}{Jared Casper}, \bibinfo{person}{Sangkug Lym}, \bibinfo{person}{Lawrence McAfee}, \bibinfo{person}{Michael Andersch}, \bibinfo{person}{Mohammad Shoeybi}, {and} \bibinfo{person}{Bryan Catanzaro}.} \bibinfo{year}{2022}\natexlab{}.
\newblock \showarticletitle{Reducing Activation Recomputation in Large Transformer Models}.
\newblock \bibinfo{journal}{\emph{CoRR}}  \bibinfo{volume}{abs/2205.05198} (\bibinfo{year}{2022}).
\newblock


\bibitem[Kurmanji et~al\mbox{.}(2024)]%
        {unlearn_db}
\bibfield{author}{\bibinfo{person}{Meghdad Kurmanji}, \bibinfo{person}{Eleni Triantafillou}, {and} \bibinfo{person}{Peter Triantafillou}.} \bibinfo{year}{2024}\natexlab{}.
\newblock \showarticletitle{Machine Unlearning in Learned Databases: An Experimental Analysis}.
\newblock \bibinfo{journal}{\emph{Proc. {ACM} Manag. Data (SIGMOD 2024)}} \bibinfo{volume}{2}, \bibinfo{number}{1} (\bibinfo{year}{2024}), \bibinfo{pages}{49:1--49:26}.
\newblock


\bibitem[Lao et~al\mbox{.}(2024)]%
        {gptuner}
\bibfield{author}{\bibinfo{person}{Jiale Lao}, \bibinfo{person}{Yibo Wang}, \bibinfo{person}{Yufei Li}, \bibinfo{person}{Jianping Wang}, \bibinfo{person}{Yunjia Zhang}, \bibinfo{person}{Zhiyuan Cheng}, \bibinfo{person}{Wanghu Chen}, \bibinfo{person}{Mingjie Tang}, {and} \bibinfo{person}{Jianguo Wang}.} \bibinfo{year}{2024}\natexlab{}.
\newblock \showarticletitle{GPTuner: {A} Manual-Reading Database Tuning System via GPT-Guided Bayesian Optimization}.
\newblock \bibinfo{journal}{\emph{Proc. {VLDB} Endow.}} \bibinfo{volume}{17}, \bibinfo{number}{8} (\bibinfo{year}{2024}), \bibinfo{pages}{1939--1952}.
\newblock


\bibitem[Li et~al\mbox{.}(2024)]%
        {NLP_SQL_EA}
\bibfield{author}{\bibinfo{person}{Boyan Li}, \bibinfo{person}{Yuyu Luo}, \bibinfo{person}{Chengliang Chai}, \bibinfo{person}{Guoliang Li}, {and} \bibinfo{person}{Nan Tang}.} \bibinfo{year}{2024}\natexlab{}.
\newblock \showarticletitle{The Dawn of Natural Language to {SQL:} Are We Fully Ready?}
\newblock \bibinfo{journal}{\emph{Proc. {VLDB} Endow.}} \bibinfo{volume}{17}, \bibinfo{number}{11} (\bibinfo{year}{2024}), \bibinfo{pages}{3318--3331}.
\newblock


\bibitem[Li et~al\mbox{.}(2022)]%
        {amp_hetero_model_parallel}
\bibfield{author}{\bibinfo{person}{Dacheng Li}, \bibinfo{person}{Hongyi Wang}, \bibinfo{person}{Eric~P. Xing}, {and} \bibinfo{person}{Hao Zhang}.} \bibinfo{year}{2022}\natexlab{}.
\newblock \showarticletitle{{AMP:} Automatically Finding Model Parallel Strategies with Heterogeneity Awareness}. In \bibinfo{booktitle}{\emph{Annual Conference on Neural Information Processing Systems 2022 (NeurIPS 2022)}}.
\newblock


\bibitem[Li et~al\mbox{.}(2021)]%
        {ai_meets_db}
\bibfield{author}{\bibinfo{person}{Guoliang Li}, \bibinfo{person}{Xuanhe Zhou}, {and} \bibinfo{person}{Lei Cao}.} \bibinfo{year}{2021}\natexlab{}.
\newblock \showarticletitle{{AI} Meets Database: {AI4DB} and {DB4AI}}. In \bibinfo{booktitle}{\emph{Proceedings of the 2021 ACM International Conference on Management of Data (SIGMOD 2021)}}. \bibinfo{publisher}{{ACM}}, \bibinfo{pages}{2859--2866}.
\newblock


\bibitem[Li et~al\mbox{.}(2023)]%
        {lyra_elastic_scheduling}
\bibfield{author}{\bibinfo{person}{Jiamin Li}, \bibinfo{person}{Hong Xu}, \bibinfo{person}{Yibo Zhu}, \bibinfo{person}{Zherui Liu}, \bibinfo{person}{Chuanxiong Guo}, {and} \bibinfo{person}{Cong Wang}.} \bibinfo{year}{2023}\natexlab{}.
\newblock \showarticletitle{Lyra: Elastic Scheduling for Deep Learning Clusters}. In \bibinfo{booktitle}{\emph{Proceedings of the Nineteenth European Conference on Computer Systems (EuroSys 2023)}}. \bibinfo{pages}{835--850}.
\newblock


\bibitem[Li et~al\mbox{.}(2014)]%
        {ps_limu}
\bibfield{author}{\bibinfo{person}{Mu Li}, \bibinfo{person}{David~G. Andersen}, \bibinfo{person}{Jun~Woo Park}, \bibinfo{person}{Alexander~J. Smola}, \bibinfo{person}{Amr Ahmed}, \bibinfo{person}{Vanja Josifovski}, \bibinfo{person}{James Long}, \bibinfo{person}{Eugene~J. Shekita}, {and} \bibinfo{person}{Bor{-}Yiing Su}.} \bibinfo{year}{2014}\natexlab{}.
\newblock \showarticletitle{Scaling Distributed Machine Learning with the Parameter Server}. In \bibinfo{booktitle}{\emph{11th {USENIX} Symposium on Operating Systems Design and Implementation (OSDI 2014)}}. \bibinfo{pages}{583--598}.
\newblock


\bibitem[Li et~al\mbox{.}(2020)]%
        {pytorch_ddp}
\bibfield{author}{\bibinfo{person}{Shen Li}, \bibinfo{person}{Yanli Zhao}, \bibinfo{person}{Rohan Varma}, \bibinfo{person}{Omkar Salpekar}, \bibinfo{person}{Pieter Noordhuis}, \bibinfo{person}{Teng Li}, \bibinfo{person}{Adam Paszke}, \bibinfo{person}{Jeff Smith}, \bibinfo{person}{Brian Vaughan}, \bibinfo{person}{Pritam Damania}, {and} \bibinfo{person}{Soumith Chintala}.} \bibinfo{year}{2020}\natexlab{}.
\newblock \showarticletitle{PyTorch Distributed: Experiences on Accelerating Data Parallel Training}.
\newblock \bibinfo{journal}{\emph{Proc. {VLDB} Endow.}} \bibinfo{volume}{13}, \bibinfo{number}{12} (\bibinfo{year}{2020}), \bibinfo{pages}{3005--3018}.
\newblock


\bibitem[Lin et~al\mbox{.}(2019)]%
        {dynamic_minibatch_sgd}
\bibfield{author}{\bibinfo{person}{Haibin Lin}, \bibinfo{person}{Hang Zhang}, \bibinfo{person}{Yifei Ma}, \bibinfo{person}{Tong He}, \bibinfo{person}{Zhi Zhang}, \bibinfo{person}{Sheng Zha}, {and} \bibinfo{person}{Mu Li}.} \bibinfo{year}{2019}\natexlab{}.
\newblock \showarticletitle{Dynamic Mini-batch {SGD} for Elastic Distributed Training: Learning in the Limbo of Resources}.
\newblock \bibinfo{journal}{\emph{CoRR}}  \bibinfo{volume}{abs/1904.12043} (\bibinfo{year}{2019}).
\newblock


\bibitem[Lin et~al\mbox{.}(2023)]%
        {SmartLite}
\bibfield{author}{\bibinfo{person}{Qiuru Lin}, \bibinfo{person}{Sai Wu}, \bibinfo{person}{Junbo Zhao}, \bibinfo{person}{Jian Dai}, \bibinfo{person}{Meng Shi}, \bibinfo{person}{Gang Chen}, {and} \bibinfo{person}{Feifei Li}.} \bibinfo{year}{2023}\natexlab{}.
\newblock \showarticletitle{SmartLite: {A} DBMS-based Serving System for {DNN} Inference in Resource-constrained Environments}.
\newblock \bibinfo{journal}{\emph{Proc. {VLDB} Endow.}} \bibinfo{volume}{17}, \bibinfo{number}{3} (\bibinfo{year}{2023}), \bibinfo{pages}{278--291}.
\newblock


\bibitem[Miao et~al\mbox{.}(2024)]%
        {dm4llm}
\bibfield{author}{\bibinfo{person}{Xupeng Miao}, \bibinfo{person}{Zhihao Jia}, {and} \bibinfo{person}{Bin Cui}.} \bibinfo{year}{2024}\natexlab{}.
\newblock \showarticletitle{Demystifying Data Management for Large Language Models}. In \bibinfo{booktitle}{\emph{Companion of the 2024 International Conference on Management of Data (SIGMOD 2024)}}. \bibinfo{publisher}{{ACM}}, \bibinfo{pages}{547--555}.
\newblock


\bibitem[Miao et~al\mbox{.}(2021)]%
        {partial_reduce}
\bibfield{author}{\bibinfo{person}{Xupeng Miao}, \bibinfo{person}{Xiaonan Nie}, \bibinfo{person}{Yingxia Shao}, \bibinfo{person}{Zhi Yang}, \bibinfo{person}{Jiawei Jiang}, \bibinfo{person}{Lingxiao Ma}, {and} \bibinfo{person}{Bin Cui}.} \bibinfo{year}{2021}\natexlab{}.
\newblock \showarticletitle{Heterogeneity-Aware Distributed Machine Learning Training via Partial Reduce}. In \bibinfo{booktitle}{\emph{Proceedings of the 2021 ACM International Conference on Management of Data (SIGMOD 2021)}}. \bibinfo{pages}{2262--2270}.
\newblock


\bibitem[Miao et~al\mbox{.}(2023)]%
        {hetu}
\bibfield{author}{\bibinfo{person}{Xupeng Miao}, \bibinfo{person}{Xiaonan Nie}, \bibinfo{person}{Hailin Zhang}, \bibinfo{person}{Tong Zhao}, {and} \bibinfo{person}{Bin Cui}.} \bibinfo{year}{2023}\natexlab{}.
\newblock \showarticletitle{Hetu: a highly efficient automatic parallel distributed deep learning system}.
\newblock \bibinfo{journal}{\emph{Sci. China Inf. Sci.}} \bibinfo{volume}{66}, \bibinfo{number}{1} (\bibinfo{year}{2023}).
\newblock


\bibitem[Miao et~al\mbox{.}(2022)]%
        {galvatron}
\bibfield{author}{\bibinfo{person}{Xupeng Miao}, \bibinfo{person}{Yujie Wang}, \bibinfo{person}{Youhe Jiang}, \bibinfo{person}{Chunan Shi}, \bibinfo{person}{Xiaonan Nie}, \bibinfo{person}{Hailin Zhang}, {and} \bibinfo{person}{Bin Cui}.} \bibinfo{year}{2022}\natexlab{}.
\newblock \showarticletitle{Galvatron: Efficient Transformer Training over Multiple GPUs Using Automatic Parallelism}.
\newblock \bibinfo{journal}{\emph{Proc. {VLDB} Endow.}} \bibinfo{volume}{16}, \bibinfo{number}{3} (\bibinfo{year}{2022}), \bibinfo{pages}{470--479}.
\newblock


\bibitem[Mo et~al\mbox{.}(2024)]%
        {heet_elastic_scheduling_in_hetero}
\bibfield{author}{\bibinfo{person}{Zizhao Mo}, \bibinfo{person}{Huanle Xu}, {and} \bibinfo{person}{Chengzhong Xu}.} \bibinfo{year}{2024}\natexlab{}.
\newblock \showarticletitle{Heet: Accelerating Elastic Training in Heterogeneous Deep Learning Clusters}. In \bibinfo{booktitle}{\emph{Architectural Support for Programming Languages and Operating Systems 2020 ({ASPLOS} 2020)}}. \bibinfo{pages}{499--513}.
\newblock


\bibitem[Murray et~al\mbox{.}(2021)]%
        {tf_data}
\bibfield{author}{\bibinfo{person}{Derek~Gordon Murray}, \bibinfo{person}{Jiri Simsa}, \bibinfo{person}{Ana Klimovic}, {and} \bibinfo{person}{Ihor Indyk}.} \bibinfo{year}{2021}\natexlab{}.
\newblock \showarticletitle{tf.data: {A} Machine Learning Data Processing Framework}.
\newblock \bibinfo{journal}{\emph{Proc. {VLDB} Endow.}} \bibinfo{volume}{14}, \bibinfo{number}{12} (\bibinfo{year}{2021}), \bibinfo{pages}{2945--2958}.
\newblock


\bibitem[Nagrecha and Kumar(2023)]%
        {saturn}
\bibfield{author}{\bibinfo{person}{Kabir Nagrecha} {and} \bibinfo{person}{Arun Kumar}.} \bibinfo{year}{2023}\natexlab{}.
\newblock \showarticletitle{Saturn: An Optimized Data System for Multi-Large-Model Deep Learning Workloads}.
\newblock \bibinfo{journal}{\emph{Proc. {VLDB} Endow.}} \bibinfo{volume}{17}, \bibinfo{number}{4} (\bibinfo{year}{2023}), \bibinfo{pages}{712--725}.
\newblock


\bibitem[Narayanan et~al\mbox{.}(2019)]%
        {pipedream}
\bibfield{author}{\bibinfo{person}{Deepak Narayanan}, \bibinfo{person}{Aaron Harlap}, \bibinfo{person}{Amar Phanishayee}, \bibinfo{person}{Vivek Seshadri}, \bibinfo{person}{Nikhil~R. Devanur}, \bibinfo{person}{Gregory~R. Ganger}, \bibinfo{person}{Phillip~B. Gibbons}, {and} \bibinfo{person}{Matei Zaharia}.} \bibinfo{year}{2019}\natexlab{}.
\newblock \showarticletitle{PipeDream: generalized pipeline parallelism for {DNN} training}. In \bibinfo{booktitle}{\emph{Proceedings of the 27th {ACM} Symposium on Operating Systems Principles (SOSP 2019)}}. \bibinfo{pages}{1--15}.
\newblock


\bibitem[Narayanan et~al\mbox{.}(2021a)]%
        {pipedream_flush}
\bibfield{author}{\bibinfo{person}{Deepak Narayanan}, \bibinfo{person}{Amar Phanishayee}, \bibinfo{person}{Kaiyu Shi}, \bibinfo{person}{Xie Chen}, {and} \bibinfo{person}{Matei Zaharia}.} \bibinfo{year}{2021}\natexlab{a}.
\newblock \showarticletitle{Memory-Efficient Pipeline-Parallel {DNN} Training}. In \bibinfo{booktitle}{\emph{International Conference on Machine Learning 2021 (ICML 2021)}}, Vol.~\bibinfo{volume}{139}. \bibinfo{pages}{7937--7947}.
\newblock


\bibitem[Narayanan et~al\mbox{.}(2021b)]%
        {megatron_2}
\bibfield{author}{\bibinfo{person}{Deepak Narayanan}, \bibinfo{person}{Mohammad Shoeybi}, \bibinfo{person}{Jared Casper}, \bibinfo{person}{Patrick LeGresley}, \bibinfo{person}{Mostofa Patwary}, \bibinfo{person}{Vijay Korthikanti}, \bibinfo{person}{Dmitri Vainbrand}, \bibinfo{person}{Prethvi Kashinkunti}, \bibinfo{person}{Julie Bernauer}, \bibinfo{person}{Bryan Catanzaro}, \bibinfo{person}{Amar Phanishayee}, {and} \bibinfo{person}{Matei Zaharia}.} \bibinfo{year}{2021}\natexlab{b}.
\newblock \showarticletitle{Efficient large-scale language model training on {GPU} clusters using megatron-LM}. In \bibinfo{booktitle}{\emph{International Conference for High Performance Computing, Networking 2021 (SC 2021)}}. \bibinfo{pages}{58}.
\newblock


\bibitem[Nie et~al\mbox{.}(2023a)]%
        {angelptm}
\bibfield{author}{\bibinfo{person}{Xiaonan Nie}, \bibinfo{person}{Yi Liu}, \bibinfo{person}{Fangcheng Fu}, \bibinfo{person}{Jinbao Xue}, \bibinfo{person}{Dian Jiao}, \bibinfo{person}{Xupeng Miao}, \bibinfo{person}{Yangyu Tao}, {and} \bibinfo{person}{Bin Cui}.} \bibinfo{year}{2023}\natexlab{a}.
\newblock \showarticletitle{Angel-PTM: {A} Scalable and Economical Large-scale Pre-training System in Tencent}.
\newblock \bibinfo{journal}{\emph{Proc. {VLDB} Endow.}} \bibinfo{volume}{16}, \bibinfo{number}{12} (\bibinfo{year}{2023}), \bibinfo{pages}{3781--3794}.
\newblock


\bibitem[Nie et~al\mbox{.}(2023b)]%
        {FlexMoE}
\bibfield{author}{\bibinfo{person}{Xiaonan Nie}, \bibinfo{person}{Xupeng Miao}, \bibinfo{person}{Zilong Wang}, \bibinfo{person}{Zichao Yang}, \bibinfo{person}{Jilong Xue}, \bibinfo{person}{Lingxiao Ma}, \bibinfo{person}{Gang Cao}, {and} \bibinfo{person}{Bin Cui}.} \bibinfo{year}{2023}\natexlab{b}.
\newblock \showarticletitle{FlexMoE: Scaling Large-scale Sparse Pre-trained Model Training via Dynamic Device Placement}.
\newblock \bibinfo{journal}{\emph{Proc. {ACM} Manag. Data (SIGMOD 2023)}} \bibinfo{volume}{1}, \bibinfo{number}{1} (\bibinfo{year}{2023}), \bibinfo{pages}{110:1--110:19}.
\newblock


\bibitem[NVIDIA(2024a)]%
        {cublas}
\bibfield{author}{\bibinfo{person}{NVIDIA}.} \bibinfo{year}{2024}\natexlab{a}.
\newblock \bibinfo{title}{cuBLAS}.
\newblock
\newblock
\shownote{\url{https://docs.nvidia.com/cuda/cublas/}}.


\bibitem[NVIDIA(2024b)]%
        {cutlass}
\bibfield{author}{\bibinfo{person}{NVIDIA}.} \bibinfo{year}{2024}\natexlab{b}.
\newblock \bibinfo{title}{cutlass}.
\newblock
\newblock
\shownote{\url{https://github.com/NVIDIA/cutlass/}}.


\bibitem[NVIDIA(2024c)]%
        {nccl}
\bibfield{author}{\bibinfo{person}{NVIDIA}.} \bibinfo{year}{2024}\natexlab{c}.
\newblock \bibinfo{title}{NVIDIA Collective Communications Library (NCCL)}.
\newblock
\newblock
\shownote{\url{https://developer.nvidia.com/nccl}}.


\bibitem[NVIDIA(2024d)]%
        {nvidia_resiliency_ext}
\bibfield{author}{\bibinfo{person}{NVIDIA}.} \bibinfo{year}{2024}\natexlab{d}.
\newblock \bibinfo{title}{NVIDIA Resiliency Extension}.
\newblock
\newblock
\shownote{\url{https://github.com/NVIDIA/nvidia-resiliency-ext}}.


\bibitem[Ooi et~al\mbox{.}(2024)]%
        {neurdb}
\bibfield{author}{\bibinfo{person}{Beng~Chin Ooi}, \bibinfo{person}{Shaofeng Cai}, \bibinfo{person}{Gang Chen}, \bibinfo{person}{Yanyan Shen}, \bibinfo{person}{Kian-Lee Tan}, \bibinfo{person}{Yuncheng Wu}, \bibinfo{person}{Xiaokui Xiao}, \bibinfo{person}{Naili Xing}, \bibinfo{person}{Cong Yue}, \bibinfo{person}{Lingze Zeng}, {et~al\mbox{.}}} \bibinfo{year}{2024}\natexlab{}.
\newblock \showarticletitle{NeurDB: An AI-powered Autonomous Data System}.
\newblock \bibinfo{journal}{\emph{Sci. China Inf. Sci.}} (\bibinfo{year}{2024}).
\newblock


\bibitem[Or et~al\mbox{.}(2020)]%
        {resource_elasticity}
\bibfield{author}{\bibinfo{person}{Andrew Or}, \bibinfo{person}{Haoyu Zhang}, {and} \bibinfo{person}{Michael~J. Freedman}.} \bibinfo{year}{2020}\natexlab{}.
\newblock \showarticletitle{Resource Elasticity in Distributed Deep Learning}. In \bibinfo{booktitle}{\emph{Proceedings of Machine Learning and Systems 2020 (MLSys 2020)}}.
\newblock


\bibitem[Paszke et~al\mbox{.}(2019)]%
        {pytorch}
\bibfield{author}{\bibinfo{person}{Adam Paszke}, \bibinfo{person}{Sam Gross}, \bibinfo{person}{Francisco Massa}, \bibinfo{person}{Adam Lerer}, \bibinfo{person}{James Bradbury}, \bibinfo{person}{Gregory Chanan}, \bibinfo{person}{Trevor Killeen}, \bibinfo{person}{Zeming Lin}, \bibinfo{person}{Natalia Gimelshein}, \bibinfo{person}{Luca Antiga}, \bibinfo{person}{Alban Desmaison}, \bibinfo{person}{Andreas K{\"{o}}pf}, \bibinfo{person}{Edward~Z. Yang}, \bibinfo{person}{Zachary DeVito}, \bibinfo{person}{Martin Raison}, \bibinfo{person}{Alykhan Tejani}, \bibinfo{person}{Sasank Chilamkurthy}, \bibinfo{person}{Benoit Steiner}, \bibinfo{person}{Lu Fang}, \bibinfo{person}{Junjie Bai}, {and} \bibinfo{person}{Soumith Chintala}.} \bibinfo{year}{2019}\natexlab{}.
\newblock \showarticletitle{PyTorch: An Imperative Style, High-Performance Deep Learning Library}. In \bibinfo{booktitle}{\emph{Annual Conference on Neural Information Processing Systems 2019 (NeurIPS 2019)}}. \bibinfo{pages}{8024--8035}.
\newblock


\bibitem[Pei et~al\mbox{.}(2023)]%
        {data_and_ai_model_markets}
\bibfield{author}{\bibinfo{person}{Jian Pei}, \bibinfo{person}{Raul~Castro Fernandez}, {and} \bibinfo{person}{Xiaohui Yu}.} \bibinfo{year}{2023}\natexlab{}.
\newblock \showarticletitle{Data and {AI} Model Markets: Opportunities for Data and Model Sharing, Discovery, and Integration}.
\newblock \bibinfo{journal}{\emph{Proc. {VLDB} Endow.}} \bibinfo{volume}{16}, \bibinfo{number}{12} (\bibinfo{year}{2023}), \bibinfo{pages}{3872--3873}.
\newblock


\bibitem[Qiao et~al\mbox{.}(2021)]%
        {pollux_goodput_scheduling}
\bibfield{author}{\bibinfo{person}{Aurick Qiao}, \bibinfo{person}{Sang~Keun Choe}, \bibinfo{person}{Suhas~Jayaram Subramanya}, \bibinfo{person}{Willie Neiswanger}, \bibinfo{person}{Qirong Ho}, \bibinfo{person}{Hao Zhang}, \bibinfo{person}{Gregory~R. Ganger}, {and} \bibinfo{person}{Eric~P. Xing}.} \bibinfo{year}{2021}\natexlab{}.
\newblock \showarticletitle{Pollux: Co-adaptive Cluster Scheduling for Goodput-Optimized Deep Learning}. In \bibinfo{booktitle}{\emph{15th {USENIX} Symposium on Operating Systems Design and Implementation (OSDI 2021)}}.
\newblock


\bibitem[Rajbhandari et~al\mbox{.}(2020)]%
        {zero}
\bibfield{author}{\bibinfo{person}{Samyam Rajbhandari}, \bibinfo{person}{Jeff Rasley}, \bibinfo{person}{Olatunji Ruwase}, {and} \bibinfo{person}{Yuxiong He}.} \bibinfo{year}{2020}\natexlab{}.
\newblock \showarticletitle{ZeRO: memory optimizations toward training trillion parameter models}. In \bibinfo{booktitle}{\emph{Proceedings of the International Conference for High Performance Computing, Networking, Storage and Analysis (SC 2020)}}. \bibinfo{pages}{20}.
\newblock


\bibitem[Sergeev and Balso(2018)]%
        {horovod}
\bibfield{author}{\bibinfo{person}{Alexander Sergeev} {and} \bibinfo{person}{Mike~Del Balso}.} \bibinfo{year}{2018}\natexlab{}.
\newblock \showarticletitle{Horovod: fast and easy distributed deep learning in TensorFlow}.
\newblock \bibinfo{journal}{\emph{CoRR}}  \bibinfo{volume}{abs/1802.05799} (\bibinfo{year}{2018}).
\newblock


\bibitem[Shoeybi et~al\mbox{.}(2019)]%
        {megatron_1}
\bibfield{author}{\bibinfo{person}{Mohammad Shoeybi}, \bibinfo{person}{Mostofa Patwary}, \bibinfo{person}{Raul Puri}, \bibinfo{person}{Patrick LeGresley}, \bibinfo{person}{Jared Casper}, {and} \bibinfo{person}{Bryan Catanzaro}.} \bibinfo{year}{2019}\natexlab{}.
\newblock \showarticletitle{Megatron-LM: Training Multi-Billion Parameter Language Models Using Model Parallelism}.
\newblock \bibinfo{journal}{\emph{CoRR}}  \bibinfo{volume}{abs/1909.08053} (\bibinfo{year}{2019}).
\newblock


\bibitem[Siddiqi et~al\mbox{.}(2023)]%
        {saga_data_cleaning}
\bibfield{author}{\bibinfo{person}{Shafaq Siddiqi}, \bibinfo{person}{Roman Kern}, {and} \bibinfo{person}{Matthias Boehm}.} \bibinfo{year}{2023}\natexlab{}.
\newblock \showarticletitle{{SAGA:} {A} Scalable Framework for Optimizing Data Cleaning Pipelines for Machine Learning Applications}.
\newblock \bibinfo{journal}{\emph{Proc. {ACM} Manag. Data (SIGMOD 2023)}} \bibinfo{volume}{1}, \bibinfo{number}{3} (\bibinfo{year}{2023}), \bibinfo{pages}{218:1--218:26}.
\newblock


\bibitem[Song et~al\mbox{.}(2020)]%
        {acc_par}
\bibfield{author}{\bibinfo{person}{Linghao Song}, \bibinfo{person}{Fan Chen}, \bibinfo{person}{Youwei Zhuo}, \bibinfo{person}{Xuehai Qian}, \bibinfo{person}{Hai Li}, {and} \bibinfo{person}{Yiran Chen}.} \bibinfo{year}{2020}\natexlab{}.
\newblock \showarticletitle{AccPar: Tensor Partitioning for Heterogeneous Deep Learning Accelerators}. In \bibinfo{booktitle}{\emph{{IEEE} International Symposium on High Performance Computer Architecture, 2020 (HPCA 2020)}}. \bibinfo{pages}{342--355}.
\newblock


\bibitem[Strati et~al\mbox{.}(2024)]%
        {gpu_shortage_cross_region_cloud}
\bibfield{author}{\bibinfo{person}{Foteini Strati}, \bibinfo{person}{Paul Elvinger}, \bibinfo{person}{Tolga Kerimoglu}, {and} \bibinfo{person}{Ana Klimovic}.} \bibinfo{year}{2024}\natexlab{}.
\newblock \showarticletitle{{ML} Training with Cloud {GPU} Shortages: Is Cross-Region the Answer?}. In \bibinfo{booktitle}{\emph{Proceedings of the 4th Workshop on Machine Learning and Systems, EuroMLSys 2024}}. \bibinfo{pages}{107--116}.
\newblock


\bibitem[Tarnawski et~al\mbox{.}(2021)]%
        {piper}
\bibfield{author}{\bibinfo{person}{Jakub Tarnawski}, \bibinfo{person}{Deepak Narayanan}, {and} \bibinfo{person}{Amar Phanishayee}.} \bibinfo{year}{2021}\natexlab{}.
\newblock \showarticletitle{Piper: Multidimensional Planner for {DNN} Parallelization}. In \bibinfo{booktitle}{\emph{Annual Conference on Neural Information Processing Systems 2021 (NeurIPS 2021)}}. \bibinfo{pages}{24829--24840}.
\newblock


\bibitem[Team(2024)]%
        {imbue_report}
\bibfield{author}{\bibinfo{person}{The~Imbue Team}.} \bibinfo{year}{2024}\natexlab{}.
\newblock \bibinfo{title}{From bare metal to a 70B model: infrastructure set-up and scripts}.
\newblock
\newblock
\shownote{\url{https://imbue.com/research/70b-infrastructure/}}.


\bibitem[Thorpe et~al\mbox{.}(2023)]%
        {bamboo}
\bibfield{author}{\bibinfo{person}{John Thorpe}, \bibinfo{person}{Pengzhan Zhao}, \bibinfo{person}{Jonathan Eyolfson}, \bibinfo{person}{Yifan Qiao}, \bibinfo{person}{Zhihao Jia}, \bibinfo{person}{Minjia Zhang}, \bibinfo{person}{Ravi Netravali}, {and} \bibinfo{person}{Guoqing~Harry Xu}.} \bibinfo{year}{2023}\natexlab{}.
\newblock \showarticletitle{Bamboo: Making Preemptible Instances Resilient for Affordable Training of Large DNNs}. In \bibinfo{booktitle}{\emph{20th {USENIX} Symposium on Networked Systems Design and Implementation (NSDI 2023)}}. \bibinfo{pages}{497--513}.
\newblock


\bibitem[Touvron et~al\mbox{.}(2023)]%
        {llama2}
\bibfield{author}{\bibinfo{person}{Hugo Touvron}, \bibinfo{person}{Louis Martin}, \bibinfo{person}{Kevin Stone}, \bibinfo{person}{Peter Albert}, \bibinfo{person}{Amjad Almahairi}, \bibinfo{person}{Yasmine Babaei}, \bibinfo{person}{Nikolay Bashlykov}, \bibinfo{person}{Soumya Batra}, \bibinfo{person}{Prajjwal Bhargava}, \bibinfo{person}{Shruti Bhosale}, \bibinfo{person}{Dan Bikel}, \bibinfo{person}{Lukas Blecher}, \bibinfo{person}{Cristian Canton{-}Ferrer}, \bibinfo{person}{Moya Chen}, \bibinfo{person}{Guillem Cucurull}, \bibinfo{person}{David Esiobu}, \bibinfo{person}{Jude Fernandes}, \bibinfo{person}{Jeremy Fu}, \bibinfo{person}{Wenyin Fu}, \bibinfo{person}{Brian Fuller}, \bibinfo{person}{Cynthia Gao}, \bibinfo{person}{Vedanuj Goswami}, \bibinfo{person}{Naman Goyal}, \bibinfo{person}{Anthony Hartshorn}, \bibinfo{person}{Saghar Hosseini}, \bibinfo{person}{Rui Hou}, \bibinfo{person}{Hakan Inan}, \bibinfo{person}{Marcin Kardas}, \bibinfo{person}{Viktor Kerkez}, \bibinfo{person}{Madian Khabsa},
  \bibinfo{person}{Isabel Kloumann}, \bibinfo{person}{Artem Korenev}, \bibinfo{person}{Punit~Singh Koura}, \bibinfo{person}{Marie{-}Anne Lachaux}, \bibinfo{person}{Thibaut Lavril}, \bibinfo{person}{Jenya Lee}, \bibinfo{person}{Diana Liskovich}, \bibinfo{person}{Yinghai Lu}, \bibinfo{person}{Yuning Mao}, \bibinfo{person}{Xavier Martinet}, \bibinfo{person}{Todor Mihaylov}, \bibinfo{person}{Pushkar Mishra}, \bibinfo{person}{Igor Molybog}, \bibinfo{person}{Yixin Nie}, \bibinfo{person}{Andrew Poulton}, \bibinfo{person}{Jeremy Reizenstein}, \bibinfo{person}{Rashi Rungta}, \bibinfo{person}{Kalyan Saladi}, \bibinfo{person}{Alan Schelten}, \bibinfo{person}{Ruan Silva}, \bibinfo{person}{Eric~Michael Smith}, \bibinfo{person}{Ranjan Subramanian}, \bibinfo{person}{Xiaoqing~Ellen Tan}, \bibinfo{person}{Binh Tang}, \bibinfo{person}{Ross Taylor}, \bibinfo{person}{Adina Williams}, \bibinfo{person}{Jian~Xiang Kuan}, \bibinfo{person}{Puxin Xu}, \bibinfo{person}{Zheng Yan}, \bibinfo{person}{Iliyan Zarov}, \bibinfo{person}{Yuchen
  Zhang}, \bibinfo{person}{Angela Fan}, \bibinfo{person}{Melanie Kambadur}, \bibinfo{person}{Sharan Narang}, \bibinfo{person}{Aur{\'{e}}lien Rodriguez}, \bibinfo{person}{Robert Stojnic}, \bibinfo{person}{Sergey Edunov}, {and} \bibinfo{person}{Thomas Scialom}.} \bibinfo{year}{2023}\natexlab{}.
\newblock \showarticletitle{Llama 2: Open Foundation and Fine-Tuned Chat Models}.
\newblock \bibinfo{journal}{\emph{CoRR}}  \bibinfo{volume}{abs/2307.09288} (\bibinfo{year}{2023}).
\newblock


\bibitem[Vaswani et~al\mbox{.}(2017)]%
        {attn}
\bibfield{author}{\bibinfo{person}{Ashish Vaswani}, \bibinfo{person}{Noam Shazeer}, \bibinfo{person}{Niki Parmar}, \bibinfo{person}{Jakob Uszkoreit}, \bibinfo{person}{Llion Jones}, \bibinfo{person}{Aidan~N. Gomez}, \bibinfo{person}{Lukasz Kaiser}, {and} \bibinfo{person}{Illia Polosukhin}.} \bibinfo{year}{2017}\natexlab{}.
\newblock \showarticletitle{Attention is All you Need}. In \bibinfo{booktitle}{\emph{Annual Conference on Neural Information Processing Systems 2017 (NeurIPS 2017)}}. \bibinfo{pages}{5998--6008}.
\newblock


\bibitem[Wang et~al\mbox{.}(2024)]%
        {galvatron_bmw}
\bibfield{author}{\bibinfo{person}{Yujie Wang}, \bibinfo{person}{Youhe Jiang}, \bibinfo{person}{Xupeng Miao}, \bibinfo{person}{Fangcheng Fu}, \bibinfo{person}{Shenhan Zhu}, \bibinfo{person}{Xiaonan Nie}, \bibinfo{person}{Yaofeng Tu}, {and} \bibinfo{person}{Bin Cui}.} \bibinfo{year}{2024}\natexlab{}.
\newblock \showarticletitle{Improving Automatic Parallel Training via Balanced Memory Workload Optimization}.
\newblock \bibinfo{journal}{\emph{{IEEE} Trans. Knowl. Data Eng.}} \bibinfo{volume}{36}, \bibinfo{number}{8} (\bibinfo{year}{2024}), \bibinfo{pages}{3906--3920}.
\newblock


\bibitem[Wang et~al\mbox{.}(2025)]%
        {flexsp}
\bibfield{author}{\bibinfo{person}{Yujie Wang}, \bibinfo{person}{Shiju Wang}, \bibinfo{person}{Shenhan Zhu}, \bibinfo{person}{Fangcheng Fu}, \bibinfo{person}{Xinyi Liu}, \bibinfo{person}{Xuefeng Xiao}, \bibinfo{person}{Huixia Li}, \bibinfo{person}{Jiashi Li}, \bibinfo{person}{Faming Wu}, {and} \bibinfo{person}{Bin Cui}.} \bibinfo{year}{2025}\natexlab{}.
\newblock \showarticletitle{FlexSP: Accelerating Large Language Model Training via Flexible Sequence Parallelism}. In \bibinfo{booktitle}{\emph{Architectural Support for Programming Languages and Operating Systems 2025 ({ASPLOS} 2025)}}. \bibinfo{pages}{421--436}.
\newblock


\bibitem[Wang et~al\mbox{.}(2023)]%
        {gemini_failure_recovery}
\bibfield{author}{\bibinfo{person}{Zhuang Wang}, \bibinfo{person}{Zhen Jia}, \bibinfo{person}{Shuai Zheng}, \bibinfo{person}{Zhen Zhang}, \bibinfo{person}{Xinwei Fu}, \bibinfo{person}{T.~S.~Eugene Ng}, {and} \bibinfo{person}{Yida Wang}.} \bibinfo{year}{2023}\natexlab{}.
\newblock \showarticletitle{{GEMINI:} Fast Failure Recovery in Distributed Training with In-Memory Checkpoints}. In \bibinfo{booktitle}{\emph{Proceedings of the 29th {ACM} Symposium on Operating Systems Principles (SOSP 2023)}}. \bibinfo{pages}{364--381}.
\newblock


\bibitem[Wu et~al\mbox{.}(2024)]%
        {falcon_straggler_hybrid_parallel}
\bibfield{author}{\bibinfo{person}{Tianyuan Wu}, \bibinfo{person}{Wei Wang}, \bibinfo{person}{Yinghao Yu}, \bibinfo{person}{Siran Yang}, \bibinfo{person}{Wenchao Wu}, \bibinfo{person}{Qinkai Duan}, \bibinfo{person}{Guodong Yang}, \bibinfo{person}{Jiamang Wang}, \bibinfo{person}{Lin Qu}, {and} \bibinfo{person}{Liping Zhang}.} \bibinfo{year}{2024}\natexlab{}.
\newblock \showarticletitle{{FALCON:} Pinpointing and Mitigating Stragglers for Large-Scale Hybrid-Parallel Training}.
\newblock \bibinfo{journal}{\emph{CoRR}}  \bibinfo{volume}{abs/2410.12588} (\bibinfo{year}{2024}).
\newblock


\bibitem[Xia et~al\mbox{.}(2023)]%
        {flashllm}
\bibfield{author}{\bibinfo{person}{Haojun Xia}, \bibinfo{person}{Zhen Zheng}, \bibinfo{person}{Yuchao Li}, \bibinfo{person}{Donglin Zhuang}, \bibinfo{person}{Zhongzhu Zhou}, \bibinfo{person}{Xiafei Qiu}, \bibinfo{person}{Yong Li}, \bibinfo{person}{Wei Lin}, {and} \bibinfo{person}{Shuaiwen~Leon Song}.} \bibinfo{year}{2023}\natexlab{}.
\newblock \showarticletitle{Flash-LLM: Enabling Cost-Effective and Highly-Efficient Large Generative Model Inference with Unstructured Sparsity}.
\newblock \bibinfo{journal}{\emph{Proc. {VLDB} Endow.}} \bibinfo{volume}{17}, \bibinfo{number}{2} (\bibinfo{year}{2023}), \bibinfo{pages}{211--224}.
\newblock


\bibitem[Xu et~al\mbox{.}(2024)]%
        {hethub}
\bibfield{author}{\bibinfo{person}{Si Xu}, \bibinfo{person}{Zixiao Huang}, \bibinfo{person}{Yan Zeng}, \bibinfo{person}{Shengen Yan}, \bibinfo{person}{Xuefei Ning}, \bibinfo{person}{Haolin Ye}, \bibinfo{person}{Sipei Gu}, \bibinfo{person}{Chunsheng Shui}, \bibinfo{person}{Zhezheng Lin}, \bibinfo{person}{Hao Zhang}, \bibinfo{person}{Sheng Wang}, \bibinfo{person}{Guohao Dai}, {and} \bibinfo{person}{Yu Wang}.} \bibinfo{year}{2024}\natexlab{}.
\newblock \showarticletitle{HetHub: {A} Heterogeneous distributed hybrid training system for large-scale models}.
\newblock \bibinfo{journal}{\emph{CoRR}}  \bibinfo{volume}{abs/2405.16256} (\bibinfo{year}{2024}).
\newblock


\bibitem[Yang et~al\mbox{.}(2023)]%
        {skypilot}
\bibfield{author}{\bibinfo{person}{Zongheng Yang}, \bibinfo{person}{Zhanghao Wu}, \bibinfo{person}{Michael Luo}, \bibinfo{person}{Wei{-}Lin Chiang}, \bibinfo{person}{Romil Bhardwaj}, \bibinfo{person}{Woosuk Kwon}, \bibinfo{person}{Siyuan Zhuang}, \bibinfo{person}{Frank~Sifei Luan}, \bibinfo{person}{Gautam Mittal}, \bibinfo{person}{Scott Shenker}, {and} \bibinfo{person}{Ion Stoica}.} \bibinfo{year}{2023}\natexlab{}.
\newblock \showarticletitle{SkyPilot: An Intercloud Broker for Sky Computing}. In \bibinfo{booktitle}{\emph{20th {USENIX} Symposium on Networked Systems Design and Implementation (NSDI 2023)}}. \bibinfo{pages}{437--455}.
\newblock


\bibitem[Zhang and R{\'{e}}(2014)]%
        {DimmWitted}
\bibfield{author}{\bibinfo{person}{Ce Zhang} {and} \bibinfo{person}{Christopher R{\'{e}}}.} \bibinfo{year}{2014}\natexlab{}.
\newblock \showarticletitle{DimmWitted: {A} Study of Main-Memory Statistical Analytics}.
\newblock \bibinfo{journal}{\emph{Proc. {VLDB} Endow.}} \bibinfo{volume}{7}, \bibinfo{number}{12} (\bibinfo{year}{2014}), \bibinfo{pages}{1283--1294}.
\newblock


\bibitem[Zhang et~al\mbox{.}(2024a)]%
        {hap}
\bibfield{author}{\bibinfo{person}{Shiwei Zhang}, \bibinfo{person}{Lansong Diao}, \bibinfo{person}{Chuan Wu}, \bibinfo{person}{Zongyan Cao}, \bibinfo{person}{Siyu Wang}, {and} \bibinfo{person}{Wei Lin}.} \bibinfo{year}{2024}\natexlab{a}.
\newblock \showarticletitle{{HAP:} {SPMD} {DNN} Training on Heterogeneous {GPU} Clusters with Automated Program Synthesis}. In \bibinfo{booktitle}{\emph{Proceedings of the Nineteenth European Conference on Computer Systems (EuroSys 2024)}}. \bibinfo{pages}{524--541}.
\newblock


\bibitem[Zhang et~al\mbox{.}(2021)]%
        {dl_on_ds}
\bibfield{author}{\bibinfo{person}{Yuhao Zhang}, \bibinfo{person}{Frank Mcquillan}, \bibinfo{person}{Nandish Jayaram}, \bibinfo{person}{Nikhil Kak}, \bibinfo{person}{Ekta Khanna}, \bibinfo{person}{Orhan Kislal}, \bibinfo{person}{Domino Valdano}, {and} \bibinfo{person}{Arun Kumar}.} \bibinfo{year}{2021}\natexlab{}.
\newblock \showarticletitle{Distributed Deep Learning on Data Systems: {A} Comparative Analysis of Approaches}.
\newblock \bibinfo{journal}{\emph{Proc. {VLDB} Endow.}} \bibinfo{volume}{14}, \bibinfo{number}{10} (\bibinfo{year}{2021}), \bibinfo{pages}{1769--1782}.
\newblock


\bibitem[Zhang et~al\mbox{.}(2022)]%
        {MiCS}
\bibfield{author}{\bibinfo{person}{Zhen Zhang}, \bibinfo{person}{Shuai Zheng}, \bibinfo{person}{Yida Wang}, \bibinfo{person}{Justin Chiu}, \bibinfo{person}{George Karypis}, \bibinfo{person}{Trishul Chilimbi}, \bibinfo{person}{Mu Li}, {and} \bibinfo{person}{Xin Jin}.} \bibinfo{year}{2022}\natexlab{}.
\newblock \showarticletitle{MiCS: Near-linear Scaling for Training Gigantic Model on Public Cloud}.
\newblock \bibinfo{journal}{\emph{Proc. {VLDB} Endow.}} \bibinfo{volume}{16}, \bibinfo{number}{1} (\bibinfo{year}{2022}), \bibinfo{pages}{37--50}.
\newblock


\bibitem[Zhang et~al\mbox{.}(2024b)]%
        {ai_sys_for_llm_training_survey}
\bibfield{author}{\bibinfo{person}{Zhen-Xing Zhang}, \bibinfo{person}{Yuan-Bo Wen}, \bibinfo{person}{Han-Qi Lv}, \bibinfo{person}{Chang Liu}, \bibinfo{person}{Rui Zhang}, \bibinfo{person}{Xia-Qing Li}, \bibinfo{person}{Chao Wang}, \bibinfo{person}{Zi-Dong Du}, \bibinfo{person}{Qi Guo}, \bibinfo{person}{Ling Li}, \bibinfo{person}{Xue-Hai Zhou}, {and} \bibinfo{person}{Yun-Ji Chen}.} \bibinfo{year}{2024}\natexlab{b}.
\newblock \showarticletitle{AI Computing Systems for LLMs Training: A Review}.
\newblock \bibinfo{journal}{\emph{J. Comput. Sci. Technol.}} (\bibinfo{year}{2024}).
\newblock


\bibitem[Zhao et~al\mbox{.}(2023b)]%
        {goldminer}
\bibfield{author}{\bibinfo{person}{Hanyu Zhao}, \bibinfo{person}{Zhi Yang}, \bibinfo{person}{Yu Cheng}, \bibinfo{person}{Chao Tian}, \bibinfo{person}{Shiru Ren}, \bibinfo{person}{Wencong Xiao}, \bibinfo{person}{Man Yuan}, \bibinfo{person}{Langshi Chen}, \bibinfo{person}{Kaibo Liu}, \bibinfo{person}{Yang Zhang}, \bibinfo{person}{Yong Li}, {and} \bibinfo{person}{Wei Lin}.} \bibinfo{year}{2023}\natexlab{b}.
\newblock \showarticletitle{GoldMiner: Elastic Scaling of Training Data Pre-Processing Pipelines for Deep Learning}.
\newblock \bibinfo{journal}{\emph{Proc. {ACM} Manag. Data (SIGMOD 2023)}} \bibinfo{volume}{1}, \bibinfo{number}{2} (\bibinfo{year}{2023}), \bibinfo{pages}{193:1--193:25}.
\newblock


\bibitem[Zhao et~al\mbox{.}(2019)]%
        {dynamic_ssp}
\bibfield{author}{\bibinfo{person}{Xing Zhao}, \bibinfo{person}{Aijun An}, \bibinfo{person}{Junfeng Liu}, {and} \bibinfo{person}{Bao~Xin Chen}.} \bibinfo{year}{2019}\natexlab{}.
\newblock \showarticletitle{Dynamic Stale Synchronous Parallel Distributed Training for Deep Learning}. In \bibinfo{booktitle}{\emph{IEEE International Conference on Distributed Computing Systems (ICDCS 2019)}}. \bibinfo{pages}{1507--1517}.
\newblock


\bibitem[Zhao et~al\mbox{.}(2023a)]%
        {pytorch_fsdp}
\bibfield{author}{\bibinfo{person}{Yanli Zhao}, \bibinfo{person}{Andrew Gu}, \bibinfo{person}{Rohan Varma}, \bibinfo{person}{Liang Luo}, \bibinfo{person}{Chien{-}Chin Huang}, \bibinfo{person}{Min Xu}, \bibinfo{person}{Less Wright}, \bibinfo{person}{Hamid Shojanazeri}, \bibinfo{person}{Myle Ott}, \bibinfo{person}{Sam Shleifer}, \bibinfo{person}{Alban Desmaison}, \bibinfo{person}{Can Balioglu}, \bibinfo{person}{Pritam Damania}, \bibinfo{person}{Bernard Nguyen}, \bibinfo{person}{Geeta Chauhan}, \bibinfo{person}{Yuchen Hao}, \bibinfo{person}{Ajit Mathews}, {and} \bibinfo{person}{Shen Li}.} \bibinfo{year}{2023}\natexlab{a}.
\newblock \showarticletitle{PyTorch {FSDP:} Experiences on Scaling Fully Sharded Data Parallel}.
\newblock \bibinfo{journal}{\emph{Proc. {VLDB} Endow.}} \bibinfo{volume}{16}, \bibinfo{number}{12} (\bibinfo{year}{2023}), \bibinfo{pages}{3848--3860}.
\newblock


\bibitem[Zheng et~al\mbox{.}(2022)]%
        {alpa}
\bibfield{author}{\bibinfo{person}{Lianmin Zheng}, \bibinfo{person}{Zhuohan Li}, \bibinfo{person}{Hao Zhang}, \bibinfo{person}{Yonghao Zhuang}, \bibinfo{person}{Zhifeng Chen}, \bibinfo{person}{Yanping Huang}, \bibinfo{person}{Yida Wang}, \bibinfo{person}{Yuanzhong Xu}, \bibinfo{person}{Danyang Zhuo}, \bibinfo{person}{Eric~P. Xing}, \bibinfo{person}{Joseph~E. Gonzalez}, {and} \bibinfo{person}{Ion Stoica}.} \bibinfo{year}{2022}\natexlab{}.
\newblock \showarticletitle{Alpa: Automating Inter- and Intra-Operator Parallelism for Distributed Deep Learning}. In \bibinfo{booktitle}{\emph{16th {USENIX} Symposium on Operating Systems Design and Implementation (OSDI 2022)}}. \bibinfo{pages}{559--578}.
\newblock


\bibitem[Zheng et~al\mbox{.}(2023)]%
        {BladeDISC}
\bibfield{author}{\bibinfo{person}{Zhen Zheng}, \bibinfo{person}{Zaifeng Pan}, \bibinfo{person}{Dalin Wang}, \bibinfo{person}{Kai Zhu}, \bibinfo{person}{Wenyi Zhao}, \bibinfo{person}{Tianyou Guo}, \bibinfo{person}{Xiafei Qiu}, \bibinfo{person}{Minmin Sun}, \bibinfo{person}{Junjie Bai}, \bibinfo{person}{Feng Zhang}, \bibinfo{person}{Xiaoyong Du}, \bibinfo{person}{Jidong Zhai}, {and} \bibinfo{person}{Wei Lin}.} \bibinfo{year}{2023}\natexlab{}.
\newblock \showarticletitle{BladeDISC: Optimizing Dynamic Shape Machine Learning Workloads via Compiler Approach}.
\newblock \bibinfo{journal}{\emph{Proc. {ACM} Manag. Data (SIGMOD 2023)}} \bibinfo{volume}{1}, \bibinfo{number}{3} (\bibinfo{year}{2023}), \bibinfo{pages}{206:1--206:29}.
\newblock


\bibitem[Zhou et~al\mbox{.}(2024)]%
        {dbgpt}
\bibfield{author}{\bibinfo{person}{Xuanhe Zhou}, \bibinfo{person}{Zhaoyan Sun}, {and} \bibinfo{person}{Guoliang Li}.} \bibinfo{year}{2024}\natexlab{}.
\newblock \showarticletitle{{DB-GPT:} Large Language Model Meets Database}.
\newblock \bibinfo{journal}{\emph{Data Sci. Eng.}} (\bibinfo{year}{2024}).
\newblock


\bibitem[Zhu et~al\mbox{.}(2024)]%
        {data_cleaning_meets_ai}
\bibfield{author}{\bibinfo{person}{Jingyu Zhu}, \bibinfo{person}{Xintong Zhao}, \bibinfo{person}{Yu Sun}, \bibinfo{person}{Shaoxu Song}, {and} \bibinfo{person}{Xiaojie Yuan}.} \bibinfo{year}{2024}\natexlab{}.
\newblock \showarticletitle{Relational Data Cleaning Meets Artificial Intelligence: A Survey}.
\newblock \bibinfo{journal}{\emph{Data Sci. Eng.}} (\bibinfo{year}{2024}).
\newblock


\end{thebibliography}

% \onecolumn
\clearpage
\appendix
\section{More Experiment Results}
\label{appendix:more_expr}

\subsection{Effectiveness of the Cost Model}
\label{appendix:cost_model_effectiveness}

As discussed in \S\ref{sec:expr_e2e}, the results in Table~\ref{tb:theoretic_ratios} show that our cost model is accurate. 
To further examine the effectiveness of our cost model, i.e., whether it is useful in helping us deduce the optimal plan, we conduct a representative testbed using the 32B model. In particular, we employ a fixed hybrid parallel strategy with DP, PP, and TP degrees of 4, 2, and 2, respectively, and we further decrease the sequence length from 4K, as used in previous experiments, to 1K to bypass all memory constraints. Additionally, we increase the global batch size from 64 to 512 and keep the micro-batch size as 1, allowing for a more refined granularity of data assignment and, consequently, a more rigorous validation for the precision of our cost model. In the experiment, we introduce a level-1 straggler, without the need of isolating heavy stragglers or non-uniform stages across the pipelines. Due to symmetry, we can actually enumerate all possibilities by traversing the layers and data allocated to the straggler GPU and measure the end-to-end performance for each partitioning.

For layer partitioning, as each pipeline only involves two stages, after enumerating the number of layers $l$ allocated to the stage containing the straggler, the remaining stage in the pipeline will be allocated $60 - l$ layers for sure, while for the other pipelines composed of non-straggling GPUs, the optimal layer partitioning remains evenly assigning 30 layers to each stage. As For data partitioning, since DP is 4, we only need to enumerate the number of micro-batches $m$ allocated to the pipeline containing the straggler, and the remaining three pipelines, being completely isomorphic, will evenly distribute the remaining 
$512 - m$ micro-batches (ideally, without considering the integer constraint of micro-batches, each normal pipeline should be assigned ${(512 - m)}/{3}$ micro-batches). 

Following the aforementioned approach, we first enumerate all layer partitioning possibilities and select the optimal from them. Based on this, we again enumerate all data partitioning possibilities. At each step of enumeration, we test the actual profiling time on norm (i.e., non-straggling) GPUs, profiling time on the straggling GPU, and the overall end-to-end time, and compare them with the estimated time given by our cost model. As shown in Figure~\ref{fig:hetero_layer_batch}, it can be observed that our cost model well approximates the actual running time, and the final layer and data partitioning solutions also coincide with the optimal solution found through actual end-to-end enumeration. This demonstrates that our cost model can effectively identify the optimal load balancing point, achieving the optimal solution in practical.

\begin{figure*}[!h]
\centering
\includegraphics[width=0.49\textwidth]{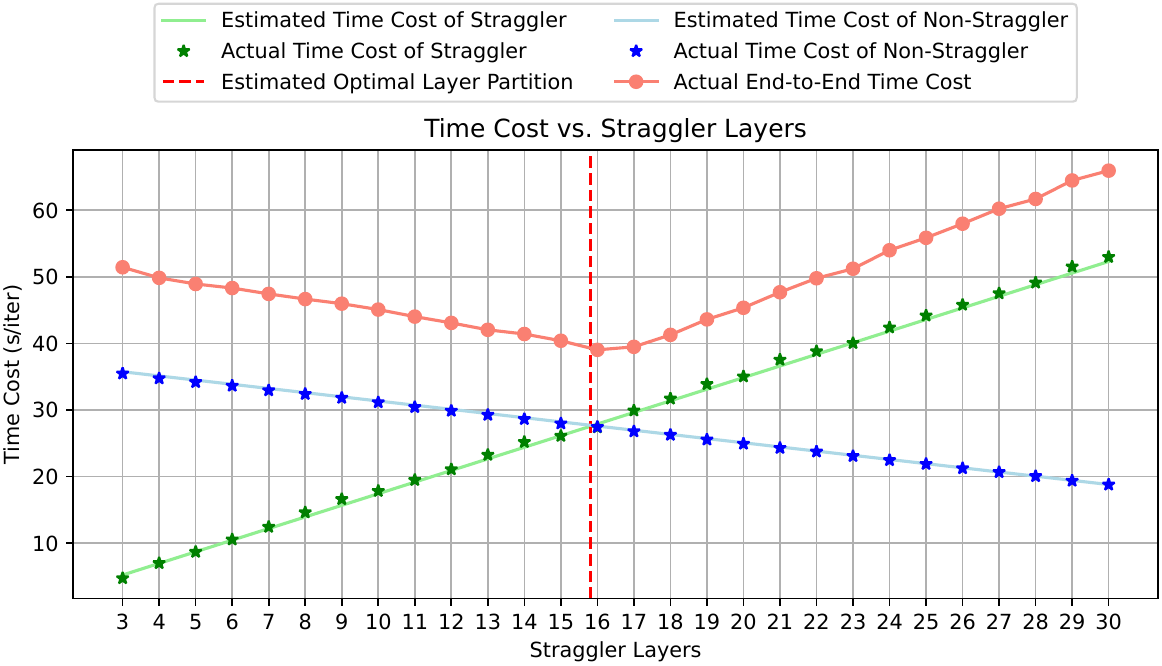}
\includegraphics[width=0.49\textwidth]{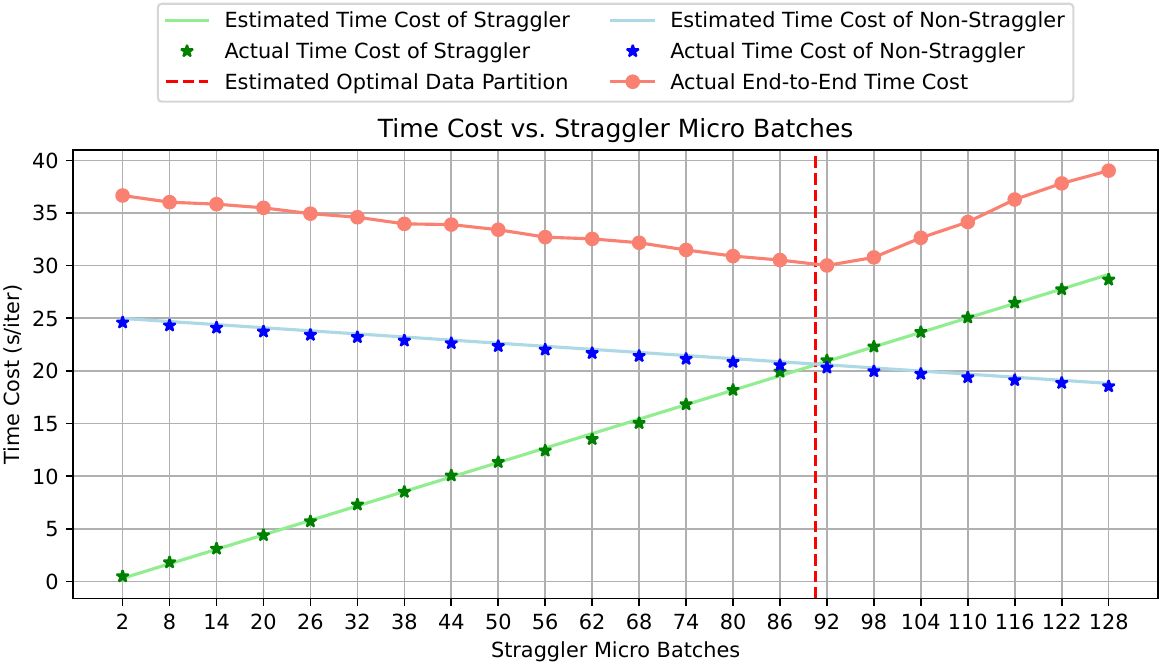}
\caption{\small{Enumeration on layer and data partitioning. The results given by our cost model coincide precisely with the optimal solution towards the final load balancing.}}
\label{fig:hetero_layer_batch}
\end{figure*}

\subsection{Scalability of Parallelization Planning}
\label{appendix:scalability}
According to the experiment results in \S\ref{sec:expr_e2e}, the parallelization planning process can be fully overlapped by the training of 1 iteration. 
In this experiment, we wish to examine the scalability of our parallelization planning algorithm to more GPUs. 
Unfortunately, due to the high expense of GPUs, we cannot evaluate the training efficiency of \system over more GPUs. 
Thus, we only focus on the time cost of planning here. 
Specifically, we assume a total of 1024 GPUs (128 8-GPU nodes) are used to train the 110B model. Meanwhile, we assume the global batch size is linearly scaled to 1024 (originally 64 in \S\ref{sec:expr}), constituting each batch with 4 million tokens, which is a reasonable configuration in LLM training. 
We further assume there are 32 stragglers (approximately 3\% of the cluster). Below, we discuss the breakdown of the time taken for each step (shown in Table~\ref{tb:scalability}) of our parallelization planning algorithm in this scenario:
\begin{itemize}[leftmargin=*]
    \item \textbf{Enumeration overhead of GPU grouping (\S\ref{sec:planner_upper_level_gpu_grouping})}: The overhead of this part is negligible (0.01 second). Although this part involves the enumeration of group splitting, each enumeration only requires using Theorem~\ref{thm:theoretic_optimality} to calculate a ratio, and the number of enumerations is not large. Specifically, we began with running four processes to solve the optimal strategy for TP limits of 1, 2, 4, and 8, respectively. With the overall TP degree fixed, we then considered how to split the TP groups for the stragglers. For example, when TP degree is 1, there is no group splitting, and when TP degree is 8, a node with straggler(s) needs to enumerate up to 6 times of group splitting (as discussed in \S\ref{sec:planner_upper_level_gpu_grouping} and detailed in Appendix~\ref{appendix:deduce_grouping-possibilities}). With 32 straggler GPUs across a maximum of 32 nodes, we only need to enumerate 32 times. Putting them together, there are only a few hundreds of enumerations.

    \item \textbf{Pipeline division overhead (first half of \S\ref{sec:planner_upper_level_pipeline_orchestration})}: This part has the largest overhead (51.23 seconds) due to the complexity of the MINLP problem and its correlation with the DP degree as well as the number of straggling TP groups. However, it can still finish within reasonable time even for more than one thousand GPUs.

    \item \textbf{Optimal group ordering overhead (second half of \S\ref{sec:planner_upper_level_pipeline_orchestration})}: After solving the MINLP problem, there are 16 straggling groups appeared in 5 pipelines (the actual result solved by MINLP was $2+3+3+4+4=16$). Then, our parallelization planning algorithm needs to enumerate the permutations within these 5 pipelines (up to 24 enumerations per pipeline as mentioned in \S\ref{sec:planner_upper_level_pipeline_orchestration} of our manuscript) and solve Eq.~\ref{eq:lower_problem_one_case_layer} for each permutation in order to achieve the best one. Since these 5 pipeline permutations are orthogonal and do not affect each other, we computed up to $24 \times 5 = 120$ ILPs in total, with multi-threading optimization, resulting in an overhead of 0.59s (a single ILP took about 0.40s).

    \item \textbf{Work (layer and data) assignment overhead (\S\ref{sec:planner_lower_level})}: At this point, we only need to solve a few ILP problems in Eq.~\ref{eq:lower_problem_one_case_layer} and Eq.~\ref{eq:lower_problem_one_case_data}, with the ones in Eq.~\ref{eq:lower_problem_one_case_layer} being orthogonal and solved using multi-threading optimization. This part took 0.75s.
\end{itemize}

\begin{table}[!h]
    \centering
    \small
    \caption{\small{Time taken for each part of our algorithm in the 64-GPU S3 scenario and the simulated 1024-GPU scenario.}}
    \begin{tabular}{|c|c|c|c|c|}
        \hline
        & GPU Grouping & Pipeline Division & Group Ordering & Work Assignment \\
        \hline
        1024 GPUs & 0.01s & 51.23s & 0.59s & 0.75s \\
        \hline
        64 GPUs & 0.01s & 22.61s & 0.07s & 0.11s \\
        \hline
    \end{tabular}
    \label{tb:scalability}
\end{table}

Table \ref{tb:scalability} shows that although the overhead increases, \system can still complete planning within a minute. Theoretically speaking, when training with a global batch size of 1024 on 1024 GPUs, each iteration's time is similar to that of training with a global batch size of 64 on 64 GPUs. Therefore, we can complete planning within 1-2 iterations, meaning that we can obtain a new optimal parallelization plan and perform migration within a maximum of 2 iterations. Consequently, our parallelizatio planning algorithm has a sound scalability.

\subsection{Optimal Training Configurations for Megatron-LM and DeepSpeed w/ Restart}
\label{appendix:configuration}

Table~\ref{tb:mega_configuration} and Table~\ref{tb:ds_configuration} illustrate the manually tuned configurations that achieve the best training efficiency under the corresponding straggling situations. In Table~\ref{tb:mega_configuration} and Table~\ref{tb:ds_configuration}, ``DP'' refers to the data parallel degree, ``TP'' refers to the tensor parallel degree, ``PP'' refers to the pipeline parallel degree, ``SP'' refers to Ulysses sequence parallel degree, ``mbs'' refers to the size of micro batches used in training, and ``AC'' indicates the use of activation checkpointing.
Additionally, the blue-highlighted PP indicates cases where the layers could not be evenly divided across the pipeline stages, requiring manual adjustment of the number of layers assigned to the first stage. Similarly, the blue-highlighted DP represents cases where the fixed global batch size (64 in our experiments) could not be evenly divided across the data parallel dimension, necessitating a slight increase in the global batch size.
These complex configurations demonstrate that it requires substantial expert experience and manual efforts if we wish to get rid of stragglers through restarting the training tasks, which is unacceptable for real-world training tasks. 
As a result, the parallelization planning algorithm in \system is of paramount significance for enhancing the stability and robustness of large-scale model training.

\begin{table}[ht]
    \centering
    \caption{\small{Manually tuned optimal configurations for Megatron-LM w/ Restart across different scenarios.}}
    \small
    \begin{tabular}{|c|c|c|c|c|c|}
        \hline
        & Noraml & S1, S2, S6 (Remove 1 Node) & S3, S5 (Remove 2 Nodes) & S4 (Remove 3 Nodes)\\
        \hline
        32B & DP2TP4PP4, mbs1 & DP2TP4PP3, mbs1 & TP4PP4, mbs1 & TP8+AC, mbs1 \\
        \hline
        70B & DP2TP8PP4, mbs1 & DP2TP8\blue{PP7}, mbs1 & DP2TP8\blue{PP6}, mbs1 & DP2TP8\blue{PP5}, mbs1 \\
        \hline
        110B & DP2TP8PP4, mbs1 & DP2TP8\blue{PP7}, mbs1 & DP2TP8\blue{PP6}, mbs1 & DP2TP8\blue{PP5}, mbs1 \\
        \hline
    \end{tabular}
    \label{tb:mega_configuration}
\end{table}

\begin{table}[ht]
    \centering
    \caption{\small{Manually tuned optimal configurations for DeepSpeed w/ Restart across different scenarios.}}
    \small
    \begin{tabular}{|c|c|c|c|c|c|}
        \hline
        & Noraml & S1, S2, S6 (Remove 1 Node) & S3, S5 (Remove 2 Nodes) & S4 (Remove 3 Nodes)\\
        \hline
        32B & DP16SP2+AC, mbs4 & \blue{DP12}SP2+AC, mbs6 & DP8SP2+AC, mbs8 & DP4SP2+AC, mbs4 \\
        \hline
        70B & DP32SP2+AC, mbs2 & \blue{DP14}SP4+AC, mbs2 & \blue{DP12}SP4+AC, mbs2 & \blue{DP20}SP2+AC, mbs2 \\
        \hline
        110B & DP32SP2+AC, mbs2 & \blue{DP14}SP4+AC, mbs2 & \blue{DP12}SP4+AC, mbs2 & \blue{DP20}SP2+AC, mbs2 \\
        \hline
    \end{tabular}
    \label{tb:ds_configuration}
\end{table}

\clearpage

\section{Proofs}
\label{appendix:thm}

In this section, we provide proofs for the theorems and some omitted deduction in our paper.

\setcounter{theorem}{0} % reset the counter

\subsection{Proof for Theorem~\ref{thm:tp_in_one_node}}
\label{appendix:thm-1}

In Section~\ref{sec:planner_upper_level_gpu_grouping}, we propose to partition GPUs with similar performance to the same group according to Theorem~\ref{thm:tp_in_one_node}. Below we provide the proof for it.

\begin{theorem}
\label{thm:thm-1}
Suppose there are $n$ GPUs in a node with straggling rates $\{x_1, \cdots, x_n\}$, and we need to partition them into $n/k$ groups (each with $k$ GPUs). 
Denote $\{i_1, \cdots, i_n\}$ as the ordering satisfying $x_{i_1} \geq \cdots \geq x_{i_n}$. 
Then, the best grouping result that minimizes the running time is 
$\{ 
    \{ x_{i_1}, \cdots, x_{i_{k}} \},
    \{ x_{i_{k+1}}, \cdots, x_{i_{2k}} \}, 
    \cdots,
    \{ x_{i_{n - k + 1}}, \cdots, x_{i_n} \}
\}$.
\end{theorem}

\begin{proof}
\label{proof:proof-1}
The proof of this theorem only leverages the following two simple statements.
\begin{statement}{1}
\label{stat:1}
\textbf{The slowest straggler in the tensor parallel group dominates the time of the whole group.} Assume each GPU group consists of $k$ GPUs, define $T(y_{1}, \ldots, y_{\frac{n}{k}})$ as the minimum cost when the straggling rate of group $i$ is $y_{i}$. Denote $X_{i}$ as the $i$-th GPU with straggling rate $x_{i}$ and $Y_{i}$ as the $i$-th GPU group. Then we have $y_{i} = \max_{X \in Y_{i}}\left\{x\right\}$. Note here we drop the $\rho$ discussed in Section~\ref{sec:planner_lower_level} because all the groups have the same amount of GPUs (equals to $k$).
\end{statement}
\begin{statement}{2}
\label{stat:2}
\textbf{If a single straggler worsens, it is impossible to find a lower minimum cost.} For any given $i$, if $y_{i} \ge y'_{i}$, we have $T(y_{1}, \ldots, y_{\frac{n}{k}}) \ge T(y_{1}, \ldots, y_{i - 1}, y'_{i}, y_{i + 1}, \ldots, y_{\frac{n}{k}})$. Note that no other assumptions are made about the properties of the function $T$.
\end{statement}

For any GPU groups combination $(Y_{1}, \ldots, Y_{\frac{n}{k}})$, where $Y_{i} = (X_{i,1}, \ldots, X_{i,k})$, we guarantee that $\max_{X \in Y_{1}}{x} \ge \ldots \ge \max_{X \in Y_{\frac{n}{k}}}{x}$ and $x_{i,1} \ge \ldots \ge x_{i,k}$ to ensure the representation is free of rotation. Then, we could denote the grouping of GPUs as a single array
\begin{gather*}
G = 
\begin{cases}
    Y_{1} = \left(X_{1,1}, \ldots, X_{1,k}\right) \\
    \ldots \\
    Y_{\frac{n}{k}} = \left(X_{\frac{n}{k},1}, \ldots, X_{\frac{n}{k},k}\right)
\end{cases}
\implies \left(x_{1,1}, \ldots, x_{1,k}, x_{2,1}, \ldots, x_{\frac{n}{k},k}\right)
\end{gather*}
This sequence has a finite number of inverse pairs. We sequentially exchange these inverse pairs until the sequence is in descending order. Next, we will prove that each exchange of the inverse pair does not increase the eventual cost of $T$.

Assume that we are exchanging $X_{i,p}$ with $X_{j,q}$, where $x_{i,p} < x_{j,q}$. Let $Y_{i}$ and $Y_{j}$ be the groups before exchanging, $Y'_{i}$ and $Y'_{j}$ be the new groups after exchanging, we have
\begin{gather*}
x_{i,p} < x_{j,q} \\
\implies i < j \text{ and } p \ne 1 \text{, meaning } X_{i,p} \in Y'_{i} \\ 
\implies y_{i} = \max\limits_{X \in Y_{i}}\left\{x\right\} = x_{i,1} \ge y_{j} = \max\limits_{X \in Y_{j}}\left\{x\right\} = x_{j,1}
\end{gather*}
Two situations may occur during the exchange process. If $q = 1$, we have $x_{i,p} < x_{j,1}$. According to Statement~\ref{stat:2}, the exchange will benefit $T$ because
\begin{gather*}
y'_{i} = \max\limits_{X \in Y'_{i}}\left\{x\right\} = \max\left\{x_{i,1}, x_{j,1}\right\} = x_{i,1} = y_{i} \\
y'_{j} = \max\limits_{X \in Y'_{j}}\left\{x\right\} = \max\left\{x_{j,2}, \ldots, x_{j,k}, x_{i,p}\right\} \le \max\left\{x_{j,2}, \ldots, x_{j,k}, x_{j,1}\right\} = y_{j} \\
\implies T(y_{1}, \ldots, y_{\frac{n}{k}}) \ge T(y_{1}, \ldots, y_{j - 1}, y'_{j}, y_{j + 1}, \ldots, y_{\frac{n}{k}})
\end{gather*}
Otherwise, if $q \ne 1$, meaning that $X_{j,q} \in Y'_{j}$, the exchange will have no effect
\begin{gather*}
y'_{i} = \max\limits_{X \in Y'_{i}}\left\{x\right\} = \max\left\{x_{i,1}, x_{j,q}\right\} = x_{i,1} = y_{i} \\
x_{i,p} < x_{j,q} \le x_{j,1} \implies y'_{j} = \max\limits_{X \in Y'_{j}}\left\{x\right\} = \max\left\{x_{j,1}, x_{i,p}\right\} = x_{j,1} = y_{j} \\
\implies T(y_{1}, \ldots, y_{\frac{n}{k}}) \text{ remains the same}
\end{gather*}
Therefore, for any GPU groups combination $(Y_{1}, \ldots, Y_{\frac{n}{k}})$, after a finite number of inverse pair exchanges, it can be transformed into the grouping result mentioned in Theorem~\ref{thm:tp_in_one_node}, which has no inverse pairs and therefore is the best grouping result among all.
\end{proof}

\subsection{Proof for Theorem~\ref{thm:theoretic_optimality}}
\label{appendix:thm-2}

In Section~\ref{sec:planner_upper_level_gpu_grouping}, we compare two possible grouping results according to Theorem~\ref{thm:theoretic_optimality}. Below we provide the proof for it.

\begin{theorem}
\label{thm:thm-2}
Suppose there are two different grouping results that consist of $M^{\prime}$ and $M^{\prime\prime}$ groups with straggling rates of $\{ y^{\prime}_1, \cdots, y^{\prime}_{M^{\prime}} \}$ and $\{ y^{\prime\prime}_1, \cdots, y^{\prime\prime}_{M^{\prime\prime}} \}$, respectively. 
If we ignore the memory constraints in Eq.~\eqref{eq:lower_problem_one_case_origin}, and further assume the layer and training data assignments are not restricted to integers (i.e., $l_{i,j}, m_i$ in Eq.~\eqref{eq:lower_problem_one_case_origin} can be any positive real numbers), then the minimum training time of the two grouping results satisfy 
${T^{\prime}}/{T^{\prime\prime}} = {(\sum_{i=1}^{M^{\prime\prime}} 1 / y^{\prime\prime}_i)}/{(\sum_{i=1}^{M^{\prime}} 1 / y^{\prime}_i)}$.
\end{theorem}

\begin{proof}
\label{proof:proof-2}
Denote $M$ as the number of total groups, $(y_{1}, \ldots, y_{M})$ as the straggling rate of all groups, $L$ as the number of total layers and $B$ as the number of training samples in one step. And let the running
time of the $j$-th stage in the $i$-th pipeline for one micro-batch
be $t_{i,j} = y_{i,j} \times l_{i,j} \times \tau(b)$, where $b$ is the size of micro-batch. Then, for 1F1B pipeline, considering the heterogeneous layers assignment, the minimum cost of a training iteration for the $i$-th pipeline could be formulated as 
\begin{gather*}
T_{i} = \min\limits_{l_{i,1}, \ldots, l_{i,\ppdi}}\left\{(m_{i} - 1) \times \max\limits_{1 \le j \le \ppdi}\left\{t_{i,j}\right\} + \sum\limits_{1 \le j \le \ppdi}{t_{i,j}}\right\} \\
= \min\limits_{l_{i,1}, \ldots, l_{i,\ppdi}}\left\{(m_{i} - 1) \times \max\limits_{1 \le j \le \ppdi}\left\{y_{i,j} \times l_{i,j}\right\} + \sum\limits_{1 \le j \le \ppdi}{y_{i,j} \times l_{i,j}}\right\} \times \tau(b) \\
\implies T_{i} \approx \min\limits_{l_{i,1}, \ldots, l_{i,\ppdi}}\max\limits_{1 \le j \le \ppdi}\left\{y_{i,j} \times l_{i,j}\right\} \times m_{i} \times \tau(b)  \text{, when $m_{i} >> \ppdi$}
\end{gather*}
Since we have ignored the memory constraints in Eq.~\eqref{eq:lower_problem_one_case_origin} and further assumed the layers assignment is not restricted to integers, given $l_{i,1} + \ldots + l_{i,\ppdi} = L$, the $T_{i}$ could be reached when
\begin{gather*}
\left(l_{i,1}, \ldots, l_{i,\ppdi}\right) = \left(\frac{\frac{L}{y_{i,1}}}{\sum\limits_{1 \le j \le \ppdi}\frac{1}{y_{i,j}}}, \ldots, \frac{\frac{L}{y_{i,\ppdi}}}{\sum\limits_{1 \le j \le \ppdi}\frac{1}{y_{i,j}}}\right) \\
\implies T_{i} = \frac{1}{\sum\limits_{1 \le j \le \ppdi}\frac{1}{y_{i,j}}} \times L \times m_{i} \times \tau(b) 
\end{gather*}
Then, the minimum cost of the whole system given current pipelines training cost $(T_{1}, \ldots, T_{\dpd})$ should be
\begin{gather*}
T = \min\limits_{m_{1} \ldots, m_{\dpd}}\max\limits_{1 \le i \le \dpd}\left\{T_{i}\right\} 
= \min\limits_{m_{1}, \ldots, m_{\dpd}}\max\limits_{1 \le i \le \dpd}\left\{\frac{m_{i}}{\sum\limits_{1 \le j \le \ppdi}\frac{1}{y_{i,j}}} \right\} \times L \times \tau(b)
\end{gather*}
Still, under the assumption that the micro-batches assignment is not restricted to integers and $(m_{1} + \ldots + m_{\dpd}) \times b = B$, we can directly solve the problem
\begin{gather*}
\left(m_{1}, \ldots, m_{\dpd}\right) = \left(\frac{\sum\limits_{1 \le j \le \ppdi}\frac{\frac{B}{b}}{y_{1,j}}}{\sum\limits_{1 \le 1 \le \dpd}\sum\limits_{1 \le j \le \ppdi}\frac{1}{y_{i,j}}}, \ldots, \frac{\sum\limits_{1 \le j \le \ppdi}\frac{\frac{B}{b}}{y_{\dpd,j}}}{\sum\limits_{1 \le 1 \le \dpd}\sum\limits_{1 \le j \le \ppdi}\frac{1}{y_{i,j}}}\right) \\
\implies T = \frac{\frac{B}{b} \times L \times \tau(b)}{\sum\limits_{1 \le 1 \le \dpd}\sum\limits_{1 \le j \le \ppdi}\frac{1}{y_{i,j}}} = \frac{\frac{B}{b} \times L \times \tau(b)}{\sum\limits_{1 \le i \le M}\frac{1}{y_{i}}}
\end{gather*}
We could find out that $T$ is independent of the layers and data assignment. And the ratio of the optimal training cost $T^{\prime}, T^{\prime\prime}$ between two grouping results will always be
\begin{gather*}
\frac{T^{\prime}}{T^{\prime\prime}} = \frac{\sum\limits_{1 \le i \le M^{\prime}}\frac{1}{y^{\prime}_{i}}}{\sum\limits_{1 \le i \le M^{\prime\prime}}\frac{1}{y^{\prime\prime}_{i}}}
\end{gather*}
\end{proof}

\subsection{Proof for Theorem~\ref{thm:group_ordering_in_pipeline}}
\label{appendix:thm-3}

In Section~\ref{sec:planner_upper_level_gpu_grouping}, we propose to sort GPU groups within the same pipeline according to their group straggling rates when the groups have the same number of GPUs, which is based on Theorem~\ref{thm:group_ordering_in_pipeline}. Below we provide the proof for it.

\begin{theorem}
\label{thm:thm-3}
Suppose the groups assigned to the same pipeline have the same number of GPUs, then the best ordering of pipeline stages satisfies that the groups are in descending order w.r.t. the group straggling rates. 
\end{theorem}

\begin{proof}
\label{proof:proof-3}
For a pipeline with $w$ stages, where each stage consists of an equal number of GPUs, let the straggling rate for the groups from the first to the last stage be denoted as $(y_{1}, \ldots, y_{w})$. Assume the optimal layer assignment in this scenario is represented by $(l_{1}, \ldots, l_{w})$. For a single GPU, let $a_{f}$ represent the peak memory consumption of activations during forward propagation for one layer when the micro-batch size is 1. Similarly, let $a_{f+b}$ denote the peak memory consumption of activations during both forward and backward propagation for one layer when the micro-batch size is 1, and let $s$ be the memory consumption of the parameters, gradients, and optimization states for one layer.
Given that the memory constraint is uniform, denoted as $C$, and assuming the memory consumption of non-uniform layers in the first and last stage (such as the embedding table) is negligible compared with a bunch of uniform layers, the constraints in Eq.\eqref{eq:lower_problem_one_case_origin} can be theoretically expressed as (please refer to Proposition~\ref{appendix:deduc_memory} for more details)
\begin{equation*}
\text{$\forall j \in \left[1, w\right]$, } l_{j} \times \left\{b \times \left[a_{f} \times (w - j) + a_{f + b}\right] + s\right\} \le C
\end{equation*}
From this expression, it is evident that the maximum number of layers is constrained by the stage number $j$. Therefore, let $(max\_l_{1}, \ldots, max\_l_{w})$ represent the maximum number of layers for each stage. We have two following properties: 
\begin{enumerate}
    \item $\forall j \in \left[1, w\right], l_{j} \le max\_l_{j}$.
    \item $max\_l_{1} \le \ldots \le max\_l_{w}$.
\end{enumerate}

Assume there exists a group pair satisfying $y_{p} < y_{q}$, where $p < q$. If $l_{p} \le l_{q}$, then by swapping the groups while maintaining the layers, the new solution becomes $(l_{1}, \ldots, l_{p}, \ldots, l_{q}, \ldots, l_{w})$, with groups arrangement $(y_{1}, \ldots, y_{q}, \ldots, y_{p}, \ldots, y_{w})$. We can prove that this will result in a better (or at least equal) training cost due to the objective function in Eq.~\eqref{eq:lower_problem_one_case_origin} when other variables remain constant
\begin{gather*}
y_{p} \times l_{q} < y_{q} \times l_{q} \\
y_{q} \times l_{p} \le y_{q} \times l_{q} \\
\implies max\left\{y_{p} \times l_{q}, y_{q} \times l_{p}\right\} \le max\left\{y_{p} \times l_{q}, y_{q} \times l_{p}\right\}
\end{gather*}
Otherwise, if $l_{p} > l_{q}$, we have
\begin{gather*}
max\_l_{p} \ge l_{p} > l_{q} \\
max\_l_{q} \ge max\_l_{p} \ge l_{p}
\end{gather*}
In this case, by swapping both groups and layers, we obtain another valid solution that still satisfies the memory constraints: $(l_{1}, \ldots, l_{q}, \ldots, l_{p}, \ldots, l_{w})$ under the groups arrangement $(y_{1}, \ldots, y_{q}, \ldots, y_{p}, \ldots, y_{w})$. Similarly, according to the objective function in Eq.~\eqref{eq:lower_problem_one_case_origin}, the new solution is as good as the previous one because $max\left\{y_{p} \times l_{p}, y_{q} \times l_{q}\right\} = max\left\{y_{q} \times l_{q}, y_{p} \times l_{p}\right\}$.

Thus, swapping any pair that satisfies $y_{p} < y_{q}$ and $p < q$ will enhance the overall performance. Consequently, the strategy that arranges the groups in total descending order of the straggling rate will be the most optimal among all configurations.
\end{proof}

\subsection{Memory Cost Model (Deduction of $\mu_{i,j}(b), \nu_{i,j}(b), C_{i,j}$)}
\label{appendix:deduc_memory}

In this section, we provide more details about how to calculate the coefficients of the memory cost model in Section~\ref{sec:planner_lower_level}. 

\begin{proposition}
\label{thm:deduc_memory}
In the hybrid parallel training scenario of large-scale models, when the size of the micro-batch $b$ is given, the memory constraint condition on a single GPU is a linear function that depends only on the number of layer assignments $l_{i,j}$. And the number of GPUs that share the layers in the same group determines the upper limit of the memory constraint condition. 
\end{proposition}

\begin{proof}
\label{proof:deduc_memory}
A large-scale model can usually be divided into many uniform layers (e.g., Transformer blocks) and a handful of non-uniform layers (e.g., the embedding table and LM head). And in the context of a pipeline parallelism, these non-uniform layers will only appear in the first and last stages of the pipeline, while the remaining uniform layers are partitioned into all stages.

For any GPU group $Y$ with $k$ GPUs, on each GPU, let $a_{f_{k}}$ represent the peak memory consumption of activations during forward propagation for one layer when the micro-batch size is 1. Similarly, let $a_{f+b_{k}}$ denote the peak memory consumption of activations during both forward and backward propagation for one layer when the micro-batch size is 1, and let $s_{k}$ be the memory consumption of the parameters, gradients, and optimization states for one layer. Theoretically, the memory consumption is proportional to the number of GPUs in a group, denoted as $k$. This is because the sizes of activations, parameters, gradients and optimization states are all directly related to the hidden states dimension, which is evenly distributed across the GPUs within a group. Consequently, for two GPU groups with $k^{\prime}$ and $k^{\prime\prime}$ GPUs, rates $k^{\prime} / k^{\prime\prime}$ would hold for $a_{f_{k^{\prime}}} / a_{f_{k^{\prime\prime}}}$, $a_{f+b_{k^{\prime}}} / a_{f+b_{k^{\prime\prime}}}$ and $s_{k^{\prime}} / s_{k^{\prime\prime}}$. 

As for the non-uniformed layers, let $\dot{a}_{f_{k}}$ and $\Ddot{a}_{f_{k}}$ represent the peak memory consumption of activations during forward propagation for the first and last several non-uniform layers of the model, respectively, when the micro-batch size is 1. And let $\dot{a}_{f+b_{k}}$ and $\Ddot{a}_{f+b_{k}}$ denote the peak memory consumption of activations during both forward and backward propagation for the first and last several non-uniform layers of the model, respectively, when the micro-batch size is 1. Finally, let $\dot{s}_{k}$ and $\Ddot{s}_{k}$ be the memory consumption of the parameters, gradients, and optimization states for the first and last several non-uniform layers of the model, respectively, when the micro-batch size is 1.

Considering the $i$-th 1F1B pipeline with $\ppdi$ stages in total, for a GPU within stage $j$ ($1 \le j \le \ppdi$), it will first accumulate $\ppdi - j$ rounds forward activations with a micro-batch size of $b$. Then, for the remaining micro-batches, it will execute a forward propagation followed by a backward propagation, until all the micro-batches are processed. Therefore, the peak memory consumption for a GPU within a group consisting of $k_{i,j}$ GPUs in stage $j$ would be
\begin{gather*}
% \begin{cases}
% l_{i,1} \times \left\{b \times \left[a_{f_{k_{i,j}}} \times \left(\ppdi - 1 \right) + a_{f+b_{k_{i,j}}}\right] + s_{k_{i,j}}\right\} + b \times \left[\dot{a}_{f_{k_{i,j}}} \times \left(\ppdi - 1 \right) + \dot{a}_{f+b_{k_{i,j}}}\right] + \dot{s}_{k_{i,j}}, & \text{for } j = 1 \\
% l_{i,\ppdi} \times \left(b \times a_{f+b_{k_{i,j}}} + s_{k_{i,j}}\right) + b \times \Ddot{a}_{f+b_{k_{i,j}}} + \Ddot{s}_{k_{i,j}}, & \text{for } j = \ppdi \\
% l_{i,j} \times \left\{b \times \left[a_{f_{k_{i,j}}} \times \left(\ppdi - j \right) + a_{f+b_{k_{i,j}}}\right] + s_{k_{i,j}}\right\}, & \text{for } 2 \le j \le \ppdi - 1
% \end{cases}
\begin{aligned}
& l_{i,1} \times \left\{b \times \left[a_{f_{k_{i,j}}} \times \left(\ppdi - 1 \right) + a_{f+b_{k_{i,j}}}\right] + s_{k_{i,j}}\right\} \\&\;\;\;\;\;\;\;\;\;\;\;\; + b \times \left[\dot{a}_{f_{k_{i,j}}} \times \left(\ppdi - 1 \right) + \dot{a}_{f+b_{k_{i,j}}}\right] + \dot{s}_{k_{i,j}}, && \text{for } j = 1 \\
& l_{i,\ppdi} \times \left(b \times a_{f+b_{k_{i,j}}} + s_{k_{i,j}}\right) + b \times \Ddot{a}_{f+b_{k_{i,j}}} + \Ddot{s}_{k_{i,j}}, && \text{for } j = \ppdi \\
&l_{i,j} \times \left\{b \times \left[a_{f_{k_{i,j}}} \times \left(\ppdi - j \right) + a_{f+b_{k_{i,j}}}\right] + s_{k_{i,j}}\right\}, && \text{for } 2 \le j \le \ppdi - 1
\end{aligned}
\end{gather*}
Assume that the memory limit for GPU $X$ is denoted by $C_{X}$. In most cases, $C_{X} = C$, where $C$ represents the GPU memory bound. However, if a GPU straggler experiences memory pressure, the GPU memory usage within the group should be limited by the minimum $C_{X}$. To accommodate practical scenarios and prevent out-of-memory (OOM) errors, we also introduce a reserved memory gap $G$ (4096MiB in our experimental setup) to allocate memory for essential operations such as NCCL and CUDA contexts. Thus, From the perspective of $k = 1$, the memory constraint can be formulated as
\begin{gather*}
% \begin{cases}
% l_{i,1} \times \left\{b \times \left[\frac{a_{f_{1}}}{k_{i,j}} \times \left(\ppdi - 1 \right) + \frac{a_{f+b_{1}}}{k_{i,j}}\right] + \frac{s_{1}}{k_{i,j}}\right\} + b \times \left[\frac{\dot{a}_{f_{1}}}{k_{i,j}} \times \left(\ppdi - 1 \right) + \frac{\dot{a}_{f+b_{1}}}{k_{i,j}}\right] + \frac{\dot{s}_{1}}{k_{i,j}} \le \min\limits_{X \in Y_{i, j}}\left\{C_{X}\right\} - G, & \text{for } j = 1 \\
% l_{i,\ppdi} \times \left(b \times \frac{a_{f+b_{1}}}{k_{i,j}} + \frac{s_{1}}{k_{i,j}}\right) + b \times \frac{\Ddot{a}_{f+b_{1}}}{k_{i,j}} + \frac{\Ddot{s}_{1}}{k_{i,j}} \le \min\limits_{X \in Y_{i, j}}\left\{C_{X}\right\} - G, & \text{for } j = \ppdi \\
% l_{i,j} \times \left\{b \times \left[\frac{a_{f_{1}}}{k_{i,j}} \times \left(\ppdi - j \right) + \frac{a_{f+b_{1}}}{k_{i,j}}\right] + \frac{s_{1}}{k_{i,j}}\right\} \le \min\limits_{X \in Y_{i, j}}\left\{C_{X}\right\} - G, & \text{for } 2 \le j \le \ppdi - 1
% \end{cases}
\begin{aligned}
&l_{i,1} \times \left\{b \times \left[\frac{a_{f_{1}}}{k_{i,j}} \times \left(\ppdi - 1 \right) + \frac{a_{f+b_{1}}}{k_{i,j}}\right] + \frac{s_{1}}{k_{i,j}}\right\} \\&\;\;\;\;\;\;\;\;\;\;\;\; + b \times \left[\frac{\dot{a}_{f_{1}}}{k_{i,j}} \times \left(\ppdi - 1 \right) + \frac{\dot{a}_{f+b_{1}}}{k_{i,j}}\right] + \frac{\dot{s}_{1}}{k_{i,j}} \le \min\limits_{X \in Y_{i, j}}\left\{C_{X}\right\} - G, && \text{for } j = 1 \\
&l_{i,\ppdi} \times \left(b \times \frac{a_{f+b_{1}}}{k_{i,j}} + \frac{s_{1}}{k_{i,j}}\right) + b \times \frac{\Ddot{a}_{f+b_{1}}}{k_{i,j}} + \frac{\Ddot{s}_{1}}{k_{i,j}} \le \min\limits_{X \in Y_{i, j}}\left\{C_{X}\right\} - G, && \text{for } j = \ppdi \\
&l_{i,j} \times \left\{b \times \left[\frac{a_{f_{1}}}{k_{i,j}} \times \left(\ppdi - j \right) + \frac{a_{f+b_{1}}}{k_{i,j}}\right] + \frac{s_{1}}{k_{i,j}}\right\} \le \min\limits_{X \in Y_{i, j}}\left\{C_{X}\right\} - G, && \text{for } 2 \le j \le \ppdi - 1
\end{aligned}
\end{gather*}
We can directly derive $\mu_{i,j}(b)$, $\nu_{i,j}(b)$ and $C_{i, j}$ from the above equation as 
\begin{align*}
\begin{cases}
\mu_{i,1}(b) = b \times \left[a_{f_{1}} \times \left(\ppdi - 1 \right) + a_{f+b_{1}}\right] + s_{1}, & \text{for } j = 1 \\
\mu_{i,\ppdi}(b) = b \times a_{f+b_{1}} + s_{1}, & \text{for } j = \ppdi \\
\mu_{i,j}(b) = b \times \left[a_{f_{1}} \times \left(\ppdi - j \right) + a_{f+b_{1}}\right] + s_{1}, & \text{for } 2 \le j \le \ppdi - 1
\end{cases}
\\
\begin{cases}
\nu_{i,1}(b) = b \times \left[\dot{a}_{f_{1}} \times \left(\ppdi - 1 \right) + \dot{a}_{f+b_{1}}\right] + \dot{s}_{1}, & \text{for } j = 1 \\
\nu_{i,\ppdi}(b) = b \times \Ddot{a}_{f+b_{1}} + \Ddot{s}_{1}, & \text{for } j = \ppdi \\
\nu_{i,j}(b) = 0, & \text{for } 2 \le j \le \ppdi - 1
\end{cases}
\end{align*}
and
\begin{gather*}
C_{i,j} = k_{i,j} \times \left(\min\limits_{X \in Y_{i, j}}\left\{C_{X}\right\} - G\right), \text{where } k_{i,j} = \left| \left\{X | X \in Y_{i,j}\right\} \right|
\end{gather*}
\end{proof}

\subsection{Deduction of Equivalence Between Eq.~\eqref{eq:lower_problem_one_case_origin} and Eq.~\eqref{eq:lower_problem_one_case_layer}, ~\eqref{eq:lower_problem_one_case_data}}
\label{appendix:deduce_eq-1-2-3}

In Section~\ref{sec:planner_lower_level}, we decouple the problem in Eq.~\eqref{eq:lower_problem_one_case_origin} into the sub-problems in Eq.~\eqref{eq:lower_problem_one_case_layer} and Eq.~\eqref{eq:lower_problem_one_case_data}. Below we provide the detailed deduction. 

\begin{proposition}
\label{thm:deduce_eq-1-2-3}
Any optimal solution of Eq.~\eqref{eq:lower_problem_one_case_layer} combined with Eq.~\eqref{eq:lower_problem_one_case_data} will also be one of the optimal solutions of Eq.~\eqref{eq:lower_problem_one_case_origin}
\end{proposition}.

\begin{proof}
\label{proof:deduce_eq-1-2-3}
Let $(\overline{l}_{1,1}, \ldots, \overline{l}_{\dpd,\ppd_{\dpd}})$ and $(\overline{m}_{1}, \ldots, \overline{m}_{\dpd})$ be one of the optimal solutions of Eq.~\eqref{eq:lower_problem_one_case_origin}. And let $\overline{o}_{i} := \max_{1 \le j \le \ppdi}\{y_{i,j} \times \overline{l}_{i,j} \times \tau(b) \times \overline{m}_{i}\}$. Let $(\hat{l}_{i,1}, \ldots, \hat{l}_{i,\ppdi})$ be one of the optimal solutions of Eq.~\eqref{eq:lower_problem_one_case_layer}, the $i$-th sub-problem, where $i$ is fixed here. Assume $(\hat{l}_{i,1}, \ldots, \hat{l}_{i,\ppdi}) \ne (\overline{l}_{i,1}, \ldots, \overline{l}_{i,\ppdi})$, then we consider $\hat{o}_{i} := \max_{1 \le j \le \ppdi}\{y_{i,j} \times \hat{l}_{i,j} \times \tau(b) \times \overline{m}_{i}\}$. Since $(\hat{l}_{i,1}, \ldots, \hat{l}_{i,\ppdi})$ is one of the optimal solutions of Eq.\eqref{eq:lower_problem_one_case_layer}, and $(\overline{l}_{i,1}, \ldots, \overline{l}_{i,\ppdi})$ is also a feasible solution for it, we have $\max_{1 \le j \le \ppdi}\{y_{i,j} \times \hat{l}_{i,j}\} \le \max_{1 \le j \le \ppdi}\{y_{i,j} \times \overline{l}_{i,j}\}$. Therefore, it must be the case that
\begin{gather*}
\hat{o}_{i} = \max\limits_{1 \le j \le \ppdi}\left\{y_{i,j} \times \hat{l}_{i,j} \times \tau(b) \times \overline{m}_{i}\right\} \le \overline{o}_{i} = \max\limits_{1 \le j \le \ppdi}\left\{y_{i,j} \times \overline{l}_{i,j} \times \tau(b) \times \overline{m}_{i}\right\}
\end{gather*}
This implies that the solution $(\overline{l}_{i,1}, \ldots, \overline{l}_{i,\ppd_{1}}, \ldots, \hat{l}_{i,1}, \ldots, \hat{l}_{i,\ppd_{i}}, \ldots, \overline{l}_{\dpd,1}, \ldots, \overline{l}_{\dpd,\ppd_{\dpd}})$ is a better (or at least equal) solution for Eq.\eqref{eq:lower_problem_one_case_origin} compared to $(\overline{l}_{i,1}, \ldots, \overline{l}_{i,\ppd_{1}}, \ldots, \overline{l}_{i,1}, \ldots, \overline{l}_{i,\ppd_{i}}, \ldots, \overline{l}_{\dpd,1}, \ldots, \overline{l}_{\dpd,\ppd_{\dpd}})$. Since we have already assumed that one of the optimal solutions of Eq.\eqref{eq:lower_problem_one_case_origin} is $(\overline{l}_{i,1},\allowbreak \ldots,\allowbreak \overline{l}_{i,\ppd_{1}},\allowbreak \ldots,\allowbreak \overline{l}_{i,1},\allowbreak \ldots,\allowbreak \overline{l}_{i,\ppd_{i}},\allowbreak \ldots,\allowbreak \overline{l}_{\dpd,1},\allowbreak \ldots,\allowbreak \overline{l}_{\dpd,\ppd_{\dpd}})$, it follows that $(\overline{l}_{i,1}, \ldots, \overline{l}_{i,\ppd_{1}}, \ldots, \hat{l}_{i,1}, \ldots, \hat{l}_{i,\ppd_{i}},$
$ \ldots, \overline{l}_{\dpd,1}, \ldots, \overline{l}_{\dpd,\ppd_{\dpd}})$ is equally optimal. And the combined solution $(\hat{l}_{i,1}, \ldots, \hat{l}_{i,\ppd_{1}}, \ldots, \hat{l}_{i,1}, \ldots, \hat{l}_{i,\ppd_{i}}, \ldots, \hat{l}_{\dpd,1}, \ldots, \hat{l}_{\dpd,\ppd_{\dpd}})$ after solving Eq.\eqref{eq:lower_problem_one_case_layer} for each given $i$, should also be one of the optimal solutions of Eq.\eqref{eq:lower_problem_one_case_origin}. Therefore, we can solve Eq.\eqref{eq:lower_problem_one_case_origin} by solving all Eq.\eqref{eq:lower_problem_one_case_layer} (totally $\dpd$ sub-problems) and obtain the optimal solution $(\hat{l}_{i,1}, \ldots, \hat{l}_{i,\ppd_{1}}, \ldots, \hat{l}_{i,1}, \ldots, \hat{l}_{i,\ppd_{i}}, \ldots, \hat{l}_{\dpd,1}, \ldots, \hat{l}_{\dpd,\ppd_{\dpd}})$. Subsequently, with $l_{i,j}$ fixed, the only variables are $m_{i}$. The problem then reduces to a single ILP problem (Eq.\eqref{eq:lower_problem_one_case_data}) represented as
\begin{gather*}
\argmin_{m_i} \max_{i \in \left[1, \dpd\right]} \left\{ \max\limits_{1 \le j \le \ppd_{i}}\left\{y_{i,j} \times \hat{l}_{i,j}\right\} \times m_i \right\} \times \tau(b) \\
\text{ s.t. } \sum_{i \in \left[1, \dpd\right]} m_i \times b = B, \\
m_i \in \mathbb{N}_0 \text{ for } \forall i \in \left[1, \dpd\right]
\end{gather*}
\end{proof}

\subsection{Deduction of Eq.~\eqref{eq:pipeline_division_estimation}}
\label{appendix:deduce_eq-4}

In Section~\ref{sec:planner_upper_level_pipeline_orchestration}, we formulate the problem of finding the best pipeline division into Eq.~\eqref{eq:pipeline_division_estimation} by relaxing a few constraints in Eq.~\eqref{eq:lower_problem_one_case_origin}. Below we provide the detailed deduction.

\begin{proposition}
\label{thm:deduce_eq-4}
If we ignore the memory constraints in Eq.~\eqref{eq:lower_problem_one_case_origin}, and further assume the layer assignments are not restricted to integers (i.e., $l_{i,j}$ in Eq.~\eqref{eq:lower_problem_one_case_origin} can be any positive real numbers), then the problem of finding the best pipeline division can be written as Eq.~\eqref{eq:pipeline_division_estimation}.
\end{proposition}

\begin{proof}
\label{proof:deduce_eq-4}
We have already proved the equivalence between the ultimate problem, Eq.~\eqref{eq:lower_problem_one_case_origin} and two sub-problems, Eq.~\eqref{eq:lower_problem_one_case_layer} and Eq.~\eqref{eq:lower_problem_one_case_data}. Thus, when first solving Eq.~\eqref{eq:lower_problem_one_case_layer}, if we neglect the memory constraints and further assume the layer assignments are not restricted to integers, the optimal solution could be directly obtained as follows
\begin{gather*}
\left(l_{i,1}, \ldots, l_{i,\ppdi}\right) = \left(\frac{\frac{L}{y_{i,1}}}{\sum\limits_{1 \le j \le \ppdi}\frac{1}{y_{i,j}}}, \ldots, \frac{\frac{L}{y_{i,\ppdi}}}{\sum\limits_{1 \le j \le \ppdi}\frac{1}{y_{i,j}}}\right) 
%\\
\implies o_{i} = \frac{1}{\sum\limits_{1 \le j \le \ppdi}\frac{1}{y_{i,j}}} \times L 
\end{gather*}
Substituting these results into  Eq.~\eqref{eq:lower_problem_one_case_data}, the problem becomes
\begin{gather*}
\argmin_{m_i}
\max_{i \in \left[1, \dpd\right]} \left\{ \frac{m_i \times \tau(b)}{\sum\limits_{1 \le j \le \ppdi}\frac{1}{y_{i,j}}} \right\} \\
\text{ s.t. } \sum_{i \in \left[1, \dpd\right]} m_i \times b = B \\
m_i \in \mathbb{N}_0 \text{ for } \forall i \in \left[1, \dpd\right]
\end{gather*}
Assuming there are $M$ GPU groups in total, with $M_{s}$ groups having stragglers. For the $i$-th pipeline, let $h_{i}$ be the number of normal groups (without stragglers). Considering the placement of each straggler group, we can form a $\{q_{i,k}\}_{\dpd, M{s}}$ matrix. In this matrix, each column $j$ contains exactly one 1, representing which pipeline the $j$-th straggler group is assigned. Each row consists of several 1s, indicating how many straggler groups the pipeline contains, summing up to $h_{i}$. Thus, for the $i$-th pipeline, its term in the objective function can be transformed to
\begin{gather*}
\frac{m_i \times \tau(b)}{\sum\limits_{1 \le j \le \ppdi}\frac{1}{y_{i,j}}} = \frac{m_i \times \tau(b)}{h_i \times \hat{y} + \sum_{k=1}^{M_{s}} {q_{i,k}}/{y_{k}}}
\end{gather*}
We also have additional constraints for $h_{i}$ and $\{q_{i,k}\}_{\dpd, M_{s}}$
\begin{gather*}
\sum_{i=1}^{\dpd} h_i = M - M_{s} \\
q_{i,k} \in \{0, 1\} \text{ for } \forall i \in [1, \dpd], \forall k \in [1, M_{s}] \\
\sum_{i=1}^{\dpd} q_{i,k} = 1 \text{ for } \forall k \in [1, M_{s}] \\
h_i \in \mathbb{N}_0 \text{ for } \forall i \in [1, \dpd]
\end{gather*}
\end{proof}

\subsection{Deduction of Possible Grouping Results after Splitting}
\label{appendix:deduce_grouping-possibilities}

In Section~\ref{sec:planner_upper_level_gpu_grouping}, we make the statement that there exist up to 6 possible grouping results to classify 7 GPUs into three groups with 1, 2, and 4 GPUs, respectively. Below we first prove a related proposition, and then elaborate the 6 possible grouping results.

\begin{proposition}
\label{thm:deduce_grouping-possibilities}
For an optimal grouping solution, regardless of the number of GPUs each group possesses (non-uniform device partitioning may happens), the GPUs within each group must be arranged in descending order of their straggling rates and form a consecutive sequence.
\end{proposition}

\begin{proof}
\label{proof:deduce_grouping-possibilities}
The underlying principle is that we can demonstrate the existence of an optimal solution by means similar to the proof in Theorem \ref{thm:tp_in_one_node}, that is, through the exchange of GPU pairs. Suppose there exists a non-consecutive group $Y = ({X_{i_{1}}, \ldots, X_{i_{k}}})$, and assuming the corresponding straggling rates satisfy $x_{i_{1}} \ge \ldots \ge x_{i_{k}}$. Then we consider the GPU with the highest straggling rate in the group, $X_{i_{1}}$. Assuming it is in the $p$-th position in the GPU sequence arranged in descending order of the straggling rate. We then consider $(x_{p}, x_{p-1}, \ldots, x_{p-k+1})$. Because the non-consecutive group already has $k-1$ GPUs with a smaller straggling rate than $X_{p}$, it is certain that such a consecutive sequence can be extracted without the issue of insufficient GPU amount. And for the new continuous GPU group with descending order $Y^{\prime} = (X_{p}, X_{p-1}, \ldots, X_{p-k+1})$, it necessarily possesses the property that
\begin{gather*}
x_{p} = x_{i_1} \\
x_{p-1} \ge x_{i_2} \\
x_{p-2} \ge x_{i_3} \\
\ldots
\\
x_{p-k+1} \ge x_{i_k} \\
x_{p} \ge x_{p-1} \ge \ldots \ge x_{p-k+1}
\end{gather*}
Then, we swap each $X_{i_j}$ with $X_{p - (j - 1)}$. For the group who owns $X_{p - (j - 1)}$ before, it gets a better GPU with smaller straggling rate, and therefore would have a less (or at least equal) group straggling rate in total. And for the newly formed group $Y^{\prime}$, since the $X_{p}$ is unchanged, it shares the same group straggling rate with $Y$ before. Thus, the total training cost will only get better after swapping to turn $Y$ into $Y^{\prime}$. And a grouping result that only consists of consecutive GPUs in each group would be better compared to all.
\end{proof}

Then, we consider the circumstance that an original group consisting of 8 GPUs is split due to a heavy straggler in it, and we need to classify the 7 remaining GPUs into three groups with 1, 2, and 4 GPUs, respectively. 
Note that the following enumeration of grouping is generalized and applicable to various partitioning scenarios and straggler situations. 

Without loss of generality, assuming that the remaining GPUs are sorted in descending order by straggling rate $x_{1} \ge \ldots \ge x_{7}$, corresponding sequentially to GPU $X_1, \ldots, X_7$. Let $Y_1$, $Y_2$ and $Y_3$ represent the groups with 1, 2 and 4 GPUs, respectively. 
According to Proposition~\ref{thm:deduce_grouping-possibilities}, if $X_1$ is partitioned to $Y_3$, then we must partition $X_2, X_3, X_4$ to $Y_3$ to achieve better performance. Similarly, if $X_1$ is partitioned to $Y_2$, then we must partition $X_2$ to $Y_2$. Consequently, we only need to consider the following 6 grouping results: 
\begin{align*}
\begin{cases}
Y_1 = \left(X_1\right) \\
Y_2 = \left(X_2, X_3\right) \\
Y_3 = \left(X_4, X_5, X_6, X_7\right) \\
\end{cases}
\begin{cases}
Y_1 = \left(X_1\right) \\
Y_2 = \left(X_6, X_7\right) \\
Y_3 = \left(X_2, X_3, X_4, X_5\right) \\
\end{cases}
\\
\begin{cases}
Y_1 = \left(X_3\right) \\
Y_2 = \left(X_1, X_2\right) \\
Y_3 = \left(X_4, X_5, X_6, X_7\right) \\
\end{cases}
\begin{cases}
Y_1 = \left(X_5\right) \\
Y_2 = \left(X_6, X_7\right) \\
Y_3 = \left(X_1, X_2, X_3, X_4\right) \\
\end{cases}
\\
\begin{cases}
Y_1 = \left(X_7\right) \\
Y_2 = \left(X_1, X_2\right) \\
Y_3 = \left(X_3, X_4, X_5, X_6\right) \\
\end{cases}
\begin{cases}
Y_1 = \left(X_7\right) \\
Y_2 = \left(X_5, X_6\right) \\
Y_3 = \left(X_1, X_2, X_3, X_4\right) \\
\end{cases}
\end{align*}

We will subsequently apply Theorem~\ref{thm:thm-2} to identify the most favorable among these candidates. 

To further substantiate the efficacy of Theorem~\ref{thm:thm-2} in selecting the optimal one of all possibilities, we evaluate the end-to-end performance and the estimated performance provided by Theorem~\ref{thm:thm-2} on the 110B model. We introduce three stragglers with straggling rates of 
$2.57$, $5.42$ and $12.53$ within a single node. Then we assess the performance of \system by crafting parallelization plans with the three different grouping results illustrated in Figure~\ref{fig:split_and_group}. From Figure~\ref{fig:theorem_2}, we can find that the correlation between the training time estimated using Theorem~\ref{thm:thm-2} and the actual training time of \system is coherent, indicating that the lower the estimated time, the lower the actual training time of \system. Consequently, we can deduce that Theorem~\ref{thm:thm-2} efficiently guides us in identifying the optimal grouping result without the necessity to evaluate all potentialities by solving the pipeline division problem as well as the lower-level model. Instead, the overhead incurred by employing Theorem~\ref{thm:thm-2} to filter the best grouping result is negligible. 
% Moreover, the overhead incurred by employing Theorem~\ref{thm:thm-2} to estimate a select few strategies, filtered through Proposition~\ref{thm:deduce_grouping-possibilities}, is negligible.

\begin{figure*}[!h]
\centering
\includegraphics[width=0.8\textwidth]{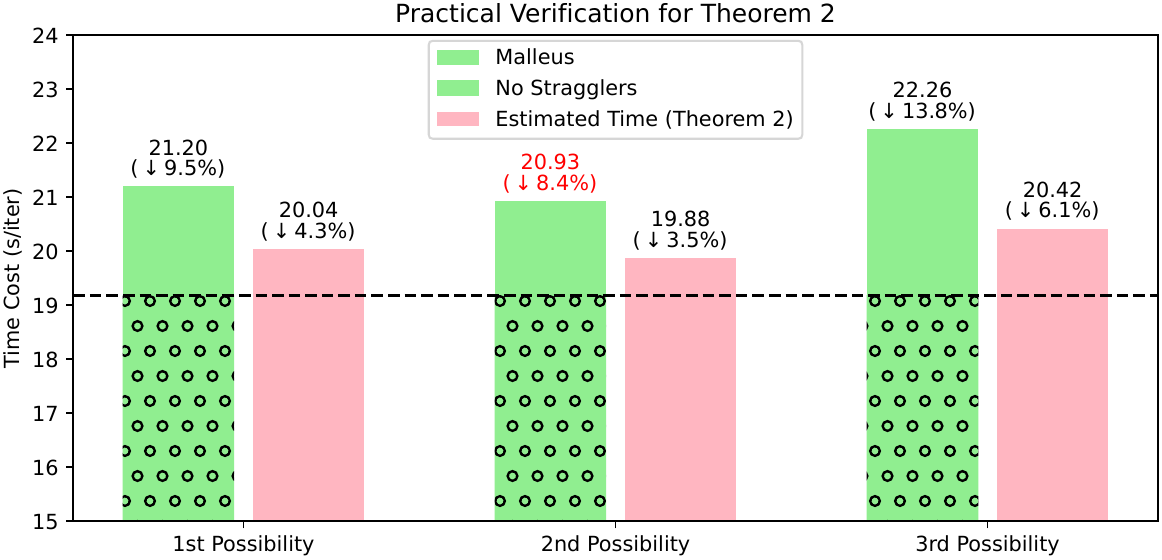}
\caption{\small{Effectiveness of Theorem~\ref{thm:thm-2}, evaluated on the 110B model. Three scenarios correspond to three different grouping possibilities after splitting depicted in Figure~\ref{fig:split_and_group}.}}
\label{fig:theorem_2}
\end{figure*}

\end{document}